\documentclass[10pt,journal,compsoc]{IEEEtran}

\usepackage{tikz}
\usepackage{amssymb}
\usepackage{amsthm}
\usepackage{amsmath}
\usepackage{mathtools}
\usepackage{balance}
\usepackage{subcaption}
\usepackage{pifont}
\usepackage{enumitem}
\usepackage{booktabs}
\usepackage{multirow}
\usepackage{lipsum}
\usepackage{pgfplots}
\usepackage{graphicx}
\usepackage{epstopdf}
\usepackage{array}
\usepackage{comment}
\usepackage{multirow}
\usepackage{makecell}
\usepackage{cite}
\usepackage{titlesec}
\usetikzlibrary{shapes}
\usepackage{float}
\usepackage{dblfloatfix}
\usepackage[bookmarks=false]{hyperref}
\hypersetup{colorlinks=true,linkcolor=black,citecolor=blue}
\usepackage[capitalise]{cleveref}
\usepackage{glossaries}
\usepackage{verbatim}
\usepackage{textcomp}
\usepackage[table]{xcolor}
\usepackage{tabstackengine}
\usepackage{booktabs,array}
\usepackage{stackengine}
\usepackage{adjustbox}
\usepackage[absolute,showboxes]{textpos}

\usepackage{mathrsfs}

\newcommand{\bcmark}{\ding{52}}%
\newcommand{\xmark}{\ding{55}}%
\newcommand\mytextsf{\bfseries\sffamily\fontsize{9pt}{9pt}\selectfont}

\newcommand*\circled[1]{\tikz[baseline=(char.base)]{
            \node[shape=circle,draw,inner sep=2pt] (char) {#1};}}

\newtheorem{theorem}{Theorem}
\newtheorem{proposition}{Proposition}

\newtheorem{lemma}[theorem]{Lemma}

\newtheoremstyle{myremark}
  {\topsep}
  {\topsep}
  {\itshape}
  {0pt}
  {\scshape}
  {.}
  { }
  {}

\theoremstyle{myremark}
\newtheorem{remark}{Remark}

\newtheoremstyle{myexample}
  {\topsep}
  {\topsep}
  {\itshape}
  {0pt}
  {\scshape}
  {.}
  { }
  {}

\theoremstyle{myexample}
\newtheorem*{example}{Example}

\newcommand{\doubletilde}[1]{\tilde{\raisebox{0pt}[0.85\height]{$\tilde{#1}$}}}

\makeatletter
\let\save@mathaccent\mathaccent
\newcommand*\if@single[3]{%
  \setbox0\hbox{${\mathaccent"0362{#1}}^H$}%
  \setbox2\hbox{${\mathaccent"0362{\kern0pt#1}}^H$}%
  \ifdim\ht0=\ht2 #3\else #2\fi
  }
\newcommand*\rel@kern[1]{\kern#1\dimexpr\macc@kerna}
\newcommand*\wideaccent[2]{\@ifnextchar^{{\wide@accent{#1}{#2}{0}}}{\wide@accent{#1}{#2}{1}}}
\newcommand*\wide@accent[3]{\if@single{#2}{\wide@accent@{#1}{#2}{#3}{1}}{\wide@accent@{#1}{#2}{#3}{2}}}
\newcommand*\wide@accent@[4]{%
  \begingroup
  \def\mathaccent##1##2{%
    \let\mathaccent\save@mathaccent
    \if#42 \let\macc@nucleus\first@char \fi
    \setbox\z@\hbox{$\macc@style{\macc@nucleus}_{}$}%
    \setbox\tw@\hbox{$\macc@style{\macc@nucleus}{}_{}$}%
    \dimen@\wd\tw@
    \advance\dimen@-\wd\z@
    \divide\dimen@ 3
    \@tempdima\wd\tw@
    \advance\@tempdima-\scriptspace
    \divide\@tempdima 10
    \advance\dimen@-\@tempdima
    \ifdim\dimen@>\z@ \dimen@0pt\fi
    \rel@kern{0.6}\kern-\dimen@
    \if#41
      #1{\rel@kern{-0.6}\kern\dimen@\macc@nucleus\rel@kern{0.4}\kern\dimen@}%
      \advance\dimen@0.4\dimexpr\macc@kerna
      \let\final@kern#3%
      \ifdim\dimen@<\z@ \let\final@kern1\fi
      \if\final@kern1 \kern-\dimen@\fi
    \else
      #1{\rel@kern{-0.6}\kern\dimen@#2}%
    \fi
  }%
  \macc@depth\@ne
  \let\math@bgroup\@empty \let\math@egroup\macc@set@skewchar
  \mathsurround\z@ \frozen@everymath{\mathgroup\macc@group\relax}%
  \macc@set@skewchar\relax
  \let\mathaccentV\macc@nested@a
  \if#41
    \macc@nested@a\relax111{#2}%
  \else
    \def\gobble@till@marker##1\endmarker{}%
    \futurelet\first@char\gobble@till@marker#2\endmarker
    \ifcat\noexpand\first@char A\else
      \def\first@char{}%
    \fi
    \macc@nested@a\relax111{\first@char}%
  \fi
  \endgroup
}
\makeatother

\newcommand\widebar{\wideaccent\overline}

\makeatletter
\newcommand{\doublewidetilde}[1]{{%
  \mathpalette\double@widetilde{#1}%
}}
\newcommand{\double@widetilde}[2]{%
  \sbox\z@{$\m@th#1\widetilde{#2}$}%
  \ht\z@=.9\ht\z@
  \widetilde{\box\z@}%
}
\makeatother

\newcommand\scaleddot{\scalebox{.89}{.}}

\makeatletter
\renewcommand{\dddot}[1]{%
  {\mathop{\kern\z@#1}\limits^{\makebox[0pt][c]{\vbox to-2.2\ex@{\kern-\tw@\ex@
   \hbox{\normalfont\scaleddot\kern-0.5pt\scaleddot\kern-0.5pt\scaleddot}\vss}}}}}
\renewcommand{\ddddot}[1]{%
  {\mathop{\kern\z@#1}\limits^{\makebox[0pt][c]{\vbox to-2.2\ex@{\kern-\tw@\ex@
   \hbox{\normalfont\scaleddot\kern-0.5pt\scaleddot\kern-0.5pt\scaleddot\kern-0.5pt\scaleddot}\vss}}}}}
\makeatother

\makeglossaries

\newacronym{ISAC}{ISAC}{integrated sensing and communications}
\newacronym{BS}{BS}{base station}
\newacronym{RF}{RF}{radio-frequency}
\newacronym{DAC}{DAC}{digital-to-analog converter}
\newacronym{IRS}{IRS}{intelligent reflecting surface}
\newacronym{PAPR}{PAPR}{peak-to-average power ratio}
\newacronym{CSI}{CSI}{channel state information}

\newacronym{COTS}{COTS}{commercial-off-the-shelf}

\newacronym{RRM}{RRM}{radio resource management}
 
\newacronym{AWGN}{AWGN}{additive white Gaussian noise}
\newacronym{SNR}{SNR}{signal-to-noise ratio}
\newacronym{DPG}{DPG}{directional power gain}
\newacronym{DR}{DR}{data rate}
\newacronym{SINR}{SINR}{signal-to-interference-plus-noise ratio}
\newacronym{SDR}{SDR}{semidefinite relaxation}
\newacronym{SDP}{SDP}{semidefinite programming}
\newacronym{SCA}{SCA}{successive convex approximation}
 
\newacronym{MILP}{MILP}{mixed-integer linear program} 
\newacronym{MINLP}{MINLP}{mixed-integer nonlinear program} 
\newacronym{MISDP}{MISDP}{mixed-integer semidefinite program} 

\newacronym{AOA}{AOA}{angle of arrival}
\newacronym{AOD}{AOD}{angle of departure}
\newacronym{RC}{RC}{reflection coefficient}
\newacronym{RSI}{RSI}{residual self-interference}

\newacronym{LoS}{LoS}{line-of-sight}
\newacronym{NLoS}{NLoS}{non-LoS}

\newacronym{LHS}{LHS}{left-hand-side}
\newacronym{RHS}{RHS}{right-hand-side}

\newacronym{ES}{ES}{exhaustive search}
\newacronym{BnC}{BnC}{branch-and-cut}
\definecolor{blue1}{HTML}{b3cde0}
\definecolor{blue2}{HTML}{6497b1}
\definecolor{blue3}{HTML}{005b96}
\definecolor{blue4}{HTML}{03396c}
\definecolor{blue5}{HTML}{011f4b}

\definecolor{dcolor1}{HTML}{253494}
\definecolor{dcolor2}{HTML}{636363}
\definecolor{dcolor3}{HTML}{fcbba1}
\definecolor{dcolor4}{HTML}{fb6a4a}
\definecolor{dcolor5}{HTML}{cb181d}
\definecolor{dcolor6}{HTML}{67000d}
\definecolor{dcolor7}{HTML}{9ecae1}
\definecolor{dcolor8}{HTML}{4292c6}
\definecolor{dcolor9}{HTML}{08519c}

\definecolor{layer1}{HTML}{003f5c}
\definecolor{layer2}{HTML}{444e86}
\definecolor{layer3}{HTML}{955196}
\definecolor{layer4}{HTML}{dd5182}
\definecolor{layer5}{HTML}{ff6e54}
\definecolor{layer6}{HTML}{ffa600}

\definecolorset{rgb}{}{}{%
AliceBlue,.94,.972,1;%
AntiqueWhite,.98,.92,.844;%
Aqua,0,1,1;%
Aquamarine,.498,1,.83;%
Azure,.94,1,1;%
Beige,.96,.96,.864;%
Bisque,1,.894,.77;%
Black,0,0,0;%
BlanchedAlmond,1,.92,.804;%
Blue,0,0,1;%
BlueViolet,.54,.17,.888;%
Brown,.648,.165,.165;%
BurlyWood,.87,.72,.53;%
CadetBlue,.372,.62,.628;%
Chartreuse,.498,1,0;%
Chocolate,.824,.41,.116;%
Coral,1,.498,.312;%
CornflowerBlue,.392,.585,.93;%
Cornsilk,1,.972,.864;%
Crimson,.864,.08,.235;%
Cyan,0,1,1;%
DarkBlue,0,0,.545;%
DarkCyan,0,.545,.545;%
DarkGoldenrod,.72,.525,.044;%
DarkGray,.664,.664,.664;%
DarkGreen,0,.392,0;%
DarkGrey,.664,.664,.664;%
DarkKhaki,.74,.716,.42;%
DarkMagenta,.545,0,.545;%
DarkOliveGreen,.332,.42,.185;%
DarkOrange,1,.55,0;%
DarkOrchid,.6,.196,.8;%
DarkRed,.545,0,0;%
DarkSalmon,.912,.59,.48;%
DarkSeaGreen,.56,.736,.56;%
DarkSlateBlue,.284,.24,.545;%
DarkSlateGray,.185,.31,.31;%
DarkSlateGrey,.185,.31,.31;%
DarkTurquoise,0,.808,.82;%
DarkViolet,.58,0,.828;%
DeepPink,1,.08,.576;%
DeepSkyBlue,0,.75,1;%
DimGray,.41,.41,.41;%
DimGrey,.41,.41,.41;%
DodgerBlue,.116,.565,1;%
FireBrick,.698,.132,.132;%
FloralWhite,1,.98,.94;%
ForestGreen,.132,.545,.132;%
Fuchsia,1,0,1;%
Gainsboro,.864,.864,.864;%
GhostWhite,.972,.972,1;%
Gold,1,.844,0;%
Goldenrod,.855,.648,.125;%
Gray,.5,.5,.5;%
Grey,.5,.5,.5;%
Green,0,.5,0;%
GreenYellow,.68,1,.185;%
Honeydew,.94,1,.94;%
HotPink,1,.41,.705;%
IndianRed,.804,.36,.36;%
Indigo,.294,0,.51;%
Ivory,1,1,.94;%
Khaki,.94,.9,.55;%
Lavender,.9,.9,.98;%
LavenderBlush,1,.94,.96;%
LawnGreen,.488,.99,0;%
LemonChiffon,1,.98,.804;%
LightBlue,.68,.848,.9;%
LightCoral,.94,.5,.5;%
LightCyan,.88,1,1;%
LightGoldenrodYellow,.98,.98,.824;%
LightGray,.828,.828,.828;%
LightGreen,.565,.932,.565;%
LightGrey,.828,.828,.828;%
LightPink,1,.712,.756;%
LightSalmon,1,.628,.48;%
LightSeaGreen,.125,.698,.668;%
LightSkyBlue,.53,.808,.98;%
LightSlateGray,.468,.532,.6;%
LightSlateGrey,.468,.532,.6;%
LightSteelBlue,.69,.77,.87;%
LightYellow,1,1,.88;%
Lime,0,1,0;%
LimeGreen,.196,.804,.196;%
Linen,.98,.94,.9;%
Magenta,1,0,1;%
Maroon,.5,0,0;%
MediumAquamarine,.4,.804,.668;%
MediumBlue,0,0,.804;%
MediumOrchid,.73,.332,.828;%
MediumPurple,.576,.44,.86;%
MediumSeaGreen,.235,.7,.444;%
MediumSlateBlue,.484,.408,.932;%
MediumSpringGreen,0,.98,.604;%
MediumTurquoise,.284,.82,.8;%
MediumVioletRed,.78,.084,.52;%
MidnightBlue,.098,.098,.44;%
MintCream,.96,1,.98;%
MistyRose,1,.894,.884;%
Moccasin,1,.894,.71;%
NavajoWhite,1,.87,.68;%
Navy,0,0,.5;%
OldLace,.992,.96,.9;%
Olive,.5,.5,0;%
OliveDrab,.42,.556,.136;%
Orange,1,.648,0;%
OrangeRed,1,.27,0;%
Orchid,.855,.44,.84;%
PaleGoldenrod,.932,.91,.668;%
PaleGreen,.596,.985,.596;%
PaleTurquoise,.688,.932,.932;%
PaleVioletRed,.86,.44,.576;%
PapayaWhip,1,.936,.835;%
PeachPuff,1,.855,.725;%
Peru,.804,.52,.248;%
Pink,1,.752,.796;%
Plum,.868,.628,.868;%
PowderBlue,.69,.88,.9;%
Purple,.5,0,.5;%
Red,1,0,0;%
RosyBrown,.736,.56,.56;%
RoyalBlue,.255,.41,.884;%
SaddleBrown,.545,.27,.075;%
Salmon,.98,.5,.448;%
SandyBrown,.956,.644,.376;%
SeaGreen,.18,.545,.34;%
Seashell,1,.96,.932;%
Sienna,.628,.32,.176;%
Silver,.752,.752,.752;%
SkyBlue,.53,.808,.92;%
SlateBlue,.415,.352,.804;%
SlateGray,.44,.5,.565;%
SlateGrey,.44,.5,.565;%
Snow,1,.98,.98;%
SpringGreen,0,1,.498;%
SteelBlue,.275,.51,.705;%
Tan,.824,.705,.55;%
Teal,0,.5,.5;%
Thistle,.848,.75,.848;%
Tomato,1,.39,.28;%
Turquoise,.25,.88,.815;%
Violet,.932,.51,.932;%
Wheat,.96,.87,.7;%
White,1,1,1;%
WhiteSmoke,.96,.96,.96;%
Yellow,1,1,0;%
YellowGreen,.604,.804,.196}

\definecolorset{rgb}{}{}{%
AntiqueWhite1,1,.936,.86;%
AntiqueWhite2,.932,.875,.8;%
AntiqueWhite3,.804,.752,.69;%
AntiqueWhite4,.545,.512,.47;%
Aquamarine1,.498,1,.83;%
Aquamarine2,.464,.932,.776;%
Aquamarine3,.4,.804,.668;%
Aquamarine4,.27,.545,.455;%
Azure1,.94,1,1;%
Azure2,.88,.932,.932;%
Azure3,.756,.804,.804;%
Azure4,.512,.545,.545;%
Bisque1,1,.894,.77;%
Bisque2,.932,.835,.716;%
Bisque3,.804,.716,.62;%
Bisque4,.545,.49,.42;%
Blue1,0,0,1;%
Blue2,0,0,.932;%
Blue3,0,0,.804;%
Blue4,0,0,.545;%
Brown1,1,.25,.25;%
Brown2,.932,.23,.23;%
Brown3,.804,.2,.2;%
Brown4,.545,.136,.136;%
Burlywood1,1,.828,.608;%
Burlywood2,.932,.772,.57;%
Burlywood3,.804,.668,.49;%
Burlywood4,.545,.45,.332;%
CadetBlue1,.596,.96,1;%
CadetBlue2,.556,.898,.932;%
CadetBlue3,.48,.772,.804;%
CadetBlue4,.325,.525,.545;%
Chartreuse1,.498,1,0;%
Chartreuse2,.464,.932,0;%
Chartreuse3,.4,.804,0;%
Chartreuse4,.27,.545,0;%
Chocolate1,1,.498,.14;%
Chocolate2,.932,.464,.13;%
Chocolate3,.804,.4,.112;%
Chocolate4,.545,.27,.075;%
Coral1,1,.448,.336;%
Coral2,.932,.415,.312;%
Coral3,.804,.356,.27;%
Coral4,.545,.244,.185;%
Cornsilk1,1,.972,.864;%
Cornsilk2,.932,.91,.804;%
Cornsilk3,.804,.785,.694;%
Cornsilk4,.545,.532,.47;%
Cyan1,0,1,1;%
Cyan2,0,.932,.932;%
Cyan3,0,.804,.804;%
Cyan4,0,.545,.545;%
DarkGoldenrod1,1,.725,.06;%
DarkGoldenrod2,.932,.68,.055;%
DarkGoldenrod3,.804,.585,.048;%
DarkGoldenrod4,.545,.396,.03;%
DarkOliveGreen1,.792,1,.44;%
DarkOliveGreen2,.736,.932,.408;%
DarkOliveGreen3,.635,.804,.352;%
DarkOliveGreen4,.43,.545,.24;%
DarkOrange1,1,.498,0;%
DarkOrange2,.932,.464,0;%
DarkOrange3,.804,.4,0;%
DarkOrange4,.545,.27,0;%
DarkOrchid1,.75,.244,1;%
DarkOrchid2,.698,.228,.932;%
DarkOrchid3,.604,.196,.804;%
DarkOrchid4,.408,.132,.545;%
DarkSeaGreen1,.756,1,.756;%
DarkSeaGreen2,.705,.932,.705;%
DarkSeaGreen3,.608,.804,.608;%
DarkSeaGreen4,.41,.545,.41;%
DarkSlateGray1,.592,1,1;%
DarkSlateGray2,.552,.932,.932;%
DarkSlateGray3,.475,.804,.804;%
DarkSlateGray4,.32,.545,.545;%
DeepPink1,1,.08,.576;%
DeepPink2,.932,.07,.536;%
DeepPink3,.804,.064,.464;%
DeepPink4,.545,.04,.312;%
DeepSkyBlue1,0,.75,1;%
DeepSkyBlue2,0,.698,.932;%
DeepSkyBlue3,0,.604,.804;%
DeepSkyBlue4,0,.408,.545;%
DodgerBlue1,.116,.565,1;%
DodgerBlue2,.11,.525,.932;%
DodgerBlue3,.094,.455,.804;%
DodgerBlue4,.064,.305,.545;%
Firebrick1,1,.19,.19;%
Firebrick2,.932,.172,.172;%
Firebrick3,.804,.15,.15;%
Firebrick4,.545,.1,.1;%
Gold1,1,.844,0;%
Gold2,.932,.79,0;%
Gold3,.804,.68,0;%
Gold4,.545,.46,0;%
Goldenrod1,1,.756,.145;%
Goldenrod2,.932,.705,.132;%
Goldenrod3,.804,.608,.112;%
Goldenrod4,.545,.41,.08;%
Green1,0,1,0;%
Green2,0,.932,0;%
Green3,0,.804,0;%
Green4,0,.545,0;%
Honeydew1,.94,1,.94;%
Honeydew2,.88,.932,.88;%
Honeydew3,.756,.804,.756;%
Honeydew4,.512,.545,.512;%
HotPink1,1,.43,.705;%
HotPink2,.932,.415,.655;%
HotPink3,.804,.376,.565;%
HotPink4,.545,.228,.385;%
IndianRed1,1,.415,.415;%
IndianRed2,.932,.39,.39;%
IndianRed3,.804,.332,.332;%
IndianRed4,.545,.228,.228;%
Ivory1,1,1,.94;%
Ivory2,.932,.932,.88;%
Ivory3,.804,.804,.756;%
Ivory4,.545,.545,.512;%
Khaki1,1,.965,.56;%
Khaki2,.932,.9,.52;%
Khaki3,.804,.776,.45;%
Khaki4,.545,.525,.305;%
LavenderBlush1,1,.94,.96;%
LavenderBlush2,.932,.88,.898;%
LavenderBlush3,.804,.756,.772;%
LavenderBlush4,.545,.512,.525;%
LemonChiffon1,1,.98,.804;%
LemonChiffon2,.932,.912,.75;%
LemonChiffon3,.804,.79,.648;%
LemonChiffon4,.545,.536,.44;%
LightBlue1,.75,.936,1;%
LightBlue2,.698,.875,.932;%
LightBlue3,.604,.752,.804;%
LightBlue4,.408,.512,.545;%
LightCyan1,.88,1,1;%
LightCyan2,.82,.932,.932;%
LightCyan3,.705,.804,.804;%
LightCyan4,.48,.545,.545;%
LightGoldenrod1,1,.925,.545;%
LightGoldenrod2,.932,.864,.51;%
LightGoldenrod3,.804,.745,.44;%
LightGoldenrod4,.545,.505,.298;%
LightPink1,1,.684,.725;%
LightPink2,.932,.635,.68;%
LightPink3,.804,.55,.585;%
LightPink4,.545,.372,.396;%
LightSalmon1,1,.628,.48;%
LightSalmon2,.932,.585,.448;%
LightSalmon3,.804,.505,.385;%
LightSalmon4,.545,.34,.26;%
LightSkyBlue1,.69,.888,1;%
LightSkyBlue2,.644,.828,.932;%
LightSkyBlue3,.552,.712,.804;%
LightSkyBlue4,.376,.484,.545;%
LightSteelBlue1,.792,.884,1;%
LightSteelBlue2,.736,.824,.932;%
LightSteelBlue3,.635,.71,.804;%
LightSteelBlue4,.43,.484,.545;%
LightYellow1,1,1,.88;%
LightYellow2,.932,.932,.82;%
LightYellow3,.804,.804,.705;%
LightYellow4,.545,.545,.48;%
Magenta1,1,0,1;%
Magenta2,.932,0,.932;%
Magenta3,.804,0,.804;%
Magenta4,.545,0,.545;%
Maroon1,1,.204,.7;%
Maroon2,.932,.19,.655;%
Maroon3,.804,.16,.565;%
Maroon4,.545,.11,.385;%
MediumOrchid1,.88,.4,1;%
MediumOrchid2,.82,.372,.932;%
MediumOrchid3,.705,.32,.804;%
MediumOrchid4,.48,.215,.545;%
MediumPurple1,.67,.51,1;%
MediumPurple2,.624,.475,.932;%
MediumPurple3,.536,.408,.804;%
MediumPurple4,.365,.28,.545;%
MistyRose1,1,.894,.884;%
MistyRose2,.932,.835,.824;%
MistyRose3,.804,.716,.71;%
MistyRose4,.545,.49,.484;%
NavajoWhite1,1,.87,.68;%
NavajoWhite2,.932,.81,.63;%
NavajoWhite3,.804,.7,.545;%
NavajoWhite4,.545,.475,.37;%
OliveDrab1,.752,1,.244;%
OliveDrab2,.7,.932,.228;%
OliveDrab3,.604,.804,.196;%
OliveDrab4,.41,.545,.132;%
Orange1,1,.648,0;%
Orange2,.932,.604,0;%
Orange3,.804,.52,0;%
Orange4,.545,.352,0;%
OrangeRed1,1,.27,0;%
OrangeRed2,.932,.25,0;%
OrangeRed3,.804,.215,0;%
OrangeRed4,.545,.145,0;%
Orchid1,1,.512,.98;%
Orchid2,.932,.48,.912;%
Orchid3,.804,.41,.79;%
Orchid4,.545,.28,.536;%
PaleGreen1,.604,1,.604;%
PaleGreen2,.565,.932,.565;%
PaleGreen3,.488,.804,.488;%
PaleGreen4,.33,.545,.33;%
PaleTurquoise1,.732,1,1;%
PaleTurquoise2,.684,.932,.932;%
PaleTurquoise3,.59,.804,.804;%
PaleTurquoise4,.4,.545,.545;%
PaleVioletRed1,1,.51,.67;%
PaleVioletRed2,.932,.475,.624;%
PaleVioletRed3,.804,.408,.536;%
PaleVioletRed4,.545,.28,.365;%
PeachPuff1,1,.855,.725;%
PeachPuff2,.932,.796,.68;%
PeachPuff3,.804,.688,.585;%
PeachPuff4,.545,.468,.396;%
Pink1,1,.71,.772;%
Pink2,.932,.664,.72;%
Pink3,.804,.57,.62;%
Pink4,.545,.39,.424;%
Plum1,1,.732,1;%
Plum2,.932,.684,.932;%
Plum3,.804,.59,.804;%
Plum4,.545,.4,.545;%
Purple1,.608,.19,1;%
Purple2,.57,.172,.932;%
Purple3,.49,.15,.804;%
Purple4,.332,.1,.545;%
Red1,1,0,0;%
Red2,.932,0,0;%
Red3,.804,0,0;%
Red4,.545,0,0;%
RosyBrown1,1,.756,.756;%
RosyBrown2,.932,.705,.705;%
RosyBrown3,.804,.608,.608;%
RosyBrown4,.545,.41,.41;%
RoyalBlue1,.284,.464,1;%
RoyalBlue2,.264,.43,.932;%
RoyalBlue3,.228,.372,.804;%
RoyalBlue4,.152,.25,.545;%
Salmon1,1,.55,.41;%
Salmon2,.932,.51,.385;%
Salmon3,.804,.44,.33;%
Salmon4,.545,.298,.224;%
SeaGreen1,.33,1,.624;%
SeaGreen2,.305,.932,.58;%
SeaGreen3,.264,.804,.5;%
SeaGreen4,.18,.545,.34;%
Seashell1,1,.96,.932;%
Seashell2,.932,.898,.87;%
Seashell3,.804,.772,.75;%
Seashell4,.545,.525,.51;%
Sienna1,1,.51,.28;%
Sienna2,.932,.475,.26;%
Sienna3,.804,.408,.224;%
Sienna4,.545,.28,.15;%
SkyBlue1,.53,.808,1;%
SkyBlue2,.494,.752,.932;%
SkyBlue3,.424,.65,.804;%
SkyBlue4,.29,.44,.545;%
SlateBlue1,.512,.435,1;%
SlateBlue2,.48,.404,.932;%
SlateBlue3,.41,.35,.804;%
SlateBlue4,.28,.235,.545;%
SlateGray1,.776,.888,1;%
SlateGray2,.725,.828,.932;%
SlateGray3,.624,.712,.804;%
SlateGray4,.424,.484,.545;%
Snow1,1,.98,.98;%
Snow2,.932,.912,.912;%
Snow3,.804,.79,.79;%
Snow4,.545,.536,.536;%
SpringGreen1,0,1,.498;%
SpringGreen2,0,.932,.464;%
SpringGreen3,0,.804,.4;%
SpringGreen4,0,.545,.27;%
SteelBlue1,.39,.72,1;%
SteelBlue2,.36,.675,.932;%
SteelBlue3,.31,.58,.804;%
SteelBlue4,.21,.392,.545;%
Tan1,1,.648,.31;%
Tan2,.932,.604,.288;%
Tan3,.804,.52,.248;%
Tan4,.545,.352,.17;%
Thistle1,1,.884,1;%
Thistle2,.932,.824,.932;%
Thistle3,.804,.71,.804;%
Thistle4,.545,.484,.545;%
Tomato1,1,.39,.28;%
Tomato2,.932,.36,.26;%
Tomato3,.804,.31,.224;%
Tomato4,.545,.21,.15;%
Turquoise1,0,.96,1;%
Turquoise2,0,.898,.932;%
Turquoise3,0,.772,.804;%
Turquoise4,0,.525,.545;%
VioletRed1,1,.244,.59;%
VioletRed2,.932,.228,.55;%
VioletRed3,.804,.196,.47;%
VioletRed4,.545,.132,.32;%
Wheat1,1,.905,.73;%
Wheat2,.932,.848,.684;%
Wheat3,.804,.73,.59;%
Wheat4,.545,.494,.4;%
Yellow1,1,1,0;%
Yellow2,.932,.932,0;%
Yellow3,.804,.804,0;%
Yellow4,.545,.545,0;%
Gray0,.333,.333,.333;%
Green0,0,1,0;%
Grey0,.745,.745,.745;%
Maroon0,.69,.19,.376;%
Purple0,.628,.125,.94}

\def\genbox#1#2#3#4#5#6{
    \leavevmode\raise#4bp\hbox to#5bp{\vrule height#5bp depth0bp width0bp
    \pdfliteral{q .5 w \csname #2COLOR\endcsname\space RG
                       \csname #3PDF\endcsname{#5}{#6} S Q
             \ifx1#1 q \csname #2COLOR\endcsname\space rg 
                       \csname #3PDF\endcsname{#5}{#6} f Q\fi}\hss}}

\begin{document}
	
\begin{textblock*}{18.2cm}(1.7cm,26.5cm) 
	\footnotesize
	\textcolor{red}{
	Accepted for publication in IEEE Transactions on Mobile Computing. © 2026 IEEE. Personal use of this material is permitted. Permission from IEEE must be obtained for all other uses, in any current or future media, including reprinting/republishing this material for advertising or promotional purposes, creating new collective works, for resale or redistribution to servers or lists, or reuse of any copyrighted component of this work in other works.}
\end{textblock*}




\title{\Huge Optimal Radio Resource Management for ISAC Under Imperfect Information: A Resource Economy-Driven Perspective}

\author{Luis F. Abanto-Leon and Setareh Maghsudi
\thanks{Luis F. Abanto-Leon and Setareh Maghsudi are affiliated to Ruhr University Bochum, with emails luis.abantoleon@ruhr-uni-bochum.de and setareh.maghsudi@ruhr-uni-bochum.de.}
}


\markboth{Journal of \LaTeX\ Class Files,~Vol.~15, No.~27, June~2019}%
{Shell \MakeLowercase{\textit{et al.}}: A Sample Article Using IEEEtran.cls for IEEE Journals}


\maketitle


\begin{abstract}

This work investigates the \gls{RRM} design for downlink \gls{ISAC} systems, jointly optimizing timeslot allocation, beam adaptation, functionality selection, and user-target pairing, with the goal of economizing resource consumption under imperfect information. Timeslot allocation assigns a number of discrete channel uses to targets and users, while beam adaptation selects transmit and receive beams with suitable directions, power levels, and beamwidths. Functionality selection determines whether each timeslot is used for sensing, communication, or their simultaneous operation, while user-target pairing specifies which users and targets are jointly served within the same timeslot. To ensure reliable operation, information imperfections arising from motion, quantization, feedback delays, and hardware limitations are considered. Resource economization is achieved by minimizing energy and time consumption through a multi-objective function, with strict prioritization of time savings. 

The resulting \gls{RRM} problem is formulated as a semi-infinite, nonconvex \gls{MINLP}. Given the lack of generic methods for solving such problems, we propose a tailor-made approach that exploits the underlying structure of the problem to uncover hidden convexities. This enables an exact reformulation as a \gls{MISDP}, which can be solved to global optimality. Simulations reveal important interdependencies among the considered \gls{RRM} components and show that the proposed approach achieves substantial performance improvements over baseline schemes, with gains up to $ 88\% $.

\end{abstract}




\glsresetall
\section{Introduction} \label{sec:introduction}

\Gls{ISAC} marks a groundbreaking advancement in wireless technology by seamlessly combining sensing and communications functionalities within the same hardware, spectrum, and waveform  \cite{liu2020:joint-radar-communication-design-applications-state-art-road-ahead, balef2023:adaptive-energy-efficient-waveform-design-joint-communication-sensing-multiobjective-multiarmed-bandits}. This tight integration promises to maximize radio resource utilization efficiency and reduce costs, enabling enhanced performance and capabilities across a wide range of applications~\cite{liu2022:integrated-sensing-communications-toward-dual-functional-wireless-networks-6g-beyond}.


Millimeter-wave and terahertz frequencies offer compelling advantages for both sensing and communications. Their shorter wavelengths enable finer angular resolution and higher Doppler sensitivity, which translate into improved localization accuracy, enhanced target discrimination, and high-resolution imaging \cite{yao2014:terahertz-active-imaging-radar-preprocessing-experiment-results}. These frequencies also provide access to vast amounts of underutilized spectrum, facilitating high-throughput communications \cite{askar2024:mobilizing-terahertz-beam-d-band-analog-beamforming-front-end-prototyping-long-range-6g-trials}. This dual benefit has motivated numerous studies to explore the synergy between \gls{ISAC} and higher frequency bands as a key enabler for fulfilling the stringent and diverse requirements of next-generation wireless networks \cite{mao2022:waveform-design-joint-sensing-communications-millimeter-wave-low-terahertz-bands, zhuo2024:multi-beam-integrated-sensing-communication-state-of-the-art-challenges-opportunities, abanto2024:hierarchical-functionality-prioritization-multicast-isac-optimal-admission-control-discrete-phase-beamforming}.

Unlocking the full potential of \gls{ISAC} at high frequencies hinges on the effectiveness of \gls{RRM} design. Beamforming stands out as a crucial enabler among \gls{ISAC}'s \gls{RRM} components, drawing significant attention across a wide range of applications, including multicasting \cite{abanto2024:hierarchical-functionality-prioritization-multicast-isac-optimal-admission-control-discrete-phase-beamforming, ren2024:fundamental-crb-rate-tradeoff-multi-antenna-isac-systems-information-multicasting-multi-target-sensing}, autonomous driving \cite{xu2023:joint-antenna-selection-beamforming-integrated-automotive-radar-sensing-communications-quantized-double-phase-shifters, cong2023:vehicular-behavior-aware-beamforming-design-integrated-sensing-communication-systems}, physical-layer security \cite{ma2023:covert-beamforming-design-integrated-radar-sensing-communication-systems, bazzi2024:secure-full-duplex-integrated-sensing-communications}, and industrial IoT \cite{dong2024:beamforming-design-integrated-sensing-over-the-air-computation-communication-internet-robotic-things}. While digital beamforming offers fine-grained control over the beampattern, its hardware complexity and power consumption scale unfavorably with both operating frequency and array size, making it impractical for large antenna arrays at millimeter-wave and terahertz frequencies \cite{abanto2022:sequential-parametric-optimization-rate-splitting-precoding-non-orthogonal-unicast-multicast-transmissions, abanto2020:fairness-aware-hybrid-precoding-mmwave-noma-unicast-multicast-transmissions-industrial-iot, elbir2024:hybrid-beamforming-integrated-sensing-communications-low-resolution-dacs}. In this context, analog beamforming constitutes an attractive alternative due to its superior energy and cost efficiency. Although its beampattern control is inherently more constrained, analog beamforming aligns well with current and near-term high-frequency hardware capabilities and can still deliver competitive performance when combined with appropriate beam selection strategies.

Motivated by this, numerous studies on analog beamforming for \gls{ISAC} have focused on beam selection as a central design mechanism \cite{jain2024:commrad-context-aware-sensing-driven-millimeter-wave-networks, zhao2020:m-cube:-millimeter-wave-massive-mimo-software-radio}. However, most existing works optimize only the beam direction while assuming constant power and constant beamwidth, e.g., \cite{xiao2024:simultaneous-multi-beam-sweeping-mmwave-massive-mimo-integrated-sensing-communication, rahman2019:joint-communication-radar-sensing-5g-mobile-network-compressive-sensing}. Such assumptions can be restrictive, as \gls{ISAC} systems must accommodate heterogeneous user and target requirements, necessitating beams with adaptive power and beamwidth. While some works explored power control, e.g., \cite{wang2024:resource-allocation-isac-networks-application-target-tracking, hersyandika2024:guard-beams-coverage-enhancement-ue-centered-isac-analog-multi-beamforming, ding2018:beam-index-modulation-wireless-communication-analog-beamforming}, they assumed infinitesimal power resolution, which is incompatible with the discrete power levels supported by practical hardware.
Beamwidth adaptation has received even less attention within \gls{ISAC} \cite{du2023:integrated-sensing-communications-v2i-networks-dynamic-predictive-beamforming-extended-vehicle-targets, zhang2023:robust-beamforming-design-uav-communications-integrated-sensing-communication}, despite its established role in communication systems, e.g., \cite{karacora2023:event-based-beam-tracking-dynamic-beamwidth-adaptation-terahertz-communications, feng2021:beamwidth-optimization-5g-nr-millimeter-wave-cellular-networks-multi-armed-bandit-approach, chung2021:adaptive-beamwidth-control-mmwave-beam-tracking}. Moreover, existing \gls{ISAC} beamforming research has predominantly focused on the transmit side, often overlooking the receive side, which is critical for effective self-interference management \cite{barneto2022:beamformer-design-optimization-joint-communication-full-duplex-sensing-mmwaves, liu2023:joint-transmit-receive-beamforming-design-full-duplex-integrated-sensing-communications}. \emph{This highlights the need for a comprehensive \gls{ISAC} \gls{RRM} design that jointly adapts beam direction, power, and beamwidth for both transmit and receive sides, while maintaining strong practical relevance.}

From a system-level perspective, another fundamental aspect of analog beamforming is time allocation. Owing to the shared \gls{RF} chain inherent to analog architectures, the system cannot support the simultaneous handling of multiple data streams, which makes time-sharing an essential design consideration. Although time allocation has been studied extensively in standalone sensing e.g., \cite{liu2020:cognitive-dwell-time-allocation-distributed-radar-sensor-networks-tracking-cone-programming , zhang2023:fast-solver-dwell-time-allocation-phased-array-radar }, and communication systems, e.g., \cite{abanto2020:hydrawave-multi-group-multicast-hybrid-precoding-low-latency-scheduling-ubiquitous-industry-4-0-mmwave-communications , nguyen2017:joint-fractional-time-allocation-beamforming-downlink-multiuser-miso-systems }, its integration into \gls{ISAC} remains limited, e.g., \cite{xu2023:sensing-enhanced-secure-communication-joint-time-allocation-beamforming-design, wang2024:resource-allocation-isac-networks-application-target-tracking}. Notably, \cite{xu2023:sensing-enhanced-secure-communication-joint-time-allocation-beamforming-design} considered continuous-time allocation, whereas \cite{wang2024:resource-allocation-isac-networks-application-target-tracking} examined discrete-time allocation. The latter is particularly relevant for practical systems, as it aligns with the discretized frame structures of modern wireless standards and facilitates flexible timeslot sharing among users and targets. \emph{Thus, timeslot allocation becomes essential for meeting heterogeneous sensing and communication demands in analog beamforming-based \gls{ISAC} systems, highlighting the need for \gls{RRM} designs that explicitly account for this key resource.}

While the concurrent operation of sensing and communication is commonly assumed in the literature, enforcing this requirement can reduce flexibility and degrade performance, as the two functionalities may not always be jointly feasible or efficient, given system requirements and resource constraints. To alleviate this, a few studies explored adaptive functionality selection strategies. For instance, \cite{li2022:multi-point-integrated-sensing-communication-fusion-model-functionality-selection} considered mutually exclusive sensing and communication operation, while \cite{abanto2024:hierarchical-functionality-prioritization-multicast-isac-optimal-admission-control-discrete-phase-beamforming} introduced a priority-based approach in which one functionality is opportunistically enabled according to predefined hierarchies. However, these approaches did not account for the temporal domain of resources and lack support for flexible functionality selection, where each timeslot can be dynamically assigned to sensing, communication, or their simultaneous operation. \emph{Hence, incorporating flexible, timeslot-level functionality selection into the \gls{RRM} design is essential for improving performance, as it enables the system to dynamically determine the active functionalities in each timeslot based on instantaneous requirements and resource availability.}

Most existing studies assumed that user-target pairs are predetermined and provided as prior information. However, recent research, e.g., \cite{dou2024:channel-sharing-aided-integrated-sensing-communication-energy-efficient-sensing-scheduling-approach, cazzella2024:deep-learning-based-target-to-user-association-integrated-sensing-communication-systems, abanto2024:optimal-user-target-scheduling-user-target-pairing-low-resolution-phase-only-beamforming-isac-systems}, has demonstrated that optimizing user-target pairing can significantly enhance system efficiency. Despite this progress, most of these studies, e.g., \cite{dou2024:channel-sharing-aided-integrated-sensing-communication-energy-efficient-sensing-scheduling-approach, cazzella2024:deep-learning-based-target-to-user-association-integrated-sensing-communication-systems}, assumed independent waveforms for communication and sensing, failing to fully exploit the benefits of single-waveform \gls{ISAC}. An exception is \cite{abanto2024:optimal-user-target-scheduling-user-target-pairing-low-resolution-phase-only-beamforming-isac-systems}, which addressed user-target pairing under a single-waveform assumption, albeit restricted to a single channel use. \emph{Incorporating dynamic user-target pairing is essential, particularly under resource-constrained conditions, as optimally pairing users and targets to share a common waveform can significantly improve system performance and reduce overall resource consumption.}

Information available for \gls{RRM} design is often compromised by errors caused by factors such as motion, quantization, feedback delay, and hardware deficiencies. Despite their relevance, these errors are frequently overlooked in \gls{ISAC} literature, with only a few exceptions accounting for them, e.g., \cite{xu2023:sensing-enhanced-secure-communication-joint-time-allocation-beamforming-design, jia2024:physical-layer-security-optimization-cramer-rao-bound-metric-isac-systems-sensing-specific-imperfect-csi-model, lyu2024:dual-robust-integrated-sensing-communication-beamforming-csi-imperfection-location-uncertainty, khalili2024:advanced-isac-design-movable-antennas-accounting-dynamic-rcs, tang2021:self-interference-resistant-ieee-802-11-ad-based-joint-communication-automotive-radar-design}. Such errors can result in imperfect estimations of \gls{CSI} \cite{jia2024:physical-layer-security-optimization-cramer-rao-bound-metric-isac-systems-sensing-specific-imperfect-csi-model, lyu2024:dual-robust-integrated-sensing-communication-beamforming-csi-imperfection-location-uncertainty}, \gls{AOD} \cite{xu2023:sensing-enhanced-secure-communication-joint-time-allocation-beamforming-design, lyu2024:dual-robust-integrated-sensing-communication-beamforming-csi-imperfection-location-uncertainty}, \gls{RC} \cite{xu2023:sensing-enhanced-secure-communication-joint-time-allocation-beamforming-design, khalili2024:advanced-isac-design-movable-antennas-accounting-dynamic-rcs}, and \gls{RSI} \cite{tang2021:self-interference-resistant-ieee-802-11-ad-based-joint-communication-automotive-radar-design}. Ignoring potential imperfections in the \gls{RRM} design can lead to infeasible solutions, thereby resulting in degraded system performance, unmet quality-of-service requirements, and frequent, inefficient resource re-allocations. \emph{Therefore, robust \gls{RRM} formulations that explicitly account for information uncertainty are essential to ensure stable and reliable \gls{ISAC} operation.}

Adopting an economy-driven perspective in \gls{RRM} design is essential for achieving efficient radio resource utilization. While a substantial body of \gls{ISAC} literature has focused on minimizing energy consumption, owing to its direct impact on interference mitigation and carbon footprint reduction, e.g., \cite{he2023:full-duplex-communication-isac-joint-beamforming-power-optimization, huang2022:coordinated-power-control-network-integrated-sensing-communication, khalili2024:advanced-isac-design-movable-antennas-accounting-dynamic-rcs}, other critical resources, particularly time, remain understudied. Minimizing time resource consumption is essential not only for meeting latency requirements in time-sensitive applications, but also because analog beamforming supports only a single signal stream at a time. \emph{Hence, the joint minimization of time and energy consumption is paramount for achieving resource efficiency and for provisioning \gls{ISAC} systems to meet future demands.}

Building on the above motivation, this paper advances the ISAC literature by investigating a novel and comprehensive \gls{RRM} problem with unique characteristics. Particularly, \textbf{\cref{tab:related-literature}} in \textbf{Appendix~\ref{app:related-work}} presents a detailed comparison of our work against existing literature\footnote{This paper includes an appendix that provides additional material.}. The key contributions of this work are as follows.

\begin{itemize}

	\item To address the inherent limitation of analog beamforming in supporting only a single data stream at a time, we introduce a timeslot allocation mechanism that efficiently assigns timeslots to users and targets according to their specific service requirements. In parallel, we optimize beam adaptation on a per-timeslot basis, including direction, power, and beamwidth, at both the transmit and receive sides, by employing a column-generation approach and formulating beam selection as a multiple-choice constraint.
	
	\item To enable flexible control over sensing and communication operations, we incorporate  timeslot-level functionality selection that allows each timeslot to be allocated exclusively to sensing, communication, or both. Furthermore, motivated by recent evidence on the benefits of optimized user-target pairing, we extend the pairing process to a dynamic, per-timeslot basis, thereby further enhancing overall system efficiency.
	
	\item To ensure reliability under practical imperfections, we develop a robust design accounting for uncertainties in \gls{CSI}, \gls{AOD}, \gls{RC}, and \gls{RSI}. We also introduce a multi-objective function that jointly minimizes energy and time consumption. Energy consumption is modeled as the sum of transmit and receive energies, while time consumption is modeled as a weighted sum of timeslots, with weights designed to promote contiguous allocation and improve temporal cohesion. The multi-objective function is then transformed into a weighted trade-off function, allowing the concurrent optimization of energy and time consumption.  
	
	\item Integrating the above elements, we formulate a comprehensive \gls{RRM} problem as a semi-infinite, nonconvex \gls{MINLP}, posing significant challenges for its solution. To address this, we propose a tailored solution approach that exploits the problem structure to uncover hidden convexities, enabling an exact reformulation as a \gls{MISDP} solvable to global optimality using general-purpose solvers.
	
	\item We evaluate the proposed \gls{RRM} framework under a wide range of parameter settings and varying levels of imperfections to demonstrate its robustness and adaptability. We also benchmark its performance against baselines, demonstrating gains of up to $ 88\% $.
	
\end{itemize}

\emph{Notation}: Boldface capital letters $ \mathbf{A} $ and boldface lowercase letters $ \mathbf{a} $ denote matrices and vectors, respectively. The transpose, Hermitian transpose, and trace of $ \mathbf{A} $ are denoted by $ \mathbf{A}^\mathrm{T} $, $ \mathbf{A}^\mathrm{H} $, and $ \mathrm{Tr} \left(  \mathbf{A} \right) $, respectively. Also, $ \mathrm{j} \triangleq \sqrt{-1} $ is the imaginary unit, $ \mathbb{E} \left\lbrace \cdot \right\rbrace  $ denotes statistical expectation, and $ \mathcal{CN} \left( \upsilon, \xi^2 \right) $ represents the complex Gaussian distribution with mean $ \upsilon $ and variance $ \xi^2 $. Symbols $ \left| \cdot \right| $, $ \left\| \cdot \right\|_2 $, and  $ \left\| \cdot \right\|_\mathrm{F} $ denote the absolute value, $ \ell_2 $-norm, and Frobenius norm, respectively. In addition, $ \wedge $ and $ \lor $ represent logical `AND' and logical `OR' operators, respectively. Vectorization, Kronecker product, and Hadamard product are represented by $ \mathrm{vec} \left( \cdot \right) $, $ \otimes $, and $ \circ $, respectively. Finally, a generic set is defined as $ \mathcal{X} = \left\lbrace \mathbf{X} \mid \mathrm{Condition}  (\mathbf{X}) \right\rbrace $, where $ \mathrm{Condition}  (\mathbf{X}) $ specifies the defining property of the set.


\section{System Model and Problem Formulation} \label{sec:system-model-problem-formulation}

\subsection{Preliminaries}

Consider an \gls{ISAC} system where a \gls{BS}, equipped with $ N_\mathrm{tx} $ transmit antennas and $ N_\mathrm{rx} $ receive antennas, serves $ U $ single-antenna users and senses $ T $ targets. Users and targets are assumed to lie in the far field of the \gls{BS}, i.e., at distances exceeding the Rayleigh distance. The \gls{BS} operates in a full-duplex downlink mode, utilizing analog beamforming for both transmission and reception with perfectly calibrated arrays. The transmit antenna array emits signals that serve a dual purpose: (i) delivering data to users and (ii) acting as radar waveforms for target sensing. Meanwhile, the receive antenna array captures reflected signals from targets.

The analog transmit beamformer can steer signals across a predefined set of angular directions, with multiple configurable options for power levels and beamwidths at each direction. Similarly, the analog receive beamformer is selected from a discrete set of angular directions, power levels, and beamwidths. Since analog beamforming is limited to processing one signal at a time, the \gls{BS} employs timeslot allocation to efficiently coordinate the communication and sensing demands of users and targets over a finite horizon of $ S $ timeslots. These timeslots may be exclusively allocated for communication, exclusively for sensing, or configured to support both functionalities simultaneously. Enabling a timeslot to support both functionalities implies that a user-target pair is jointly served within that timeslot. Additionally, measurements of \gls{CSI}, \gls{AOD}, \gls{RC}, and \gls{RSI} are available at the \gls{BS} for \gls{RRM} design, albeit impaired by unknown errors that introduce uncertainty.

The allocation of timeslots and beams, along with user-target pairing and timeslot functionality decisions, is driven by the goal of minimizing energy and time consumption while ensuring that the communication and sensing requirements are met. For simplicity, we denote the $u$-th user, $t$-th target, and $s$-th timeslot as $ \mathsf{U}_u $, $ \mathsf{T}_t $, and, $ \mathsf{S}_s $, respectively. We define the sets of users, targets, and timeslots as $ \mathcal{U} = \left\lbrace 1, \dots, U \right\rbrace $, $ \mathcal{T} = \left\lbrace 1, \dots, T \right\rbrace $, and $ \mathcal{S} = \left\lbrace 1, \dots, \widetilde{S} \right\rbrace $, respectively, where $ \widetilde{S} $ is the maximal number of timeslots required by the \gls{BS} for RRM. The system model is illustrated in \cref{fig:system-model}, and all parameters and variables used in this work are summarized in \textbf{\cref{tab:parameters-variables}} in \textbf{Appendix~\ref{app:table-parameters}}.

\begin{figure}[!t]
	\centering
	\includegraphics[width=0.9\columnwidth]{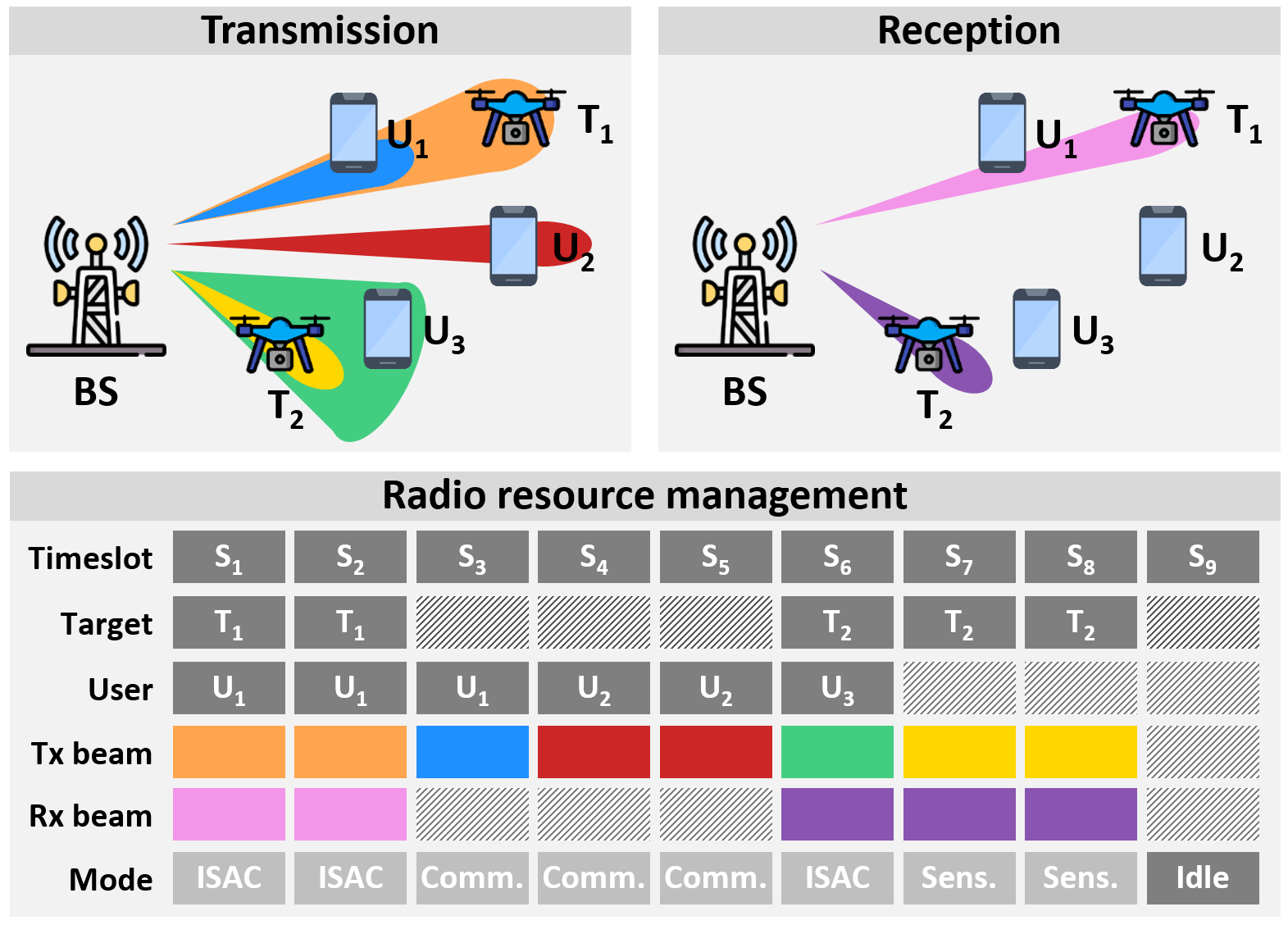}
	\caption{System model comprising a BS, users, and targets.}	
	\label{fig:system-model}
	\vspace{-3mm}
\end{figure}

\begin{example}
\cref{fig:system-model} illustrates $ 3 $ users ($ \mathsf{U}_1 $, $ \mathsf{U}_2 $, $ \mathsf{U}_3 $) and $ 2 $ targets ($ \mathsf{T}_1 $, $ \mathsf{T}_2 $), which must be serviced over $ 9 $ timeslots ($ \mathsf{S}_1 $ to $ \mathsf{S}_9 $), aiming to minimize time and energy consumption. In $ \mathsf{S}_1 - \mathsf{S}_2 $, $ \mathsf{U}_1 $ and $ \mathsf{T}_1 $ are paired and jointly serviced by the orange beam, leveraging their angular alignment. During $\mathsf{S}_1 - \mathsf{S}_2$, the pink beam enables sensing of $\mathsf{T}_1$. In $ \mathsf{S}_3 $, $ \mathsf{U}_1 $ is served alone by the blue beam, as $ \mathsf{T}_1 $'s requirements have been fully met in $ \mathsf{S}_1 - \mathsf{S}_2 $, while $ \mathsf{U}_1 $ requires an additional timeslot. Although the orange beam could have been maintained in $ \mathsf{S}_3 $, switching to the lower-power blue beam minimizes energy consumption. No receive beam is used in $ \mathsf{S}_3 $ since no target is sensed. In $ \mathsf{S}_4 - \mathsf{S}_5 $, $ \mathsf{U}_2 $ is served alone by the red beam. In $ \mathsf{S}_6 $, $ \mathsf{T}_2 $ and $ \mathsf{U}_3 $ are simultaneously served by the broad green beam, reducing time usage. In $ \mathsf{S}_7 - \mathsf{S}_8 $, $ \mathsf{T}_2 $ is served alone by the narrower yellow beam. This shift occurs because $ \mathsf{U}_3 $'s requirements are met in $ \mathsf{S}_6 $, while $ \mathsf{T}_2 $ still requires additional sensing time. Continuing to serve $ \mathsf{T}_2 $ with the green beam beyond $ \mathsf{S}_6 $ would be energy-inefficient. Thus, the narrower yellow beam is selected as a more efficient alternative. Throughout $ \mathsf{S}_6 - \mathsf{S}_8 $, the purple beam enables sensing of $ \mathsf{T}_2 $. Finally, $ \mathsf{S}_9 $ remains idle, ensuring no unnecessary resource expenditure.
\end{example}

\subsection{Functionality selection} 
To determine which functionalities take place in each timeslot, we introduce constraints
\begin{align*}  
	& \mathrm{C}_{1}: \kappa_s = \left\lbrace 0, 1 \right\rbrace, \forall s \in \mathcal{S},   
	\\
	& \mathrm{C}_{2}: \zeta_s = \left\lbrace 0, 1 \right\rbrace, \forall s \in \mathcal{S},   
\end{align*}
where $ \kappa_s = 1 $ indicates that $ \mathsf{S}_s $ is used for communication, and $ \kappa_s = 0 $ otherwise. Similarly, $ \zeta_s = 1 $ indicates that $ \mathsf{S}_s $ is used for sensing, and $ \zeta_s = 0 $ otherwise. In particular, any timeslot can support both functionalities simultaneously when feasible.

To determine whether a given timeslot is active, meaning it is in use, we introduce constraint
\begin{align*}  
	\mathrm{C}_{3}: \gamma_s = \kappa_s \lor \zeta_s, \forall s \in \mathcal{S}. 
\end{align*}

Here, $ \gamma_s = 1 $ indicates that $ \mathsf{S}_s $ is active, while $ \gamma_s = 0 $ indicates that $ \mathsf{S}_s $ is idle. The former case occurs when sensing, communication, or both take place in $ \mathsf{S}_s $, and the latter case occurs when neither functionality is implemented in $ \mathsf{S}_s $. Additionally, we include constraint
\begin{align*}  
	\textstyle
	\mathrm{C}_{4}: \sum_{ s \in \mathcal{S} } \gamma_s \leq S, 
\end{align*}
to ensure that the \gls{RRM} design is confined within a finite horizon of $ S $ timeslots. In practice, we choose $ S \ll \widetilde{S} $ because a smaller value of $ S $ promotes the pairing of users and targets so that they can share timeslots, which is key to realizing the benefits of ISAC. Additional details on parameters $ S $ and $ \widetilde{S} $ are provided in \textbf{Appendix~\ref{app:finite-horizon}}.

\subsection{Timeslot allocation and user-target pairing} 
To allocate timeslots to users, we introduce constraint
\begin{align*}  
	\mathrm{C}_{5}: \mu_{u,s} \in \left\lbrace 0, 1 \right\rbrace, \forall u \in \mathcal{U}, s \in \mathcal{S},   
\end{align*}
where $ \mu_{u,s} = 1 $ indicates that $ \mathsf{U}_u $ is served in $ \mathsf{S}_s $, and $ \mu_{u,s} = 0 $ otherwise. We also include constraint
\begin{align*}  
	\textstyle
	\mathrm{C}_{6}: \kappa_s = \sum_{u \in \mathcal{U}} \mu_{u,s}, \forall s \in \mathcal{S},  
\end{align*}
to ensure that each timeslot serves at most one user. This is guaranteed because $ \kappa_s $ is binary, meaning that at most one of the variables $ \mu_{u,s} $ can equal one. This restriction arises from the use of analog beamforming at the \gls{BS}, which confines signal processing to a single stream within any given timeslot.

To allocate timeslots to targets, we introduce constraint 
\begin{align*}  
	& \mathrm{C}_{7}: \tau_{t,s} \in \left\lbrace 0, 1 \right\rbrace, \forall t \in \mathcal{T}, s \in \mathcal{S},   
\end{align*}
where $ \tau_{t,s} = 1 $ indicates that $ \mathsf{T}_t $ is served in $ \mathsf{S}_s $, and $ \tau_{t,s} = 0 $ otherwise. Furthermore, we add constraint 
\begin{align*}  
	\textstyle
 	\mathrm{C}_{8}: \zeta_s = \sum_{t \in \mathcal{T}} \tau_{t,s}, \forall s \in \mathcal{S},   
\end{align*}
to restrict each timeslot to serve no more than one target. This is guaranteed because $ \zeta_s $ is binary, ensuring that at most one of the variables $ \tau_{t,s} $ can equal one.

User-target pairing determines which users and targets are jointly served within a given timeslot. This departs from existing approaches that rely on fixed associations or consider pairing for a single channel use, without exploiting the time dimension. In our framework, pairing is enabled only when it is feasible and leads to reduced time and energy consumption. Pairing is particularly beneficial when a user and a target are closely aligned in the angular domain, allowing them to be served simultaneously through a shared beam. However, angular similarity alone is not sufficient, as all system factors also influence the pairing decision. Therefore, treating user-target pairing as an optimizable variable is essential for achieving adaptive and resource-efficient \gls{RRM}. In our framework, user-target pairing is enforced through constraints $ \mathrm{C}_{1} $, $ \mathrm{C}_{2} $, $ \mathrm{C}_{5} $, $ \mathrm{C}_{6} $, $ \mathrm{C}_{7} $, and $ \mathrm{C}_{8} $, which collectively ensure a one-to-one pairing between users and targets whenever feasible. Complementary details on this aspect are provided in \textbf{Appendix~\ref{app:pairing-details}}.

\subsection{Beam adaptation} 

To achieve highly flexible beamforming and accommodate diverse sensing and communication requirements, the \gls{BS} can adjust the direction, power, and beamwidth of its transmit and receive beampatterns. For practical implementation, these adaptations are restricted to a finite set of feasible configurations. To achieve this effectively, we adopt a column generation approach \cite{desaulniers2005:column-generation} and construct a predefined codebook, where each codeword corresponds to a unique combination of direction, power, and beamwidth. In each active timeslot, the \gls{BS} selects one codeword for transmission, and another for reception when sensing is enabled.

Consider $ {D}_\mathrm{tx} $ different angular directions $ \widetilde{\theta}_i $, $ i \in \left\lbrace 1, \dots, {D}_\mathrm{tx} \right\rbrace $, in which the transmit signal can be steered. For each $ \widetilde{\theta}_i $, assume $ {B}_\mathrm{tx} $ different beamwidths, leading to normalized codewords $ \widetilde{\mathbf{b}}_{i,j} $, such that $ \big\| \widetilde{\mathbf{b}}_{i,j} \big\|_2 = 1 $, where $ j \in \left\lbrace 1, \dots, B_\mathrm{tx} \right\rbrace $. Additionally, for each codeword $ \widetilde{\mathbf{b}}_{i,j} $, up to $ {P}_\mathrm{tx} $ different powers can be applied, resulting in codewords $ \widehat{\mathbf{b}}_{i,j,k} $, where $ k \in \left\lbrace 1, \dots, P_\mathrm{tx} \right\rbrace $. The complete set of codewords $ \widehat{\mathbf{b}}_{i,j,k} $ results in $ L_\mathrm{tx} = D_\mathrm{tx} B_\mathrm{tx} P_\mathrm{tx} $ distinct transmit beampatterns. To facilitate a column-generation formulation, we replace the nested indices $ i $, $ j $, and $ k $ with a single subindex $ b \in \mathcal{L}_\mathrm{tx} $, where $ \mathcal{L}_\mathrm{tx} = \left\lbrace 1, \dots, L_\mathrm{tx} \right\rbrace $. The effective transmit codewords are thus denoted as $ \mathbf{b}_b \in \mathbb{C}^{N_\mathrm{tx} \times 1} $, with $ \mathbf{b}_b \in \left\lbrace \widehat{\mathbf{b}}_{1,1,1}, \widehat{\mathbf{b}}_{1,1,2}, \dots, \widehat{\mathbf{b}}_{D_\mathrm{tx},B_\mathrm{tx},P_\mathrm{tx}} \right\rbrace $. Similarly, on the receive side, we consider $ D_\mathrm{rx} $ directions, $ B_\mathrm{rx} $ beamwidths, and $ P_\mathrm{rx} $ power levels, resulting in $ L_\mathrm{rx} = D_\mathrm{rx} B_\mathrm{rx} P_\mathrm{rx} $ distinct receive beampatterns. These are indexed by $ c \in \mathcal{L}_\mathrm{rx} $, with $ \mathcal{L}_\mathrm{rx} = \left\lbrace 1, \dots, L_\mathrm{rx} \right\rbrace $, and the corresponding effective receive codewords are denoted as $ \mathbf{c}_c \in \mathbb{C}^{N_\mathrm{rx} \times 1} $, where $ \mathbf{c}_c \in \left\lbrace \widehat{\mathbf{c}}_{1,1,1}, \widehat{\mathbf{c}}_{1,1,2}, \dots, \widehat{\mathbf{c}}_{D_\mathrm{rx},B_\mathrm{rx},P_\mathrm{rx}} \right\rbrace $

%


To enable transmit beamforming, we introduce 
\begin{align*}  
	\mathrm{C}_{9}: \chi_{b,s} \in \left\lbrace 0, 1 \right\rbrace, \forall b \in \mathcal{L}_\mathrm{tx}, s \in \mathcal{S}, 
\end{align*}
where $ \chi_{b,s} = 1 $ indicates that codeword $ \mathbf{b}_b $ is used in $ \mathsf{S}_s $, and $ \chi_{b,s} = 0 $ otherwise. Furthermore, we introduce constraint
\begin{align*}  
	\textstyle
	\mathrm{C}_{10}: \sum_{b \in \mathcal{L}_\mathrm{tx}} \chi_{b,s} = \gamma_s, \forall s \in \mathcal{S}, 
\end{align*}
to ensure that transmit codewords are selected solely for timeslots that are active, as indicated by $ \gamma_s $. Specifically, no codeword is selected for idle timeslots. Thus, the transmit beamformer employed in $ \mathsf{S}_s $ is defined by constraint
\begin{align*}  
	\textstyle
	\mathrm{C}_{11}: \mathbf{w}_s = \sum_{b \in \mathcal{L}_\mathrm{tx}} \mathbf{b}_b \cdot \chi_{b,s}, \forall s \in \mathcal{S}.
\end{align*}

To enable receive beamforming, we introduce 
\begin{align*}  
	\mathrm{C}_{12}: \rho_{c,s} \in \left\lbrace 0, 1 \right\rbrace, \forall c \in \mathcal{L}_\mathrm{rx}, s \in \mathcal{S}, 
\end{align*}
where $ \rho_{c,s} = 1 $ indicates that codeword $ \mathbf{c}_c $ is used in $ \mathsf{S}_s $, and  $ \rho_{c,s} = 0 $ otherwise. Additionally, we include constraint
\begin{align*}  
	\textstyle
	\mathrm{C}_{13}: \sum_{c \in \mathcal{L}_\mathrm{rx}} \rho_{c,s} = \zeta_s, \forall s \in \mathcal{S}, 
\end{align*}
to ensure that receive codewords are selected only for timeslots performing sensing, as indicated by $ \zeta_s $. The receive beamformer used in $ \mathsf{S}_s $ is given by constraint
\begin{align*} 
	\textstyle 
	\mathrm{C}_{14}: \mathbf{v}_s = \sum_{c \in \mathcal{L}_\mathrm{rx}} \mathbf{c}_c \cdot \rho_{c,s}, \forall s \in \mathcal{S}. 
\end{align*}


\subsection{Communication metric}

Let $ d_s \in \mathbb{C} $ denote the symbol transmitted by the \gls{BS} in timeslot $ \mathsf{S}_s $.  This symbol may be used for communication, sensing, or both functionalities simultaneously. The transmitted symbol follows a complex Gaussian distribution with zero mean and unit variance, i.e., $ \mathbb{E} \left\lbrace d_s d_s^{*} \right\rbrace = 1 $.

Although calibrated antenna arrays are assumed, impairments such as mutual coupling among neighboring antennas may still be present. Mutual coupling alters the beampattern shape by modifying the effective beamforming vectors applied at the antenna ports. Accordingly, when a nominal transmit codeword $ \mathbf{w}_{s} $ is selected in timeslot $ \mathsf{S}_s $, the effective beamformed transmit signal is given by
\begin{align}
	\mathbf{x}_s = \mathbf{Z}_\mathrm{tx} \mathbf{w}_s d_s,
\end{align}
where $ \mathbf{Z}_\mathrm{tx} \in \mathbb{C}^{N_\mathrm{tx} \times N_\mathrm{tx}} $ denotes the transmit-side mutual coupling matrix, defined with banded Toeplitz structure,
\begin{align} \label{eqn:coupling-matrix}
	\mathbf{Z}_\mathrm{tx} = 
	\begin{bmatrix}
		1  & \delta_\mathrm{tx} & 0 & 0 & \cdots & 0 & 0
		\\
		\delta_\mathrm{tx}  & 1 & \delta_\mathrm{tx} & 0 & \cdots & 0 & 0		  
		\\
		0 & \delta_\mathrm{tx}  & 1 & \delta_\mathrm{tx} &  \cdots & 0 & 0	
		\\
		\vdots & \vdots & \vdots & \vdots & \ddots & \vdots & \vdots
		\\
		0 & 0 & 0 & 0 & \cdots & 1 & \delta_\mathrm{tx}
		\\
		0 & 0 & 0 & 0 & \cdots & \delta_\mathrm{tx} & 1
	\end{bmatrix},
\end{align}
where $ \left| \delta_\mathrm{tx} \right| \ll 1 $ represents the coupling strength between adjacent antenna elements \cite{wolosinski2020:closed-form-characterization-mutual-coupling-uniform-linear-arrays, pan2022:joint-estimation-doa-mutual-coupling-imposing-annihilation-constraint}.

The signal received by $ \mathsf{U}_u $ in $ \mathsf{S}_s $ is therefore expressed as
\begin{align}
	\begin{split}
	y_\mathrm{com}^{u,s} & = \textstyle \mathbf{h}_u^\mathrm{H} \mathbf{x}_s \cdot \mu_{u,s} + n_{\mathrm{com},u}, 
	\\
	& = \textstyle \mathbf{h}_u^\mathrm{H} \mathbf{Z}_\mathrm{tx} \mathbf{w}_s d_s \cdot \mu_{u,s} + n_{\mathrm{com},u}, 
	\end{split}
\end{align}
where $ \mathbf{h}_u \in \mathbb{C}^{N_\mathrm{tx} \times 1} $ is the communication channel between the \gls{BS} and $ \mathsf{U}_u $, while $ n_{\mathrm{com},u} \sim \mathcal{CN} \left( 0,\sigma_\mathrm{com}^2 \right) $ represents \gls{AWGN}. The communication \gls{SNR} for $ \mathsf{U}_u $ in $ \mathsf{S}_s $ is defined as
\begin{align} \label{eqn:communication-snr}
	\mathsf{SNR}_\mathrm{com}^{u,s} = \frac{\left| \mathbf{h}_u^\mathrm{H} \mathbf{Z}_\mathrm{tx} \mathbf{w}_s \cdot \mu_{u,s} \right|^2}{\sigma_\mathrm{com}^2}.
\end{align}

In practice, \glspl{BS} operate with imperfect \gls{CSI} due to factors such as quantization and feedback delay, which introduce errors. These errors can significantly affect communication performance, making it essential to incorporate them into the \gls{RRM} design. Imperfect \gls{CSI} is modeled via constraint 
\begin{align*}
	\resizebox{1\columnwidth}{!}
	{$
	\mathrm{C}_{15}: \mathcal{H}_u \triangleq \left\lbrace \mathbf{h}_u \mid \mathbf{h}_u = \widebar{\mathbf{h}}_u + \boldsymbol{\Delta} \mathbf{h}_u, \left\| \boldsymbol{\Delta} \mathbf{h}_u \right\|_2^2 \leq \epsilon_\mathrm{CSI}^2 \right\rbrace, \forall u \in \mathcal{U},
	$}
\end{align*}
using the model from \cite{zhao2020:non-orthogonal-unicast-broadcast-transmission-joint-beamforming-ldm-cellular-networks, ren2023:robust-transmit-beamforming-secure-integrated-sensing-communication}, where $ \mathbf{h}_u $ is the actual but unknown channel, $ \widebar{\mathbf{h}}_u $ is the estimated channel, $ \boldsymbol{\Delta} \mathbf{h}_u $ is the random channel error whose power is bounded by $ \epsilon_\mathrm{CSI}^2 $, and $ \mathcal{H}_u $ represents the uncertainty set of all possible channel vectors for $ \mathsf{U}_u $. 
 
To ensure reliable communication performance, we introduce constraint 
\begin{align*}
	\displaystyle \mathrm{C}_{16}: \min_{ \mathbf{h}_u \in \mathcal{H}_u }  \mathsf{SNR}_\mathrm{com}^{u,s} 
	\geq \Upsilon_{\mathrm{snr},u} \cdot \mu_{u,s}, \forall u \in \mathcal{U}, s \in \mathcal{S}, 
\end{align*}
which guarantees a predefined \gls{SNR} $ \Upsilon_{\mathrm{snr},u} $ for each user $ \mathsf{U}_u $ across its allocated timeslots, while accounting for imperfect \gls{CSI}. Additionally, we include constraint
\begin{align*}  
	\textstyle
	\mathrm{C}_{17}: \sum_{s \in \mathcal{S}} \mu_{u,s} = S_{\mathrm{com},u}, \forall u \in \mathcal{U},  
\end{align*}
to ensure that each $ \mathsf{U}_u $ is allocated exactly $ S_{\mathrm{com},u} $ timeslots. 

\subsection{Sensing metric}
Targets are assumed to lie in the far field of the \gls{BS} and are modeled as single point scatters. The \gls{BS} operates in a monostatic configuration, implying identical \gls{AOD} and \gls{AOA}. Although the far-field assumption removes any explicit dependence of the \gls{AOD} on the physical extent of the target, practical \gls{AOD} estimates remain affected by estimation errors, discrete-time tracking effects, and target micro-motion. To capture these imperfections, we adopt the modeling approach in \cite{xu2020:multiuser-miso-uav-communications-uncertain-environments-no-fly-zones-robust-trajectory-resource-allocation-design, xu2018:robust-resource-allocation-uav-systems-uav-jittering-user-location-uncertainty} and introduce constraint
\begin{align*} 
	\textstyle
	\mathrm{C}_{18}: \Theta_t \triangleq \left\lbrace \theta_t \mid \theta_t = \widebar{\theta}_t + \Delta \theta_t, \left| \Delta \theta_t \right|^2 \leq \epsilon_\mathrm{AOD}^2 \right\rbrace, \forall t \in \mathcal{T},
\end{align*}
where $ \theta_t $ is the actual but unknown \gls{AOD}, $ \widebar{\theta}_t $ is the estimated \gls{AOD}, $ \Delta \theta_t $ is the random \gls{AOD} error whose power is bounded by $ \epsilon_\mathrm{AOD}^2 $, and $ \Theta_t $ represents the uncertainty set of all possible \glspl{AOD} for $ \mathsf{T}_t $.

Another important factor is the \gls{RC}, which captures both the target's radar cross-section and the round-trip path-loss relative to the \gls{BS} \cite{liu2022:cramer-rao-bound-optimization-joint-radar-communication-beamforming}. However, due to quantization errors and interference from clutter, the \gls{RC} is typically not perfectly known. To account for potential errors in this parameter, we adopt the model from \cite{xu2023:sensing-enhanced-secure-communication-joint-time-allocation-beamforming-design} and introduce
\begin{align*} 
	\textstyle
	\mathrm{C}_{19}: \Psi_t \triangleq \left\lbrace \psi_t \mid \psi_t = \widebar{\psi}_t + \Delta \psi_t, \left| \Delta \psi_t \right|^2 \leq \epsilon_\mathrm{RC}^2 \right\rbrace, \forall t \in \mathcal{T},
\end{align*}
where $ \psi_t $ is the actual but unknown \gls{RC}, $ \widebar{\psi}_t $ is the estimated \gls{RC}, $ \Delta \psi_t $ is the random \gls{RC} error whose power is bounded by $ \epsilon_\mathrm{RC}^2 $, and $ \Psi_t $ represents the uncertainty set of all possible \glspl{RC} for $ \mathsf{T}_t $.

\begin{remark} \label{rmk:clarification-aspect-angle}
	The nominal RC $ \widebar{\psi}_t $ corresponds to the estimated AOD $ \widebar{\theta}_t $. In practice, the actual RC varies with the aspect angle. This effect is captured through the bounded uncertainty term $ \Delta \psi_t $ in $ \mathrm{C}_{19} $, where parameter $ \epsilon_\mathrm{RC} $ encompasses both measurement errors and aspect-angle RC fluctuations, thereby including even the worst-case deviation within the angular uncertainty set $ \Theta_t $.
\end{remark}

Considering the \gls{AOD} and \gls{RC}, the sensing channel between the \gls{BS} and $ \mathsf{T}_t $ is defined as
\begin{align} \label{eqn:response-matrix}
	\begin{aligned}
	\mathbf{G}_t & = \psi_t \cdot \mathbf{a}_\mathrm{rx} \left( \theta_t \right) \mathbf{a}^\mathrm{H}_\mathrm{tx} \left( \theta_t \right) 
	 = \psi_t \cdot \mathbf{A} \left( \theta_t \right),
	\end{aligned}
\end{align}
where the transmit and receive steering vectors are given by
\begin{align} \label{eqn:steering-vectors}
	\begin{aligned}
		\mathbf{a}_{l} \left( \theta_t \right) = \tfrac{1}{\sqrt{N_l}} \mathrm{e}^{\mathrm{j} \boldsymbol{\phi}_l \cos \left( \theta_t \right)},
	\end{aligned}
\end{align}
with $ l  = \left\lbrace \mathrm{tx}, \mathrm{rx} \right\rbrace $. Besides, vector $ \boldsymbol{\phi}_l $ is defined as
\begin{align} \label{eqn:steering-angles}
	\begin{aligned}
		\boldsymbol{\phi}_l = \tfrac{2 \pi d_l}{\lambda} \left[  \tfrac{-N_l+1}{2}, \dots,  \tfrac{N_l-1}{2} \right]^\mathrm{T},
	\end{aligned}
\end{align}
where $ \lambda $ represents the wavelength and $ d_l $ denotes the antenna spacings of the transmit/receive array.

\begin{remark} \label{rmk:clarification-point-scatter}
	An extended target can be modeled as a set of $ M $ scattering centers, with the sensing channel expressed as $ \mathbf{G}_t = \sum_{m \in \mathcal{M}} \psi_{t,m} \cdot \mathbf{A} \left( \theta_{t,m} \right) $, where $ \theta_{t,m} $ and $ \psi_{t,m} $ denote the AOD and RC of the $m$-th scattering center, respectively, for $ m \in \mathcal{M} = \left\lbrace 1, \dots, M \right\rbrace $. In this work, we focus on the special case of a single effective scattering center, i.e., $ M = 1 $. Notably, the proposed framework and solution can be directly extended to an arbitrary number of resolved scattering centers ($ M > 1 $), where the sensing SINR follows from incoherent power summation across all centers.
\end{remark}



Note that matrix $ \mathbf{G}_t $ is linear with respect to $ \psi_t $ but nonlinear with respect to $ \theta_t $. The nonlinearity in $ \theta_t $ complicates the model's tractability for further optimization. To address it, we adopt the linear model employed in \cite{xu2020:multiuser-miso-uav-communications-uncertain-environments-no-fly-zones-robust-trajectory-resource-allocation-design, xu2018:robust-resource-allocation-uav-systems-uav-jittering-user-location-uncertainty}, originally designed for transmit and receive steering vectors of equal length. We extend this model to accommodate transmit and receive arrays of different lengths, resulting in a more general construct. Since $ \mathbf{A} (\theta_t) $ is the component of $ \mathbf{G}_t $ that encapsulates the nonlinear dependence on $ \theta_t $, we adopt a linearized approximation of $ \mathbf{A} (\theta_t) $, given by
\begin{align} \label{eqn:linearized-response-matrix}
	\begin{aligned}
	\mathbf{A} \left( \theta_t \right) & \approx  \mathbf{A} \left( \widebar{\theta}_t \right) + \widetilde{\mathbf{A}} \left( \widebar{\theta}_t \right) \Delta \theta_t,
	\end{aligned}
\end{align}
where
$\mathbf{A} \left( \widebar{\theta}_t \right) = \tfrac{1}{\sqrt{N_\mathrm{rx} N_\mathrm{tx}}} \mathrm{e}^{\mathrm{j} \boldsymbol{\Phi}  \cos \left( \widebar{\theta}_t \right)}, \boldsymbol{\Phi} = \mathbf{1}^\mathrm{T} \otimes \boldsymbol{\phi}_\mathrm{rx} - \mathbf{1} \otimes \boldsymbol{\phi}_\mathrm{tx}^\mathrm{T}, \widetilde{\mathbf{A}} \left( \widebar{\theta}_t \right) = \mathbf{A} \left( \widebar{\theta}_t \right) \circ \mathbf{E} \left( \widebar{\theta}_t \right), \mathbf{E} \left( \widebar{\theta}_t \right) = - \mathrm{j} \boldsymbol{\Phi} \sin \left( \widebar{\theta}_t \right) $. The derivation of this model is detailed in \textbf{Appendix~\ref{app:linearized-model}}.

In full-duplex systems, the transmit and receive arrays at the \gls{BS} are typically placed in close proximity, which causes the receive array to capture strong leakage from the transmit array. As a result, in addition to the echoes reflected from the target of interest and environmental clutter, the received signal contains a self-interference component arising from direct coupling between the transmit and receive arrays.

Before processing the target reflections, the BS applies standard monostatic radar clutter-filtering techniques \cite{skolnik2001:introduction-radar-systems, richards2005:fundamentals-radar-signal-processing, cao2020:clutter-suppression-target-tracking-low-rank-representation-airborne-maritime-surveillance-radar}.
Static clutter (e.g., buildings) is typically removed using Doppler windowing or high-pass filters, while dynamic clutter (e.g., moving objects) is mitigated through subspace projection or low-rank filtering. These operations primarily suppress external environmental reflections and are assumed not to significantly distort the deterministic self-interference component. After clutter suppression, the signal received at the \gls{BS} from target $ \mathsf{T}_t $ during $ \mathsf{S}_s $ is 
\begin{align} \label{eqn:signal-received-bs}
	\begin{aligned}
	y_\mathrm{sen}^{t,s} & = \mathbf{v}_s^\mathrm{H} \mathbf{Z}_\mathrm{rx}^\mathrm{H} \big( \mathbf{G}_t + \tilde{\mathbf{I}} \big) \mathbf{x}_s \cdot \tau_{t,s} + \mathbf{v}_s^\mathrm{H} \mathbf{Z}_\mathrm{rx}^\mathrm{H} \mathbf{n}_{\mathrm{sen},t},
	\\
	& = \underbrace{\mathbf{v}_s^\mathrm{H} \mathbf{Z}_\mathrm{rx}^\mathrm{H} \mathbf{G}_t \mathbf{Z}_\mathrm{tx} \mathbf{w}_s d_s \cdot \tau_{t,s}}_\text{desired signal} 
	+ \underbrace{\mathbf{v}_s^\mathrm{H} \mathbf{Z}_\mathrm{rx}^\mathrm{H} \tilde{\mathbf{I}} \mathbf{Z}_\mathrm{tx} \mathbf{w}_s d_s \cdot \tau_{t,s}}_\text{self-interference}
	\\
	& ~~~~~~~~~~~~~~~~~~~~~~~~~~~~~~~~~~~~~~~~~~ + \underbrace{\mathbf{v}_s^\mathrm{H} \mathbf{Z}_\mathrm{rx}^\mathrm{H} \mathbf{n}_{\mathrm{sen},t}}_\text{noise},
	\end{aligned}
\end{align}
where $ \tilde{\mathbf{I}} $ denotes the direct self-interference channel between the transmit and receive arrays, and $ \mathbf{n}_{\mathrm{sen},t} \sim \mathcal{CN} \left( \mathbf{0}, \sigma_\mathrm{sen}^2 \mathbf{I} \right) $ models thermal noise and any residual clutter remaining after suppression \cite{skolnik2001:introduction-radar-systems, richards2005:fundamentals-radar-signal-processing}. Receive-side mutual coupling is modeled via the matrix $ \mathbf{Z}_\mathrm{rx} \in \mathbb{C}^{N\mathrm{rx} \times N_\mathrm{rx}} $, parameterized by coupling strength $ \delta_\mathrm{rx} $, similarly to (\ref{eqn:coupling-matrix}).

Afterwards, the \gls{BS} employs analog and/or digital signal processing techniques to mitigate self-interference \cite{kim2023:performance-analysis-self-interference-cancellation-full-duplex-massive-mimo-systems-subtraction-versus-spatial-suppression }. These techniques operate on the already clutter-filtered signal and rely on estimates of the self-interference channel. Due to hardware imperfections and channel estimation errors, cancellation is inherently imperfect and results in \gls{RSI}. Specifically, the \gls{BS} performs self-interference cancellation by subtracting $ \mathbf{v}_s^\mathrm{H} \mathbf{Z}_\mathrm{rx}^\mathrm{H} \doubletilde{\mathbf{I}} \mathbf{Z}_\mathrm{tx} \mathbf{w}_s d_s \cdot \tau_{t,s} $ from $ y_\mathrm{sen}^{t,s} $, where $ \doubletilde{\mathbf{I}} $ is an estimate of $ \tilde{\mathbf{I}} $. The post-cancellation signal is expressed as
\begin{align} \label{eqn:signal-received-bs-rsi}
\begin{aligned}
	\widebar{y}_\mathrm{sen}^{t,s} & = y_\mathrm{sen}^{t,s} - \mathbf{v}_s^\mathrm{H} \mathbf{Z}_\mathrm{rx}^\mathrm{H} \doubletilde{\mathbf{I}} \mathbf{Z}_\mathrm{tx} \mathbf{w}_s d_s \cdot \tau_{t,s}
	\\
	& = \mathbf{v}_s^\mathrm{H} \mathbf{Z}_\mathrm{rx}^\mathrm{H} \mathbf{G}_t \mathbf{Z}_\mathrm{tx} \mathbf{w}_s d_s \cdot \tau_{t,s} 
	\\
	& ~~~~~~~~~~~~~ + \mathbf{v}_s^\mathrm{H} \mathbf{Z}_\mathrm{rx}^\mathrm{H} \mathbf{R} \mathbf{Z}_\mathrm{tx} \mathbf{w}_s d_s \cdot \tau_{t,s} + \mathbf{v}_s^\mathrm{H} \mathbf{Z}_\mathrm{rx}^\mathrm{H} \mathbf{n}_{\mathrm{sen},t},
\end{aligned}
\end{align}
where $ \mathbf{R} = \tilde{\mathbf{I}} - \doubletilde{\mathbf{I}} $ is the \gls{RSI} channel and $ \mathbf{v}_s^\mathrm{H} \mathbf{Z}_\mathrm{rx}^\mathrm{H} \mathbf{R} \mathbf{Z}_\mathrm{tx} \mathbf{w}_s d_s \cdot \tau_{t,s} $ represents the interfering signal. The primary challenge in achieving complete interference cancellation lies in the imperfect knowledge of $ \mathbf{R} $. Following the models in \cite{rodriguez2014:optimal-power-allocation-capacity-full-duplex-af-relaying-residual-self-interference, liu2017:analog-self-interference-cancellation-full-duplex-communications-imperfect-channel-state-information}, the \gls{RSI} channel is characterized by constraint
\begin{align*}
	\mathrm{C}_{20}: \mathcal{R} \triangleq \left\lbrace \mathbf{R} \mid \mathbf{R} = \upsilon \widebar{\mathbf{R}} + \upsilon \boldsymbol{\Delta} \mathbf{R}, \left\| \boldsymbol{\Delta} \mathbf{R} \right\|_\mathrm{F}^2 \leq \epsilon_\mathrm{RSI}^2 \right\rbrace,
\end{align*}
where $ \mathcal{R} $ is the uncertainty set encompassing all \gls{RSI} possibilities. Here, $ \mathbf{R} $ consists of a deterministic and a random component, denoted by $ \widebar{\mathbf{R}} $ and $ \boldsymbol{\Delta} \mathbf{R} $, respectively. The deterministic component depends on the relative geometry between arrays, and is defined as
\begin{align} \label{eqn:definition-rsi}
	\left[ \widebar{\mathbf{R}} \right]_{\widebar{m},\widebar{n}} = \tfrac{\lambda}{4 \pi \widebar{d}_{\widebar{m},\widebar{n}}} \tfrac{1}{\sqrt{N_\mathrm{rx} N_\mathrm{tx}}} \mathrm{e}^{-\mathrm{j} \frac{2 \pi}{\lambda} \widebar{d}_{\widebar{m},\widebar{n}}}.
\end{align}

Specifically, $ \widebar{d}_{\widebar{m},\widebar{n}} $ denotes the distance between the $ \widebar{m} $-th receive antenna and the $ \widebar{n} $-th transmit antenna \cite{wang2023:near-field-integrated-sensing-communications, tse2005:fundamentals-wireless-ommunication}, with the centers of the antenna arrays separated by a distance $ \widebar{d}_c $. Meanwhile, the random component accounts for spurious self-interference whose power is bounded by $ \epsilon_\mathrm{RSI}^2 $. Parameter $ \upsilon $ quantifies the severity of self-interference in the processed signal. Thus,  $ \upsilon = 0 $ indicates perfect self-interference cancellation, while $ \upsilon = 1 $ signifies the absence of any cancellation.

\begin{remark} \label{rmk:clarification-rsi}
	The norm-bounded uncertainty set that models RSI is designed to capture the aggregate residual self-interference after analog/digital cancellation. This modeling approach accounts not only for linear estimation errors but also for distortions arising from nonlinear RF impairments, such as power amplifier saturation and IQ imbalance. These nonlinear effects are implicitly captured through the parameter $ \epsilon_\mathrm{RSI} $, which upper-bounds the total residual interference energy after cancellation.
\end{remark}

From (\ref{eqn:signal-received-bs-rsi}), the sensing \gls{SINR} for target $ \mathsf{T}_t $ in $ \mathsf{S}_s $ is defined as
\begin{align} \label{eqn:sensing-sinr}
	\mathsf{SINR}_\mathrm{sen}^{t,s} = \frac
	{ 
	\left| \mathbf{v}_s^\mathrm{H} \mathbf{Z}_\mathrm{rx}^\mathrm{H} \mathbf{G}_t \mathbf{Z}_\mathrm{tx} \mathbf{w}_s \cdot \tau_{t,s} \right|^2 
	}
	{ 
	\left|\mathbf{v}_s^\mathrm{H} \mathbf{Z}_\mathrm{rx}^\mathrm{H} \mathbf{R} \mathbf{Z}_\mathrm{tx} \mathbf{w}_s \cdot \tau_{t,s} \right|^2 + \sigma_\mathrm{sen}^2 \left\| \mathbf{Z}_\mathrm{rx} \mathbf{v}_s \right\|_2^2 
	}.
\end{align}

Additionally, we include constraint 
\begin{align*}
	\mathrm{C}_{21}: \min_{\psi_t \in \Psi_t} \min_{\theta_t \in \Theta_t} \min_{\mathbf{R} \in \mathcal{R}}  \mathsf{SINR}_\mathrm{sen}^{t,s} \geq \Lambda_{\mathrm{sinr},t} \cdot \tau_{t,s}, \forall t \in \mathcal{T},
\end{align*}
which ensures that the sensing \gls{SINR} in every allocated timeslot for each target remains above a predefined threshold $ \Lambda_{\mathrm{sinr},t} $,  even under uncertainties in the \gls{RC}, \gls{AOA}, and \gls{RSI}. Furthermore, each target is allocated $ S_{\mathrm{sen},t} $ timeslots, which dictates the dwell time and is enforced through
\begin{align*}
	\textstyle
	\mathrm{C}_{22}: \sum_{s \in \mathcal{S}} \tau_{t,s} = S_{\mathrm{sen},t}, \forall t \in \mathcal{T}.
\end{align*}

\begin{remark} \label{rmk:clarification-doppler-resolution}
	Sufficient cumulative SINR for accurate Doppler estimation can be achieved by jointly tuning the dwell time $ S_{\mathrm{sen},t} $ and the sensing SINR threshold $ \Lambda_{\mathrm{sinr},t} $. These parameters can be traded off to meet a required Doppler resolution, with longer dwell times compensating for lower SINR thresholds and vice versa.
\end{remark}

\subsection{Objective function}

We define the multi-objective function 
\begin{align} 
	  \mathbf{f} \left( \boldsymbol{\Omega} \right) 
	  \triangleq 
	  \begin{bmatrix} 
	  f_1 \left( \boldsymbol{\Omega} \right) 
	  \\  
	  f_2 \left( \boldsymbol{\Omega} \right) 
	  \end{bmatrix},
\end{align}
where $ \boldsymbol{\Omega} $ represents the set of all decision variables. Functions $ f_1 \left( \boldsymbol{\Omega} \right) $ and $ f_2 \left( \boldsymbol{\Omega} \right) $ are defined as follows
\begin{align*} 
	f_1 \left( \boldsymbol{\Omega} \right) \triangleq & S_\mathrm{dur} \textstyle \sum_{s \in \mathcal{S}} \left\| \mathbf{w}_s \right\|_2^2 + S_\mathrm{dur} \sum_{s \in \mathcal{S}} \left\| \mathbf{v}_s \right\|_2^2,
	\\ 
	f_2 \left( \boldsymbol{\Omega} \right) \triangleq & \textstyle S_\mathrm{dur} \sum_{s \in \mathcal{S}} \omega_{s} \cdot \gamma_s.
\end{align*}




Function $ f_1 \left( \boldsymbol{\Omega} \right) $ represents the total energy consumption of the \gls{BS} for both transmit and receive beamforming over all allocated timeslots, where $ S_\mathrm{dur} $ denotes the duration of each timeslot. In particular, minimizing this function helps reduce the carbon footprint while also limiting interference to neighboring networks.

Function $ f_2(\boldsymbol{\Omega}) $ represents the total consumption of time resources, modeled as a weighted sum of timeslots. Minimizing $ f_2(\boldsymbol{\Omega}) $ helps reduce the number of active timeslots, fostering timeslot reuse by pairing users and targets. Here, weights $ \omega_s $ are carefully chosen to mitigate sparsity in timeslot utilization, promoting a more cohesive allocation that eliminates idle timeslots between active ones. This ensures contiguous timeslot usage, allowing subsequent allocations to start earlier. In particular, \textbf{Lemma~\ref{lem:lemma-weights}} presents a family of weights designed to enhance reuse and compactness in timeslot allocation.

\begin{lemma} \label{lem:lemma-weights}
	A set of weights that promotes compact timeslot allocation in $ f_2 \left( \boldsymbol{\Omega} \right) $ is given by $ {\omega}_{s} = \Delta_0 + (s-1) \cdot \Delta_\omega $, $ \forall s \in \mathcal{S} $, where $ \Delta_\omega > 0 $.
\end{lemma}
\begin{proof}
Please, refer to \textbf{Appendix \ref{app:proof-lemma-weights}}.
\end{proof}
 
Minimizing energy consumption alone is insufficient, as it does not account for the duration over which the energy is expended. Conversely, minimizing consumption of time resources alone overlooks the amount of power used. Hence, jointly considering both aspects is essential to achieve true resource economy and high efficiency.

\subsection{Problem formulation}
We formulate the \gls{RRM} design as
\begin{align*} 
	\mathcal{P} \left( \boldsymbol{\Omega} \right): & \underset{ \boldsymbol{\Omega} }{\mathrm{~minimize}}
	& \mathbf{f} \left( \boldsymbol{\Omega} \right) ~~~ \mathrm{s.t.} ~~~ {\mathrm{C}}_{1} - {\mathrm{C}}_{22}.
\end{align*}

Problem $ \mathcal{P} \left( \boldsymbol{\Omega} \right) $ is challenging to solve due to the presence of couplings between variables and the semi-infinite number of constraints arising from imperfect knowledge of \gls{CSI}, \gls{RC}, \gls{AOD}, and \gls{RSI}. Specifically, $ \mathcal{P} \left( \boldsymbol{\Omega} \right) $ is a nonconvex \gls{MINLP}, for which no generic solution methods exist.


\section{Proposed Optimal Approach} \label{sec:proposed-approach}

This section presents an equivalent transformation of problem $ \mathcal{P} \left( \boldsymbol{\Omega} \right) $ that preserves optimality. The complicating terms within $ \mathcal{P} \left( \boldsymbol{\Omega} \right) $ are reformulated into more tractable expressions, revealing that the problem is convex, specifically, a \gls{MISDP} that can be solved to global optimality. For clarity, \cref{fig:proposed-approach} outlines the steps involved in simplifying the challenging expressions. This process results in a transformed problem $ \mathcal{P}' \left( \boldsymbol{\Omega}' \right) $, where $ \boldsymbol{\Omega}' $ denotes the updated set of decision variables. 

In the following, we delve into each of the steps that facilitate this transformation.

\subsection{Addressing the logical couplings induced by $ \mathrm{C}_{3} $}
We introduce \textbf{\cref{thm:proposition-1}} to cope with the logical operator $ \lor $ appearing in constraint $ \mathrm{C}_{3} $, which links variables $ \kappa_s $ and $ \zeta_s $. We decouple the binary variables and reformulate $ \mathrm{C}_{3} $ into an equivalent set of linear inequalities.
\begin{proposition} \label{thm:proposition-1}
	Constraint $ \mathrm{C}_{3} $ can be equivalently rewritten as constraints $ \mathrm{D}_{1} $, $ \mathrm{D}_{2} $, $ \mathrm{D}_{3} $, and $ \mathrm{D}_{4} $, shown below
	\begin{align} \nonumber
		\mathrm{C}_{3} \Leftrightarrow
			\begin{cases}
				   	\mathrm{D}_{1}: \gamma_s \leq \kappa_s + \zeta_s, \forall s \in \mathcal{S}, 
				   	\\	
				   	\mathrm{D}_{2}: \gamma_s \geq \kappa_s, \forall s \in \mathcal{S},   
				   	\\	
				   	\mathrm{D}_{3}: \gamma_s \geq \zeta_s, \forall s \in \mathcal{S},
				   	\\	
				   	\mathrm{D}_{4}: \gamma_s \leq 1, \forall s \in \mathcal{S}. 
			\end{cases}
	\end{align}
\end{proposition} 
\begin{proof} 
	For the detailed derivation, refer to \textbf{Appendix \ref{app:proof-proposition-1}}.
\end{proof}

\subsection{Addressing the couplings between variables and the infinite constraints in $ \mathrm{C}_{11} $, $ \mathrm{C}_{15} $, and $ \mathrm{C}_{16} $}

We introduce \textbf{\cref{thm:proposition-2}} to address the challenges posed by $ \mathrm{C}_{11} $, $ \mathrm{C}_{15} $, and $ \mathrm{C}_{16} $, stemming from variable coupling and the infinite constraint set due to imperfect \gls{CSI}. Specifically, we decouple variables $ \mathbf{w}_{s} $ and $ \mu_{u,s} $ in $ \mathrm{C}_{16} $, by analyzing the expressions in which the coupling appears and deriving equivalent reformulations that preserve the original solution space. This step is critical to our approach, as it allows us to restructure the expressions in a way that eliminates couplings without introducing additional variables. Furthermore, to manage the infinite constraint set induced by $ \mathrm{C}_{15} $, we employ the \emph{S-Procedure}, outlined in \textbf{Lemma~\ref{lem:s-procedure}}, which transforms an infinite constraint set into a finite and computationally tractable form. The overall procedure reformulates $ \mathrm{C}_{11} $, $ \mathrm{C}_{15} $, and $ \mathrm{C}_{16} $ into an equivalent set of linear inequalities, while ensuring that the solution space of these constraints remains unchanged.

\begin{figure}[!t]
	\centering
	\includegraphics[width=0.77\columnwidth]{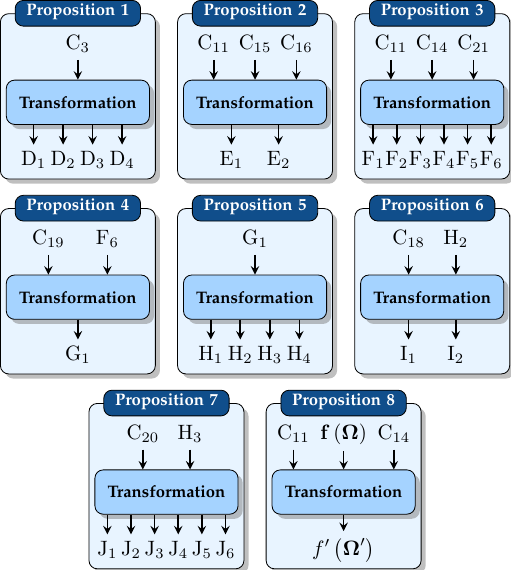}
	\caption{Diagram illustrating the proposed approach.}
	\label{fig:proposed-approach}
\end{figure}


\begin{figure*}[!b]
	\hrulefill
		\begin{align} \nonumber
			\mathrm{C}_{11}, \mathrm{C}_{14}, \mathrm{C}_{21} \Leftrightarrow
				\begin{cases}
					   	\mathrm{F}_{1}: \chi_{b,s} \geq \pi_{b,c,t,s}, \forall b \in \mathcal{L}_\mathrm{tx}, c \in \mathcal{L}_\mathrm{rx}, t \in \mathcal{T}, s \in \mathcal{S}, 
					   	\\
					   	\mathrm{F}_{2}: \rho_{c,s} \geq \pi_{b,c,t,s}, \forall b \in \mathcal{L}_\mathrm{tx}, c \in \mathcal{L}_\mathrm{rx}, t \in \mathcal{T}, s \in \mathcal{S},
					   	\\	
					   	\mathrm{F}_{3}: \tau_{t,s} \geq \pi_{b,c,t,s}, \forall b \in \mathcal{L}_\mathrm{tx}, c \in \mathcal{L}_\mathrm{rx}, t \in \mathcal{T}, s \in \mathcal{S},
					   	\\
					   	\mathrm{F}_{4}: 2 + \pi_{b,c,t,s} \geq \chi_{b,s} + \rho_{c,s} + \tau_{t,s}, \forall b \in \mathcal{L}_\mathrm{tx}, c \in \mathcal{L}_\mathrm{rx}, t \in \mathcal{T}, s \in \mathcal{S},
					   	\\
					   	\mathrm{F}_{5}: \pi_{b,c,t,s} \in \left[ 0, 1 \right], \forall b \in \mathcal{L}_\mathrm{tx}, c \in \mathcal{L}_\mathrm{rx}, t \in \mathcal{T}, s \in \mathcal{S},
					   	\\
					   	\mathrm{F}_{6}: \min_{\psi_t \in \Psi_t} \min_{\theta_t \in \Theta_t} \ell_{t,s} (\boldsymbol{\Omega}')
					   	\geq \Lambda_{\mathrm{sinr},t} \cdot \tau_{t,s}, \forall t \in \mathcal{T}.
				\end{cases}
		\end{align}
\end{figure*}

\begin{proposition} \label{thm:proposition-2}
	Constraints $ \mathrm{C}_{11} $, $ \mathrm{C}_{15} $, and $ \mathrm{C}_{16} $ can be equivalently rewritten as constraints $ \mathrm{E}_{1} $ and $ \mathrm{E}_{2} $, shown below, 
 	\begin{align} \nonumber
	\mathrm{C_{11}}, \mathrm{C_{15}}, \mathrm{C_{16}} =
 		\begin{cases}
 			   	\mathrm{E}_{1}: \left[ 
	 			   					\begin{matrix}
	 			   					   	\mathbf{E}_{u,s}^\mathrm{a}
	 			   					   	& 
	 			   					   	\mathbf{e}_{u,s}^\mathrm{b}
	 			   					   	\\
	 			   					    \mathbf{e}_{u,s}^\mathrm{c}
	 			   					    & 
	 			   					    e_{u,s}^\mathrm{d}
	 			   				  	\end{matrix} 
 			   					\right] 
 					   		\succcurlyeq \mathbf{0}, \forall u \in \mathcal{U}, s \in \mathcal{S},
 				\\
 				\mathrm{E}_{2}: \alpha_u \geq 0, \forall u \in \mathcal{U},
 		\end{cases}
	\end{align}
	where $ \mathbf{E}_{u,s}^\mathrm{a} = \alpha_u \mathbf{I} - \mathbf{L}_s $, $ \mathbf{e}_{u,s}^\mathrm{b} = - \mathbf{L}_s \widebar{\mathbf{h}}_u $, $ \mathbf{e}_{u,s}^\mathrm{c} = - \widebar{\mathbf{h}}_u^\mathrm{H} \mathbf{L}_s $, $ e_{u,s}^\mathrm{d} =  -\alpha_u \cdot \epsilon_\mathrm{CSI}^2 - \widebar{\mathbf{h}}_u^\mathrm{H} \mathbf{L}_s \widebar{\mathbf{h}}_u - \Upsilon_{\mathrm{snr},u} \cdot \mu_{u,s} $, $ \mathbf{L}_s = - \sum_{b \in \mathcal{L}_\mathrm{tx}} \frac{1}{\sigma_\mathrm{com}^2} \mathbf{Z}_\mathrm{tx} \mathbf{b}_b \mathbf{b}_b^\mathrm{H} \mathbf{Z}_\mathrm{tx}^\mathrm{H} \chi_{b,s} $, and $ \alpha_u $ are newly introduced variables. 
\end{proposition} 
\begin{proof}
	For the detailed derivation, refer to \textbf{Appendix \ref{app:proof-proposition-2}}.
\end{proof}

\begin{lemma}[S-Procedure] \label{lem:s-procedure}
	Let $ h_i(\mathbf{x}) \triangleq \mathbf{x}^\mathrm{H} \mathbf{A}_i \mathbf{x} + \mathbf{b}_i^\mathrm{H} \mathbf{x} + \mathbf{x}^\mathrm{H} \mathbf{b}_i + c_i $, for $ i = 1, 2 $, be quadratic functions, where $ \mathbf{A}_i \in \mathbb{C}^{N \times N} $ is Hermitian, $ \mathbf{b}_i \in \mathbb{C}^{N \times 1} $, and $ c_i \in \mathbb{R} $. Thus, the following two statements are equivalent: 
 
 	i) If $ \mathbf{x} $ satisfies  $ h_1(\mathbf{x}) \leq 0 $, then $ h_2(\mathbf{x}) \leq 0 $ also holds.
 
 	ii) There exists $ \alpha \geq 0 $ such that 
 	$ \alpha 	\left[ 
   					\begin{smallmatrix}
   					   	\mathbf{A}_1  	  & \mathbf{b}_1 \\
 	    			\mathbf{b}_1^\mathrm{H}        & c_1
   				  	\end{smallmatrix} 
   			\right]  
   			- 
   			\left[ 
  					\begin{smallmatrix}
 							\mathbf{A}_2   & \mathbf{b}_2 \\
 			    	\mathbf{b}_2^\mathrm{H}         & c_2
 						\end{smallmatrix} 
 			\right]  \succcurlyeq \mathbf{0} 
 	$,
 	where $ \alpha $ is a newly introduced variable.
\end{lemma}

\subsection{Addressing the couplings between variables in $ \mathrm{C}_{11} $, $ \mathrm{C}_{14} $, and $ \mathrm{C}_{21} $}

We introduce \textbf{\cref{thm:proposition-3}} to deal with the multiplicative coupling between variables $ \mathbf{w}_s $, $ \mathbf{v}_s $, and $ \tau_{t,s} $, posed by $ \mathrm{C}_{11} $, $ \mathrm{C}_{14} $, and $ \mathrm{C}_{21} $. We reformulate these constraints into an equivalent set of equations and inequalities without altering the original solution space.

\begin{proposition} \label{thm:proposition-3}
	Constraints $ \mathrm{C}_{11} $, $ \mathrm{C}_{14} $, and $ \mathrm{C}_{21} $ can be equivalently written as constraints $ \mathrm{F}_{1} $, $ \mathrm{F}_{2} $, $ \mathrm{F}_{3} $, $ \mathrm{F}_{4} $, $ \mathrm{F}_{5} $, and $ \mathrm{F}_{6} $ (see bottom of this page), where $ \ell_{t,s} (\boldsymbol{\Omega}') \triangleq  \min_{\mathbf{R} \in \mathcal{R}} $ $\frac{ \left| \sum_{c \in \mathcal{L}_\mathrm{rx}} \sum_{b \in \mathcal{L}_\mathrm{tx}} \mathbf{c}_c^\mathrm{H} \mathbf{Z}_\mathrm{rx}^\mathrm{H} \mathbf{A} \left( \theta_t \right) \mathbf{Z}_\mathrm{tx} \mathbf{b}_b \cdot \pi_{b,c,t,s} \right|^2 } { \left| \sum_{c \in \mathcal{L}_\mathrm{rx}} \sum_{b \in \mathcal{L}_\mathrm{tx}}  \mathbf{c}_c^\mathrm{H} \mathbf{Z}_\mathrm{rx}^\mathrm{H} \mathbf{R} \mathbf{Z}_\mathrm{tx} \mathbf{b}_b \cdot \pi_{b,c,t,s} \right|^2 + \sigma_\mathrm{sen}^2 \sum_{c \in \mathcal{L}_\mathrm{rx}} \left\| \mathbf{Z}_\mathrm{rx} \mathbf{c}_c \right\|_2^2 \rho_{c,s} } $ and $ \pi_{b,c,t,s} $ are newly introduced variables.
\end{proposition} 
\begin{proof} 
	For the detailed derivation, refer to \textbf{Appendix \ref{app:proof-proposition-3}}.
\end{proof}

\subsection{Addressing the infinite constraints in $ \mathrm{C}_{19} $ and $ \mathrm{F}_{6} $}

We present \textbf{\cref{thm:proposition-4}} to address the infinite constraint set arising from imperfect \gls{RC}, which affects $ \mathrm{C}_{19} $ and $ \mathrm{F}_{6} $. Specifically, we rearrange these constraints to disentangle the effects of imperfect \gls{RC} from those caused by imperfect \gls{RSI} and \gls{AOD}. After isolating the effect of imperfect \gls{RC}, we observe that it effectively scales and magnifies the nominal \gls{SINR} threshold, leading to a more demanding sensing performance requirement. Notably, the scaled \gls{SINR} threshold admits a closed-form expression, facilitating further analysis and optimization. Overall, this transformation converts the infinite constraint set induced by imperfect \gls{RC} into a more tractable form, leading to an equivalent set of inequalities.
\begin{proposition} \label{thm:proposition-4}
	Constraints $ \mathrm{C}_{19} $ and $ \mathrm{F}_{6} $ can be equivalently rewritten as constraints $ \mathrm{G}_{1} $, shown below, 
		\begin{align} \nonumber
			\mathrm{C}_{19}, \mathrm{F}_{6} \Leftrightarrow \mathrm{G}_{1}: \min_{\theta_t \in \Theta_t} \ell_{t,s} (\boldsymbol{\Omega}') \geq \widebar{\Lambda}_{\mathrm{sinr},t} \cdot \tau_{t,s}, \forall t \in \mathcal{T}, s \in \mathcal{S},
		\end{align}
	where $ \widebar{\Lambda}_{\mathrm{sinr},t} = \frac{\Lambda_{\mathrm{sinr},t}}{\left( \left| \widebar{\psi}_t  \right| - \epsilon_{\mathrm{RC}} \right)^2 } $ and $ \ell_{t,s} (\boldsymbol{\Omega}') \triangleq  \min_{\mathbf{R} \in \mathcal{R}} $ $ \frac{ 
		\left| \sum_{c \in \mathcal{L}_\mathrm{rx}} \sum_{b \in \mathcal{L}_\mathrm{tx}} \mathbf{c}_c^\mathrm{H} \mathbf{Z}_\mathrm{rx}^\mathrm{H}  \mathbf{A} \left( \theta_t \right) \mathbf{Z}_\mathrm{tx} \mathbf{b}_b \cdot \pi_{b,c,t,s} \right|^2 
	} { 
	\left| \sum_{c \in \mathcal{L}_\mathrm{rx}} \sum_{b \in \mathcal{L}_\mathrm{tx}}  \mathbf{c}_c^\mathrm{H} \mathbf{Z}_\mathrm{rx}^\mathrm{H} \mathbf{R} \mathbf{Z}_\mathrm{tx} \mathbf{b}_b \cdot \pi_{b,c,t,s} \right|^2 + \sigma_\mathrm{sen}^2 \sum_{c \in \mathcal{L}_\mathrm{rx}} \left\| \mathbf{Z}_\mathrm{rx}  \mathbf{c}_c \right\|_2^2 \rho_{c,s} 
	} $.
\end{proposition} 
\begin{proof} 
	For the detailed derivation, refer to \textbf{Appendix \ref{app:proof-proposition-4}}.
\end{proof}

\subsection{Addressing the rational nature of the SINR in $ \mathrm{G}_{1} $}

We introduce \textbf{\cref{thm:proposition-6}} to address the rational nature of $ \mathrm{G}_{1} $. This is accomplished by introducing additional variables that separate the numerator and denominator of the \glspl{SINR}, leading to an equivalent set of inequalities that preserve the original solution space of $ \mathrm{G}_{1} $. Importantly, this transformation facilitates the decoupling of effects arising from imperfect \gls{AOD} and imperfect \gls{RSI}, allowing us to address them independently in subsequent steps.

\begin{proposition} \label{thm:proposition-6}
	Constraint $ \mathrm{G}_{1} $ can be equivalently rewritten as constraints $ \mathrm{H}_{1} $, $ \mathrm{H}_{2} $, $ \mathrm{H}_{3} $, and $ \mathrm{H}_{4} $, shown below,
		\begin{align} \nonumber
			\mathrm{G}_{1} \Leftrightarrow
				\begin{cases}
					   	\mathrm{H}_{1}: z_{t,s} \geq 0, \forall t \in \mathcal{T}, s \in \mathcal{S},
					   	\\
					   	\mathrm{H}_{2}:  \hat{\ell}_{t,s} (\boldsymbol{\Omega}') \geq z_{t,s}, \forall t \in \mathcal{T}, s \in \mathcal{S},
					   	\\
					   	\mathrm{H}_{3}: \check{\ell}_{t,s} (\boldsymbol{\Omega}') \leq z_{t,s}, \forall t \in \mathcal{T}, s \in \mathcal{S},  
					   	\\
		   				\mathrm{H}_{4}: z_{t,s} \leq \tau_{t,s} \cdot \ddot{M}_{t}^\mathrm{UB}, \forall t \in \mathcal{T}, s \in \mathcal{S},
				\end{cases}
		\end{align}
		where $ \hat{\ell}_{t,s} (\boldsymbol{\Omega}') \triangleq \frac{1}{\sigma_\mathrm{sen}^2} \min_{\theta_t \in \Theta_t} \big| \sum_{c \in \mathcal{L}_\mathrm{rx}} \sum_{b \in \mathcal{L}_\mathrm{tx}}  \mathbf{c}_c^\mathrm{H} \mathbf{Z}_\mathrm{rx}^\mathrm{H} \mathbf{A} \left( \theta_t \right) $ $ \mathbf{Z}_\mathrm{tx} \mathbf{b}_b \cdot \pi_{b,c,t,s} \big|^2 $, $ \check{\ell}_{t,s} (\boldsymbol{\Omega}') \triangleq \frac{1}{\sigma_\mathrm{sen}^2} \widebar{\Lambda}_{\mathrm{sinr},t} \cdot \tau_{t,s} \cdot \max_{\mathbf{R} \in \mathcal{R}}  $ $ \big| \sum_{c \in \mathcal{L}_\mathrm{rx}} \sum_{b \in \mathcal{L}_\mathrm{tx}}  \mathbf{c}_c^\mathrm{H} \mathbf{Z}_\mathrm{rx}^\mathrm{H} \mathbf{R} \mathbf{Z}_\mathrm{tx} \mathbf{b}_b \cdot \pi_{b,c,t,s} \big|^2 + \widebar{\Lambda}_{\mathrm{sinr},t} \cdot \tau_{t,s} \sum_{c \in \mathcal{L}_\mathrm{rx}} $ $  \left\| \mathbf{Z}_\mathrm{rx} \mathbf{c}_c \right\|_2^2 \rho_{c,s} $, and $ \ddot{M}_{t}^\mathrm{UB} = \max_{b \in \mathcal{L}_\mathrm{tx}} \max_{c \in \mathcal{L}_\mathrm{rx}} \frac{\big| \mathbf{c}_c^\mathrm{H} \mathbf{Z}_\mathrm{rx}^\mathrm{H} \mathbf{A} \left( \theta_t \right) \mathbf{Z}_\mathrm{tx} \mathbf{b}_b \big|^2}{\sigma_\mathrm{sen}^2}$.
\end{proposition} 
\begin{proof} 
	For the detailed derivation, refer to \textbf{Appendix \ref{app:proof-proposition-5}}.
\end{proof}

\subsection{Addressing the infinite constraints in $ \mathrm{C}_{18} $ and $ \mathrm{H}_{2} $}

We introduce \textbf{\cref{thm:proposition-7}} to address the infinite constraint set arising from imperfect \gls{AOD} in $ \mathrm{C}_{18} $ and $ \mathrm{H}_{2} $. Specifically, we first restructure $ \mathrm{H}_{2} $ and then employ the \emph{S-Procedure}, as detailed in \textbf{Lemma~\ref{lem:s-procedure}}, to transform the infinite constraint set into a finite collection of linear inequalities. The overall procedure ensures that the original solution space is preserved.

\begin{proposition} \label{thm:proposition-7}
	Constraints $ \mathrm{C}_{18} $ and $ \mathrm{H}_{2} $ can be equivalently rewritten as constraints $ \mathrm{I}_{1} $ and $ \mathrm{I}_{2} $, shown below,
	\begin{align*}
	 	\mathrm{C}_{18}, \mathrm{H_{2}} \Leftrightarrow
	 		\begin{cases}
				   	  	 			   	\mathrm{I}_{1}: \left[ 
				   			 			   					\begin{matrix}
				   			 			   					   	i_{t,s}^\mathrm{a}
				   			 			   					   	& 
				   			 			   					   	i_{t,s}^\mathrm{b}
				   			 			   					   	\\
				   			 			   					    i_{t,s}^\mathrm{c}
				   			 			   					    & 
				   			 			   					    i_{t,s}^\mathrm{d}
				   			 			   				  	\end{matrix} 
				   		 			   					\right] 
				   		 					   		\succcurlyeq \mathbf{0}, \forall t \in \mathcal{T}, s \in \mathcal{S},
				   	\\	 					   		
					  \mathrm{I}_{2}: \xi_{t} \geq 0, \forall t \in \mathcal{T}, 
			\end{cases}				
	\end{align*}
	where $ i_{t,s}^\mathrm{a} = \xi_{t} - \widetilde{\mathbf{a}} \left( \widebar{\theta}_t \right)^\mathrm{H} \mathbf{F}_{t,s} \widetilde{\mathbf{a}} \left( \widebar{\theta}_t \right) $, 
	$ i_{t,s}^\mathrm{b} = - \widetilde{\mathbf{a}} \left( \widebar{\theta}_t \right)^\mathrm{H} \mathbf{F}_{t,s} \mathbf{a} \left( \widebar{\theta}_t \right) $, 
	$ i_{t,s}^\mathrm{c} = - \mathbf{a} \left( \widebar{\theta}_t \right)^\mathrm{H} \mathbf{F}_{t,s} \widetilde{\mathbf{a}} \left( \widebar{\theta}_t \right) $, 
	$ i_{t,s}^\mathrm{d} = - \xi_{t} \cdot \epsilon_\mathrm{AOD}^2 - \mathbf{a} \left( \widebar{\theta}_t \right)^\mathrm{H} \mathbf{F}_{t,s} \mathbf{a} \left( \widebar{\theta}_t \right) - z_{t,s} $, $ \mathbf{a} \left( \widebar{\theta}_t \right) = \mathrm{vec} \left( \mathbf{A} \left( \widebar{\theta}_t \right) \right) $, $ \widetilde{\mathbf{a}} \left( \widebar{\theta}_t \right) = \mathrm{vec} \left( \widetilde{\mathbf{A}} \left( \widebar{\theta}_t \right) \right) $, $ \mathbf{F}_{t,s} = -  \sum_{c \in \mathcal{L}_\mathrm{rx}} \sum_{b \in \mathcal{L}_\mathrm{tx}} \frac{1}{\sigma_\mathrm{sen}^2} \mathbf{d}_{b,c} \mathbf{d}_{b,c}^\mathrm{H} \pi_{b,c,t,s} $, $ \mathbf{d}_{b,c} = \left( \mathbf{Z}_\mathrm{tx}^* \mathbf{b}_b^* \right) $ $ \otimes \left(  \mathbf{Z}_\mathrm{rx} \mathbf{c}_c \right) $, and $ \xi_t $ are newly introduced variables.
\end{proposition} 
\begin{proof} 
	For the detailed derivation, refer to \textbf{Appendix \ref{app:proof-proposition-6}}.
\end{proof}

\subsection{Addressing the infinite constraints in $ \mathrm{C}_{20} $ and $ \mathrm{H}_{3} $}
 
We introduce \textbf{\cref{thm:proposition-8}} to handle the couplings between variables $ \rho_{c,s} $ and $ \tau_{t,s} $ as well as the infinite constraint set resulting from imperfect \gls{RSI} affecting $ \mathrm{C}_{20} $ and $ \mathrm{H}_{3} $. First, we apply propositional calculus to eliminate the binary couplings, simplifying the constraint structure. Subsequently, we leverage the \emph{S-Procedure}, as detailed in \textbf{Lemma~\ref{lem:s-procedure}}, to convert the infinite constraint set into a finite set of linear inequalities. The overall procedure ensures that the original solution space remains unaltered.
\begin{proposition} \label{thm:proposition-8}
	Constraints $ \mathrm{C}_{20} $ and $ \mathrm{H}_{3} $ can be equivalently rewritten as constraints $ \mathrm{J}_{1} $, $ \mathrm{J}_{2} $, $ \mathrm{J}_{3} $, $ \mathrm{J}_{4} $, $ \mathrm{J}_{5} $, and $ \mathrm{J}_{6} $, shown below,
		\begin{align*}
			\begin{matrix}
			\mathrm{C}_{20} \\ \mathrm{H_{3}} 
			\end{matrix}
		 	\Leftrightarrow
		 		\begin{cases}
		 						\mathrm{J}_{1}: \delta_{c,t,s} \leq \rho_{c,s}, \forall c \in \mathcal{L}_\mathrm{rx}, t \in \mathcal{T}, s \in \mathcal{S},
		 						\\
		 						\mathrm{J}_{2}: \delta_{c,t,s} \leq \tau_{t,s}, \forall c \in \mathcal{L}_\mathrm{rx}, t \in \mathcal{T}, s \in \mathcal{S},
		 						\\
		 						\mathrm{J}_{3}: \delta_{c,t,s} \geq \rho_{c,s} + \tau_{t,s} - 1, \forall c \in \mathcal{L}_\mathrm{rx}, t \in \mathcal{T}, s \in \mathcal{S},
		 						\\
		 						\mathrm{J}_{4}: \delta_{c,t,s} \in \left[ 0, 1 \right], \forall c \in \mathcal{L}_\mathrm{rx}, t \in \mathcal{T}, s \in \mathcal{S},
		 						\\
			   	  			   	\mathrm{J}_{5}: \left[ 
			   			 			   					\begin{matrix}
			   			 			   					   	\mathbf{J}_{t,s}^\mathrm{a}
			   			 			   					   	& 
			   			 			   					   	\mathbf{j}_{t,s}^\mathrm{b}
			   			 			   					   	\\
			   			 			   					    \mathbf{j}_{t,s}^\mathrm{c}
			   			 			   					    & 
			   			 			   					    j_{t,s}^\mathrm{d}
			   			 			   				  	\end{matrix} 
			   		 			   					\right] 
			   		 					   		\succcurlyeq \mathbf{0}, 
														\begin{matrix}
			   			 			   					   	\forall t \in \mathcal{T}, s \in \mathcal{S},
			   			 			   				  	\end{matrix} 
			  	\\ 		 					   		
				\mathrm{J}_{6}: \iota \geq 0, 
				\end{cases}				
		\end{align*}
		where
		$ \mathbf{J}_{t,s}^\mathrm{a} = \iota \mathbf{I} + \upsilon^2 \widebar{\Lambda}_{\mathrm{sinr},t} \mathbf{F}_{t,s} $, 
		$ \mathbf{j}_{t,s}^\mathrm{b} = \upsilon^2 \widebar{\Lambda}_{\mathrm{sinr},t} \mathbf{F}_{t,s} \widebar{\mathbf{r}} $,
		$ \mathbf{j}_{t,s}^\mathrm{c} = \upsilon^2 \widebar{\Lambda}_{\mathrm{sinr},t} \widebar{\mathbf{r}}^\mathrm{H} \mathbf{F}_{t,s}  $,
		$ j_{t,s}^\mathrm{d} = - \iota \cdot \epsilon_\mathrm{RSI}^2 + \upsilon^2 \widebar{\Lambda}_{\mathrm{sinr},t} \widebar{\mathbf{r}}^\mathrm{H} \mathbf{F}_{t,s} \widebar{\mathbf{r}} - \widebar{\Lambda}_{\mathrm{sinr},t} \sum_{c \in \mathcal{L}_\mathrm{rx}} \left\| \mathbf{Z}_\mathrm{rx} \mathbf{c}_c \right\|_2^2 \delta_{c,t,s} + z_{t,s} $, 
		$ \mathbf{F}_{t,s} = - \sum_{c \in \mathcal{L}_\mathrm{rx}} \sum_{b \in \mathcal{L}_\mathrm{tx}} \frac{1}{\sigma_\mathrm{sen}^2} \mathbf{d}_{b,c} \mathbf{d}_{b,c}^\mathrm{H} \pi_{b,c,t,s} $, and $ \mathbf{d}_{b,c} = \left( \mathbf{Z}_\mathrm{tx}^* \mathbf{b}_b^* \right) \otimes \left(  \mathbf{Z}_\mathrm{rx} \mathbf{c}_c \right) $, while $ \delta_{c,t,s} $ and $ \iota $ are newly introduced variables.
\end{proposition} 
\begin{proof} 
	For the detailed derivation, refer to \textbf{Appendix \ref{app:proof-proposition-7}}.
\end{proof}

\subsection{Addressing the vector nature of $ \mathbf{f} \left( \boldsymbol{\Omega} \right) $ and its dependence on $ \mathrm{C}_{11} $ and $ \mathrm{C}_{14} $}

We introduce \textbf{\cref{thm:proposition-9}} to first reformulate the vector-valued function $ \mathbf{f} \left( \boldsymbol{\Omega} \right) $ into a scalarized form using carefully designed weights $ \eta_1 $ and $ \eta_2 $. These weights are constructed to ensure that $ \eta_1 f_1 \left( \boldsymbol{\Omega} \right) $ and $ \eta_2 f_2 \left( \boldsymbol{\Omega} \right) $ map to distinct, non-overlapping intervals, with the explicit intent of prioritizing the minimization of time over energy consumption as a special case\footnote{Prioritization of energy minimization can be addressed by adjusting the weights $ \eta_1 $ and $ \eta_3 $. While this case is not the primary focus of our study, its impact is examined in \textbf{Appendix~\ref{app:additional-scenario-viii}}.}. Subsequently, we analyze the quadratic terms in $ f_1 \left( \boldsymbol{\Omega} \right) $ which introduce complexity through their quadratic dependence on both $ \mathbf{w}_s $ and $ \mathbf{v}_s $. By exploiting the binary-variable dependence of $ \mathbf{w}_s $ and $ \mathbf{v}_s $, as specified in constraints $ \mathrm{C}_{11} $ and $ \mathrm{C}_{14} $, we transform $ f_1 \left( \boldsymbol{\Omega} \right) $ into an equivalent linear expression that is more tractable.
\begin{proposition} \label{thm:proposition-9}
	Vector function $ \mathbf{f} \left( \boldsymbol{\Omega} \right) $ and constraints $ \mathrm{C}_{11} $, $ \mathrm{C}_{14} $ is rewritten as $ f' \left( \boldsymbol{\Omega} \right) $, shown below,
	\begin{align*}
			\mathbf{f} \left( \boldsymbol{\Omega} \right), \mathrm{C}_{11}, \mathrm{C}_{14} \Leftrightarrow f' \left( \boldsymbol{\Omega}' \right) \triangleq \eta_1  f'_1 \left( \boldsymbol{\Omega}' \right) + \eta_2 f_2' \left( \boldsymbol{\Omega}' \right), 
	\end{align*}
	where the coefficients $ \eta_1 $ and $ \eta_2 $ are weights that determine the relative priorities of $ f_1' \left( \boldsymbol{\Omega}' \right) $ and $ f_2' \left( \boldsymbol{\Omega}' \right) $, respectively. Here, $ f_1' \left( \boldsymbol{\Omega}' \right) = S_\mathrm{dur} \sum_{s \in \mathcal{S}} \big( \sum_{b \in \mathcal{L}_\mathrm{tx}}  \left\| \mathbf{b}_b \right\|_2^2 \cdot \chi_{b,s} + \sum_{c \in \mathcal{L}_\mathrm{rx}} \left\| \mathbf{c}_c \right\|_2^2 \cdot \rho_{c,s} \big) $ and $ f_2' \left( \boldsymbol{\Omega}' \right) = S_\mathrm{dur} \sum_{s \in \mathcal{S}} \omega_{s} \cdot \gamma_s $.  
\end{proposition} 
\begin{proof} 
	For the detailed derivation, refer to \textbf{Appendix \ref{app:proof-proposition-8}}.
\end{proof}

\subsection{Adding custom cutting planes to reduce binary search burden}
To alleviate the burden of binary search, we introduce a set of problem-specific cutting planes, inspired by the approaches in \cite{cheng2015:joint-discrete-rate-adaptation-downlink-beamforming-using-mixed-integer-conic-programming, abanto2024:radio-resource-management-design-RSMA-optimization-beamforming-user-admission-discrete-continuous-rates-imperfect-sic}, but adapted to the structure of our problem. These tailored constraints, denoted as $ \mathrm{K}_{1} $ and $ \mathrm{K}_{2} $, eliminate infeasible or redundant solutions\footnote{Cutting planes are optional but they are highly effective in accelerating binary search by eliminating infeasible and redundant solutions.}. 

We add constraint $ \mathrm{K}_{1} $ to impose an upper bound $ \dddot{M}_{u}^\mathrm{UB} $ on the numerator of the \gls{SNR}, as shown below,
\begin{align*}  
	& \mathrm{K}_{1}: \left| \mathbf{h}_u^\mathrm{H} \mathbf{Z}_\mathrm{tx} \mathbf{w}_s \right|^2 \leq \dddot{M}_{u}^\mathrm{UB}, \forall u \in \mathcal{U}, s \in \mathcal{S},  
\end{align*}
where $ \dddot{M}_{u}^\mathrm{UB} = \big( \left\| \mathbf{Z}_\mathrm{tx}^\mathrm{H} \widebar{\mathbf{h}}_u \right\|_2 + \epsilon_{\mathrm{CSI}} \left\| \mathbf{Z}_\mathrm{tx} \right\|_\mathrm{F} \big)^2 P_\mathrm{tx}^\mathrm{max} $, and $ P_\mathrm{tx}^\mathrm{max} $ is the maximum transmit power of the \gls{BS}, i.e, $ P_\mathrm{tx}^\mathrm{max} = \max_{b \in \mathcal{L}_\mathrm{tx}} \left\| \mathbf{b}_b \right\|^2_2 $. Here, $ \dddot{M}_{u}^\mathrm{UB} $ is an upper bound on the \gls{LHS} of $ \mathrm{K}_{1} $, obtained by applying the triangle inequality and the Cauchy-Schwarz inequality.

Additionally, we include constraint $ \mathrm{K}_{2} $ to address the inherent redundancy in the temporal ordering of users and targets, as shown below, 
\begin{align*}  
	& \mathrm{K}_{2}: \widebar{\ell}_{s+1} (\boldsymbol{\Omega}') \geq \widebar{\ell}_{s} (\boldsymbol{\Omega}'), \forall s \in \mathcal{S} \setminus \{ S \},  
\end{align*}
where $ \widebar{\ell}_{s} (\boldsymbol{\Omega}') \triangleq \sum_{b \in \mathcal{L}_\mathrm{tx}} \left\| \mathbf{b}_b \right\|_2^2 \cdot \chi_{b,s} + \sum_{c \in \mathcal{L}_\mathrm{rx}} \left\| \mathbf{c}_c \right\|_2^2 \cdot \rho_{c,s} $ represents the allocated power in $ \mathsf{S}_s $. Since the sequence in which users or targets are served over time does not affect the final outcome, it is beneficial to eliminate equivalent permutations.  To this end, we enforce a non-increasing ordering of power allocation across timeslots. By requiring power to decrease or remain constant over time, we prune equivalent configurations that differ only in their ordering, thus significantly reducing the solution space. For the detailed derivation leading to $ \mathrm{K}_{1} $ and $ \mathrm{K}_{2} $, refer to \textbf{Appendix~\ref{app:derivation-cutting-planes}}. 


The transformed problem $ \mathcal{P}' \left( \boldsymbol{\Omega}' \right) $, which constitutes a \gls{MISDP}, is presented below, while the full formulation of $ \mathcal{P}' \left( \boldsymbol{\Omega}' \right) $, including the explicit set of constraints, is provided in \textbf{Appendix \ref{app:transformed-problem}} for further reference.
\begin{align*} 
	\mathcal{P}' \left( \boldsymbol{\Omega}' \right): & \min_{
			\substack{ \boldsymbol{\Omega}' }
	} 
	& & f'\left( \boldsymbol{\Omega}' \right)
	\\
	& ~ \mathrm{s.t.} & & \mathrm{C}_{1}, \mathrm{C}_{2}, \mathrm{C}_{4}, \mathrm{C}_{5}, \mathrm{C}_{6}, \mathrm{C}_{7}, \mathrm{C}_{8}, \mathrm{C}_{9}, \mathrm{C}_{10}, 
	\\
	& & & \mathrm{C}_{12}, \mathrm{C}_{13}, \mathrm{C}_{17}, \mathrm{C}_{22}, \mathrm{D}_{1}, \mathrm{D}_{2}, \mathrm{D}_{3}, \mathrm{D}_{4}, 
	\\
	& & & \mathrm{E}_{1}, \mathrm{E}_{2}, \mathrm{F}_{1}, \mathrm{F}_{2}, \mathrm{F}_{3}, \mathrm{F}_{4}, \mathrm{F}_{5}, \mathrm{H}_{1}, \mathrm{H}_{4},
	\\ 
	& & &  \mathrm{I}_{1}, \mathrm{I}_{2}, \mathrm{J}_{1}, \mathrm{J}_{2}, \mathrm{J}_{3}, \mathrm{J}_{4}, \mathrm{J}_{5}, \mathrm{J}_{6}, \mathrm{K}_{1}, \mathrm{K}_{2}.
\end{align*}

The computational complexity of solving $ \mathcal{P}' \left( \boldsymbol{\Omega}' \right) $ is highly dependent on the specific cuts, relaxations, and built-in heuristics employed by the solvers, in this case, \texttt{MOSEK}, \texttt{CVX}, and \texttt{YALMIP}. While it is not feasible to determine the exact computational complexity due to these factors, we provide an upper bound represented by the worst-case complexity, denoted by  $ \mathcal{C}_\mathrm{complexity} = \mathcal{O} \left(  \sum_{y = 0}^{Y}\frac{X!}{(X-y)!} \left(  L_\mathrm{tx}^{X+y}  L_\mathrm{rx}^{T} \right)  \left( m^2 n^2 + m n^3 + n^6 \right) \right) $, where $ X = \max \left\lbrace U, T \right\rbrace $, $ Y = \min \left\lbrace U, T \right\rbrace $, $ m = US + 2TS $, and $ n = N_\mathrm{tx} N_\mathrm{rx} + N_\mathrm{tx} + 2 $.

\begin{table*}[t!]
	\setlength\tabcolsep{3.8pt} 
	\renewcommand{\arraystretch}{1.8}
	\tiny
	\centering
	\caption{Simulation parameters.}
	\begin{tabular}{|c|| c c c c c | c c | c c | c c | c c c c | c c |}
		\hline
		Scenario
		& $ U $
		& $ T $
		& $ S_\mathrm{com} $ 
		& $ S_\mathrm{sen} $
		& $ S $
		& $ \beta_u $ 
		& $ \theta_t \mid \widebar{\theta}_t $ 
		& $ r_u $ 
		& $ \psi_t \mid \widebar{\psi}_t $
		& $ \Upsilon_\mathrm{snr} $
		& $ \Lambda_\mathrm{sinr} $
		& $ \frac{\epsilon_{\mathrm{CSI}}}{\sigma_\mathrm{com}} $ 
		& $ \epsilon_{\mathrm{AOD}} $
		& $ \frac{\epsilon_\mathrm{RC}}{\widebar{\psi}} $
		& $ \frac{\epsilon_\mathrm{RSI}}{\left\| \widebar{\mathbf{R}} \right\|_\mathrm{F} } $
		& $ \upsilon $
		& $ \widebar{d}_c $
		\\
		\hline  
		I 
		& $ 1 $ 
		& $ 1 $ 
		& $ 1 $ 
		& $ 1 $ 
		& $ 1 $ 
		& $ 90 $ 
		& \multicolumn{1}{>{\columncolor{DodgerBlue1!20}}c|}{\tabularCenterstack{c}{$ \left[ 86, 134 \right] $}} 
		& $ 50 $ 
		& $ 10^{-3} $ 
		& $ 80 $ 
		& $ 2 $ 
		& $ 0 $ 
		& $ 0 $ 
		& $ 0 $ 
		& $ 0 $ 
		& $ 0 $ 
		& $ - $ 
		\\
		\hline 
		II  
		& $ 1 $ 
		& $ 1 $ 
		& $ 1 $ 
		& $ 1 $ 
		& \multicolumn{1}{>{\columncolor{DodgerBlue1!20}}c|}{\tabularCenterstack{c}{$ 1, 2 $}} 
		& $ 90 $ 
		& \multicolumn{1}{>{\columncolor{DodgerBlue1!20}}c|}{\tabularCenterstack{c}{$ 90, \cdots, 130 $}} 
		& $ 50 $ 
		& $ 10^{-3} $ 
		& \multicolumn{1}{>{\columncolor{DodgerBlue1!20}}c}{\tabularCenterstack{c}{$ \left[ 10, 150 \right] $}} 
		& \multicolumn{1}{>{\columncolor{DodgerBlue1!20}}c|}{\tabularCenterstack{c}{$ \left[ 0, 15 \right] $}} 
		& $ 0 $ 
		& $ 0 $ 
		& $ 0 $ 
		& $ 0 $ 
		& $ 0 $ 
		& $ - $ 
		\\
		\hline  
		III 
		& $ 1 $ 
		& $ 1 $ 
		& $ 1 $ 
		& $ 1 $ 
		& $ 1 $ 
		& $ 90 $ 
		& \multicolumn{1}{>{\columncolor{DodgerBlue1!20}}c|}{\tabularCenterstack{c}{$ 90, 110, 130 $}} 
		& $ 50 $ 
		& $ 10^{-3} $ 
		& $ 50 $ 
		& \multicolumn{1}{>{\columncolor{DodgerBlue1!20}}c|}{\tabularCenterstack{c}{$ 1,2,3 $}} 
		& $ 0 $ 
		& \multicolumn{1}{>{\columncolor{DodgerBlue1!20}}c}{\tabularCenterstack{c}{$ \left[ 0, 4.8 \right] $}} 
		& \multicolumn{1}{>{\columncolor{DodgerBlue1!20}}c}{\tabularCenterstack{c}{$ \left[ 0, 1 \right] $}} 
		& $ 0 $ 
		& $ 0 $ 
		& $ - $  
		\\
		\hline  
		IV 
		& $ 1 $ 
		& $ 1 $ 
		& $ 1 $ 
		& $ 1 $ 
		&  $ 1 $ 
		& $ 90 $ 
		& \multicolumn{1}{>{\columncolor{DodgerBlue1!20}}c|}{\tabularCenterstack{c}{$ 90, 110, 130 $}} 
		& $ 50 $ 
		& $ 10^{-3} $ 
		& \multicolumn{1}{>{\columncolor{DodgerBlue1!20}}c}{\tabularCenterstack{c}{$ 50, 70 $}} 
		& \multicolumn{1}{>{\columncolor{DodgerBlue1!20}}c|}{\tabularCenterstack{c}{$ 1,2,3 $}} 
		& \multicolumn{1}{>{\columncolor{DodgerBlue1!20}}c}{\tabularCenterstack{c}{$ \left[ 0, 10 \right] $}} 
		& $ 0 $ 
		& $ 0 $ 
		& \multicolumn{1}{>{\columncolor{DodgerBlue1!20}}c|}{\tabularCenterstack{c}{$ \left[ 0, 1 \right] $}} 
		& \multicolumn{1}{>{\columncolor{DodgerBlue1!20}}c}{\tabularCenterstack{c}{$ \left[ 0, 0.3 \right] $}} 
		& \multicolumn{1}{>{\columncolor{DodgerBlue1!20}}c|}{\tabularCenterstack{c}{$ 0.1, 0.2, 0.3 $}} 
		\\
		\hline
		\hline
		V 
		& $ 1 $ 
		& $ 1 $ 
		& \multicolumn{1}{>{\columncolor{DodgerBlue1!20}}c}{\tabularCenterstack{c}{$ 4 $}} 
		& \multicolumn{1}{>{\columncolor{DodgerBlue1!20}}c}{\tabularCenterstack{c}{$ 2 $}} 
		&  \multicolumn{1}{>{\columncolor{DodgerBlue1!20}}c|}{\tabularCenterstack{c}{$ 10 $}} 
		& $ 90 $ 
		& $ 110 $ 
		& $ 50 $ 
		& $ 10^{-3} $ 
		& $ 50 $ 
		& $ 2 $ 
		& $ 0 $ 
		& $ 0 $ 
		& $ 0 $ 
		& $ 0 $ 
		& $ 0 $ 
		& $ - $ 
		\\
		\hline
		VI 
		& \multicolumn{1}{>{\columncolor{DodgerBlue1!20}}c}{\tabularCenterstack{c}{$ 3 $}} 
		& \multicolumn{1}{>{\columncolor{DodgerBlue1!20}}c}{\tabularCenterstack{c}{$ 2 $}} 
		& $ 1 $   
		& $ 1 $ 
		& \multicolumn{1}{>{\columncolor{DodgerBlue1!20}}c|}{\tabularCenterstack{c}{$ 5 $}} 
		& \multicolumn{1}{>{\columncolor{DodgerBlue1!20}}c}{\tabularCenterstack{c}{$ \left[ 45, 135 \right] $}} 
		& \multicolumn{1}{>{\columncolor{DodgerBlue1!20}}c|}{\tabularCenterstack{c}{$ \left[ 50, 130 \right] $}} 
		& \multicolumn{1}{>{\columncolor{DodgerBlue1!20}}c}{\tabularCenterstack{c}{$ \left[ 40, 70 \right] $}} 
		& \multicolumn{1}{>{\columncolor{DodgerBlue1!20}}c|}{\tabularCenterstack{c}{$ \left[ 4, 12 \right] 10^{-4} $}} 
		& $ 50 $ 
		& \multicolumn{1}{>{\columncolor{DodgerBlue1!20}}c|}{\tabularCenterstack{c}{$ 4 $}}  
		& $ 0 $  
		& $ 0 $ 
		& $ 0 $  
		& $ 0 $  
		& $ 0 $ 
		& $ - $  
		\\
		\hline
		VII
		& \multicolumn{1}{>{\columncolor{DodgerBlue1!20}}c}{\tabularCenterstack{c}{$ 6 $}}  
		& \multicolumn{1}{>{\columncolor{DodgerBlue1!20}}c}{\tabularCenterstack{c}{$ 4 $}} 
		& \multicolumn{1}{>{\columncolor{DodgerBlue1!20}}c}{\tabularCenterstack{c}{$ 3 $}}   
		& \multicolumn{1}{>{\columncolor{DodgerBlue1!20}}c}{\tabularCenterstack{c}{$ 3 $}}     
		& \multicolumn{1}{>{\columncolor{DodgerBlue1!20}}c|}{\tabularCenterstack{c}{$ 30 $}}   
		& \multicolumn{1}{>{\columncolor{DodgerBlue1!20}}c}{\tabularCenterstack{c}{$ \left[ 45, 135 \right] $}} 
		& \multicolumn{1}{>{\columncolor{DodgerBlue1!20}}c|}{\tabularCenterstack{c}{$ \left[ 50, 130 \right] $}} 
		& \multicolumn{1}{>{\columncolor{DodgerBlue1!20}}c}{\tabularCenterstack{c}{$ \left[ 40, 70 \right] $}} 
		& \multicolumn{1}{>{\columncolor{DodgerBlue1!20}}c|}{\tabularCenterstack{c}{$ \left[ 4, 12 \right] 10^{-4} $}} 
		& $ 50 $  
		& \multicolumn{1}{>{\columncolor{DodgerBlue1!20}}c|}{\tabularCenterstack{c}{$ \left[ 1, 6 \right] $}} 
		& \multicolumn{1}{>{\columncolor{DodgerBlue1!20}}c}{\tabularCenterstack{c}{$ 4, 14 $}}  
		& \multicolumn{1}{>{\columncolor{DodgerBlue1!20}}c}{\tabularCenterstack{c}{$ 2, 3 $}}  
		& \multicolumn{1}{>{\columncolor{DodgerBlue1!20}}c}{\tabularCenterstack{c}{$ 0.1, 0.3 $}} 
		& \multicolumn{1}{>{\columncolor{DodgerBlue1!20}}c|}{\tabularCenterstack{c}{$ 0.3, 0.6 $}} 
		& \multicolumn{1}{>{\columncolor{DodgerBlue1!20}}c}{\tabularCenterstack{c}{$ 0.2 $}} 
		& \multicolumn{1}{>{\columncolor{DodgerBlue1!20}}c}{\tabularCenterstack{c}{$ 0.2 $}}  
		\\
		\hline 
	\end{tabular}
	\label{tab:simulation-settings}
	\vspace{-2mm}
\end{table*}

\section{Simulation Results} \label{sec:simulation-results}

We evaluate the investigated problem $\mathcal{P}'(\boldsymbol{\Omega}') $ under various configurations, varying the \gls{SNR} and \gls{SINR} requirements, the number of allocated timeslots, the number of users and targets, and the severity of imperfect information.

We consider the Rician fading channel model, which allows \gls{LoS} and \gls{NLoS} channel components. The channel for $ \mathsf{U}_u $ is given by $ \mathbf{h}_u = \gamma_u \mathbf{v}_u $, where $ \gamma_u $ accounts for large-scale fading and
\begin{equation*}
	\mathbf{v}_u = \sqrt{K / (K+1)} \mathbf{v}^\mathrm{LoS}_u + \sqrt{1/(K+1)} \mathbf{v}^\mathrm{NLoS}_u, \forall u \in \mathcal{U},
\end{equation*}
is the normalized small-scale fading, with $ K = 100 $ being the Rician fading factor. The \gls{LoS} component is given by $ \mathbf{v}^\mathrm{LoS}_u = \tfrac{1}{\sqrt{N_\mathrm{tx}}} \mathrm{e}^{\mathrm{j} \boldsymbol{\phi}_\mathrm{tx} \cos \left( \beta_u \right)} $, where $ \beta_u $ is the \gls{LoS} angle, and the \gls{NLoS} components are defined as $ \mathbf{v}^\mathrm{NLoS}_u \sim \mathcal{CN} \left( \mathbf{0}, \mathbf{I} \right) $. For large-scale fading, we adopt the \texttt{UMa} channel model \cite{3gpp:38.901}, modeled as $ \gamma_u = 28 + 22 \log_{10} (r_u) + 20 \log_{10} (f_\mathrm{c}) $ dB, where $ f_\mathrm{c} = 41 $ GHz is the carrier frequency, and $ r_u $ denotes the distance between the \gls{BS} and $ \mathsf{U}_u $, ranging in the interval $ \left[ 40, 70 \right] $ m. We assume communication and sensing noise powers $ \sigma_\mathrm{com}^2 = -110 $ dBm and $ \sigma_\mathrm{sen}^2 = -70 $ dBm, respectively, which reflect the distinct operational environments. Specifically, sensing is often clutter/interference-limited, whereas communication is typically noise-limited.

The \gls{BS} is equipped with $ N_\mathrm{tx} = 8 $ transmit antennas and $ N_\mathrm{rx} = 16 $ receive antennas. The transmit and receive directions span the interval from $ 45^\circ $ to $ 135^\circ $ with spacing of $ 5^\circ $ degrees,  yielding $ D_\mathrm{tx} = 19 $ and $ D_\mathrm{rx} = 19 $ distinct directions for transmission and reception, respectively. The transmit power ranges from  $ 0.1 $ to $ 1 $ W in increments of $ 0.1 $ W, resulting in $ P_\mathrm{tx} = 10 $ distinct power levels, while the receive power is fixed\footnote{Power scaling in receive beamforming does not enhance the sensing \gls{SINR}, as seen in (\ref{eqn:sensing-sinr}), because numerator and the denominator are scaled proportionally, leaving the \gls{SINR} unchanged. Thus, receive power is typically maintained at a fixed and low level.} at $ 0.1 $ W, yielding $ P_\mathrm{rx} = 1 $. The available beamwidths for transmission are $ \left\lbrace 13^\circ, 26^\circ, 60^\circ \right\rbrace $, while the beamwidth options for reception are $ \left\lbrace  6^\circ, 13^\circ, 26^\circ \right\rbrace $, leading to $ B_\mathrm{tx} = 3 $ and $ B_\mathrm{rx} = 3 $ beamwidth configurations. The mutual coupling parameters are set to $ \delta_\mathrm{tx} = \delta_\mathrm{rx} = 0 $, unless stated otherwise. The number of users $ U $ varies in the set $ \left\lbrace 1, 2, 3, 4, 5, 6 \right\rbrace $, while the number of targets $ T $ ranges over $ \left\lbrace 1, 2, 3, 4 \right\rbrace $.

The \gls{LoS} angles of the user channels vary according to $ \beta_u \in \left[ 45^\circ, 135^\circ \right] $, while the target \glspl{AOD} are drawn from $ \theta_t \in \left[ 50^\circ, 130^\circ \right] $. Both the communication \gls{SNR} and sensing \gls{SINR} requirements are assumed to be identical across all users and targets, denoted by $ \Upsilon_{\mathrm{snr}} $ and $ \Lambda_{\mathrm{sinr}} $, respectively. Specifically, $ \Upsilon_{\mathrm{snr}} $ ranges within $ \left[ 10, 150 \right] $, and $ \Lambda_{\mathrm{sinr}} $ varies in the interval $ \left[ 1, 15 \right] $. Similarly, the number of timeslots allocated for communication and sensing, denoted by $ S_\mathrm{com} $ and $ S_\mathrm{sen} $, respectively, are identical across users and targets, with values selected from the interval $ \left[ 1, 10 \right] $. The \gls{RC} varies according to $ \widebar{\psi}_t \in \left[4 \cdot 10^{-4}, 12 \cdot 10^{-4} \right] $.

The normalized \gls{CSI} error varies within $ \frac{\epsilon_{\mathrm{CSI}}}{\sigma_\mathrm{com}} \in \left[ 0, 10 \right] $ and the \gls{AOD} error ranges as $ \epsilon_\mathrm{AOD} \in \left[ 0^\circ, 4.8^\circ \right] $. The normalized \gls{RC} error lies within $ \frac{\epsilon_\mathrm{RC}}{\widebar{\psi}_t} \in \left[ 0, 1 \right] $. The severity of \gls{RSI} is controlled via $ \upsilon $ and $ \epsilon_\mathrm{RSI} $, which are varied in the intervals $ \upsilon \in \left[ 0, 1 \right] $ and $ \frac{\epsilon_\mathrm{RSI}}{\left\| \widebar{\mathbf{R}} \right\|_\mathrm{F} } \in \left[ 0, 1 \right] $. Additionally, the distance between the antenna array centers ranges over $ \widebar{d}_c \in \left\lbrace 0, 0.1, 0.2, 0.3 \right\rbrace $ m, while the timeslot duration, $ S_\mathrm{dur} $, is set to $ 1 $ ms. When not specified, all simulation results are the average over $ 50 $ independent realizations.

We consider two sets of scenarios: (i) single-user, single-target, single-timeslot settings (\textbf{Scenario~I} to \textbf{Scenario~IV}), and (ii) multi-user, multi-target, multi-timeslot settings (\textbf{Scenario~V} to \textbf{Scenario~VII}). The first set of scenarios provides key insights into how the various parameters influence system performance. Specifically, we examine the impact of user-target alignment, \gls{SNR} and \gls{SINR} requirements, beam adaptation, functionality selection, and imperfect information. The second set of scenarios reflects more general configurations, aiming to evaluate system performance at a larger scale. We evaluate our optimal user-target pairing against a state-of-the-art heuristic, analyze the time-energy trade-off, and benchmark the proposed approach against baseline methods. For clarity, \textbf{\cref{tab:simulation-settings}} provides a summary of the most relevant parameters for each scenario\footnote{Throughout \textbf{Scenario I} to \textbf{Scenario IV}, we adopt $ \eta_1 = 1 $, $ \eta_2 = 1 $, and $ \omega_1 = 1 $, unless stated otherwise. In \textbf{Scenario V} to \textbf{Scenario VII}, $ \eta_1 $ remains fixed at $ 1 $, while $ \eta_2 $ is selected based on the weighting conditions derived in \textbf{\cref{thm:proposition-9}}. For these scenarios, the structured weights $ \omega_s $ are generated assuming $ \Delta_\omega = \Delta_0 = 2 $.}.


\subsection{Scenario I: Impact of user-target alignment}

\begin{figure*}[!t]
 	\begin{subfigure}[b]{0.23\textwidth}
		\begin{center}
			\includegraphics[]{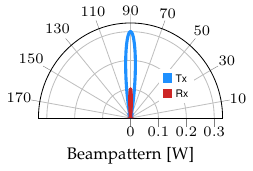}
			\caption{$  \beta = 90^\circ \mid \theta = 90^\circ $}
			\label{fig:results-scenario-1a}
		\end{center}
 	\end{subfigure}
    \hfill 
 	\begin{subfigure}[b]{0.23\textwidth}
		\begin{center}
			\includegraphics[]{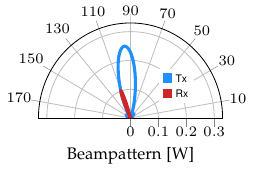}
			\caption{$ \beta = 90^\circ \mid \theta = 110^\circ $}
			\label{fig:results-scenario-1b}
		\end{center}
 	\end{subfigure}
    \hfill 
 	\begin{subfigure}[b]{0.23\textwidth}
		\begin{center}
			\includegraphics[]{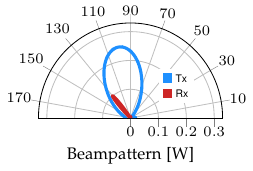}
			\caption{$ \beta = 90^\circ \mid \theta = 130^\circ $}
			\label{fig:results-scenario-1c}
		\end{center}
 	\end{subfigure}
 	\hfill
 	\begin{subfigure}[b]{0.27\textwidth}
		\begin{center}
			\includegraphics[]{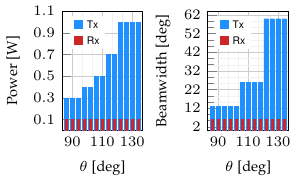}
			\caption{$ \beta = 90^\circ \mid \theta \in \left\lbrace 90^\circ, 110^\circ, 130^\circ \right\rbrace $}
			\label{fig:results-scenario-1d}
		\end{center}
 	\end{subfigure}
	%
 	\begin{subfigure}[b]{0.23\textwidth}
		\begin{center}
			\includegraphics[]{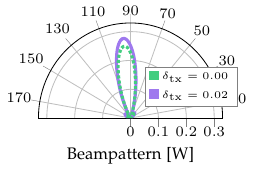}
			\caption{$ \theta = 110^\circ \mid \delta_\mathrm{tx} = 0.02 $}
			\label{fig:results-scenario-1e}
		\end{center}
 	\end{subfigure}
    \hfill 
 	\begin{subfigure}[b]{0.23\textwidth}
		\begin{center}
			\includegraphics[]{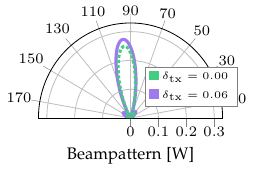}
			\caption{$ \theta = 110^\circ \mid \delta_\mathrm{tx} = 0.06 $}
			\label{fig:results-scenario-1f}
		\end{center}
 	\end{subfigure}
    \hfill 
 	\begin{subfigure}[b]{0.23\textwidth}
		\begin{center}
			\includegraphics[]{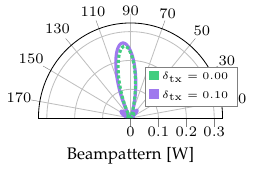}
			\caption{$ \theta = 110^\circ \mid \delta_\mathrm{tx} = 0.10 $}
			\label{fig:results-scenario-1g}
		\end{center}
 	\end{subfigure}
 	\hfill
 	\begin{subfigure}[b]{0.27\textwidth}
		\begin{center}
			\includegraphics[]{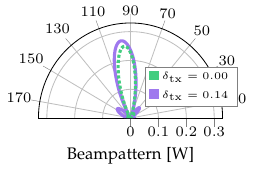}
			\caption{$ \theta = 110^\circ \mid \delta_\mathrm{tx} = 0.14 $}
			\label{fig:results-scenario-1h}
		\end{center}
 	\end{subfigure}
 	\vspace{-2mm}
    \caption{Impact of user-target alignment on beam adaptation (\emph{Scenario I}). \emph{The system can dynamically adapt beam direction, beamwidth, and transmit power in response to variations in user-target alignment, thereby ensuring that both communication and sensing requirements are satisfied. In addition, mutual coupling increases power/energy consumption by inducing energy leakage among adjacent antennas, which distorts the beampattern and reduces the effective power concentrated in the main lobe.}}
 	\label{fig:results-scenario-1}
\end{figure*}

In this scenario, we analyze how the angular alignment between user and target influences the adaptation of power, direction, and beamwidth. The setup consists of a single user and a single target, with requirements, $ \Upsilon_{\mathrm{snr}} = 80 $ and $ \Lambda_{\mathrm{sinr}} = 2 $, assuming $ S = S_{\mathrm{com}} = S_{\mathrm{sen}} = 1 $. The user's \gls{LoS} angle is $ \beta = 90^\circ $, while the target's \gls{AOD} $ \theta $ varies from $ 86^\circ $ to $ 134^\circ $, resulting in varying degrees of alignment. We also evaluate the impact of mutual coupling on beampattern shape and power consumption.

In \cref{fig:results-scenario-1a}, both user and target are perfectly aligned at $ 90^\circ $, and accordingly, the transmit and receive beams are also directed towards this angle. The transmit power to jointly serve both is $ 0.3 $ W, while the receive power for target sensing remains constant at $ 0.1 $ W. The transmit beamwidth is $13^\circ $, while the receive beamwidth is $ 6^\circ $, which remains fixed across all cases in \cref{fig:results-scenario-1}, as there is no self-interference to mitigate, i.e., $ \upsilon = 0 $ (see \cref{tab:simulation-settings}).

As the target moves to $ 110^\circ $, as shown in \cref{fig:results-scenario-1b}, the transmit beam adapts by broadening its beamwidth to $ 26^\circ $ and steering to $ 95^\circ $ to accommodate the target's angular shift. Additionally, the transmit power increases to $ 0.5 $ W to compensate for the reduced radiated power due to beam widening. This adaptation is necessary, as the beam configuration from \cref{fig:results-scenario-1a} is no longer suitable to jointly serve both the user and the target.

When the target reaches $ 130^\circ $, as illustrated in \cref{fig:results-scenario-1c}, a more pronounced behavior is observed. Specifically, the transmit beam is steered to $ 100^\circ $, while broadening its beamwidth to $ 60^\circ $. To counteract the significant loss in directionality caused by the wider beam, the transmit power is further increased to $ 1 $ W.

Meanwhile, \cref{fig:results-scenario-1d} provides a more granular view, displaying transmit and receive power alongside beamwidth across a broader range of $ \theta $ values. Additionally, since this setting enforces a shared timeslot for the user and the target, it allows us to analyze channel similarity, which takes the values $1$, $0.2232$, and $0.1438$ for \cref{fig:results-scenario-1a}, \cref{fig:results-scenario-1b}, and \cref{fig:results-scenario-1c}, respectively, clearly illustrating the strong correlation between channel similarity and user-target alignment.

\cref{fig:results-scenario-1e} to \cref{fig:results-scenario-1h} illustrate the impact of mutual coupling on the beampattern, considering $ \beta = 90^\circ $ and $ \theta = 110^\circ $ (as in \cref{fig:results-scenario-1b}). In this setting, the receive-side coupling is fixed to $ \delta_\mathrm{rx} = 0.02 $, while the transmit-side coupling is varied as $ \delta_\mathrm{tx} = \left\lbrace 0.02, 0.06, 0.10, 0.14  \right\rbrace $. As $ \delta_\mathrm{tx} $ increases, stronger interactions among neighboring antenna elements arise, leading to increased power leakage. Consequently, higher transmit power is required to satisfy the sensing and communication constraints, since the coupling matrix $ \mathbf{Z}_\mathrm{tx} $ in (\ref{eqn:coupling-matrix}) induces phase distortions that reduce coherent combining across the array. Across all four cases in \cref{fig:results-scenario-1}, the resulting beampatterns under mutual coupling exhibit higher main-lobe energy than the ideal case (with $ \delta_\mathrm{tx} = \delta_\mathrm{rx} = 0 $), indicating that both sensing and communication requirements remain satisfied. However, this comes at the cost of increased transmit power. Specifically, the selected powers for $ \delta_\mathrm{tx} = \left\lbrace 0.02, 0.06, 0.10, 0.14  \right\rbrace $ are $0.6$W, $0.7$W, $0.8$W and $1$W, respectively, whereas the nominal case without mutual coupling requires only $0.5$W. This increase represents the additional power needed to guarantee the ISAC constraints under coupling effects. Furthermore, the sidelobe levels become more pronounced as $ \delta_\mathrm{tx} $ increases, reflecting the growing distortion of the array response due to mutual coupling.


\subsection{Scenario II: Impact of SNR/SINR requirements and functionality selection}

\begin{figure*}[!t]
 	\begin{subfigure}[b]{0.19\textwidth}
		\begin{center}
			\includegraphics[]{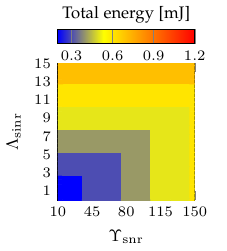}
			\vspace{-5mm}
			\caption{$ \theta = 90^\circ \mid S = 1 $}
			\label{fig:results-scenario-2a}
		\end{center}
 	\end{subfigure}
 	\hfill 
 	\begin{subfigure}[b]{0.19\textwidth}
		\begin{center}
			\includegraphics[]{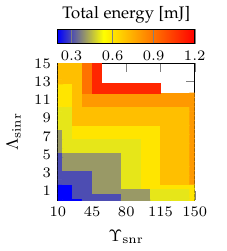}
			\vspace{-5mm}
			\caption{$ \theta = 100^\circ \mid S = 1 $}
			\label{fig:results-scenario-2b}
		\end{center}
 	\end{subfigure}
    \hfill 
 	\begin{subfigure}[b]{0.19\textwidth}
		\begin{center}
			\includegraphics[]{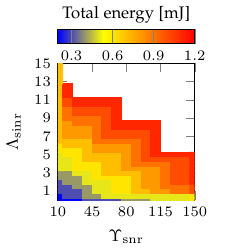}
			\vspace{-5mm}
			\caption{$  \theta = 110^\circ \mid S = 1 $}
			\label{fig:results-scenario-2c}
		\end{center}
 	\end{subfigure}
    \hfill 
 	\begin{subfigure}[b]{0.19\textwidth}
		\begin{center}
			\includegraphics[]{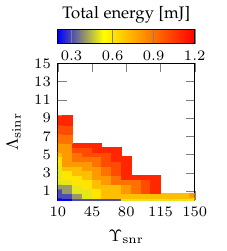}
			\vspace{-5mm}
			\caption{$ \theta = 120^\circ \mid S = 1 $}
			\label{fig:results-scenario-2d}
		\end{center}
 	\end{subfigure}
    \hfill 
 	\begin{subfigure}[b]{0.19\textwidth}
		\begin{center}
			\includegraphics[]{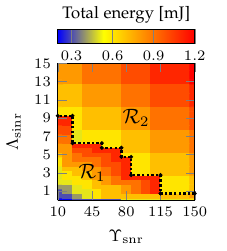}
			\vspace{-5mm}
			\caption{$ \theta = 120^\circ  \mid S = 2 $}
			\label{fig:results-scenario-2e}
		\end{center}
 	\end{subfigure}
 	\vspace{-1mm}
    \caption{Impact of SNR/SINR requirements and functionality selection on energy consumption (\emph{Scenario II}). \emph{The energy consumption is affected by \gls{SNR} and \gls{SINR} requirements, with higher demands leading to greater infeasible regions. While static functionality selection (\cref{fig:results-scenario-2a} to \cref{fig:results-scenario-2d}) is more prone to infeasibility, flexible functionality selection (\cref{fig:results-scenario-2e}) significantly increases the feasible solution space and improves adaptability via intelligent switching between joint and separate user/target service.}}
 	\label{fig:results-scenario-2}
\end{figure*}
\begin{figure*}[!t]
 	\begin{subfigure}[b]{0.40\textwidth}
		\begin{center}
			\includegraphics[]{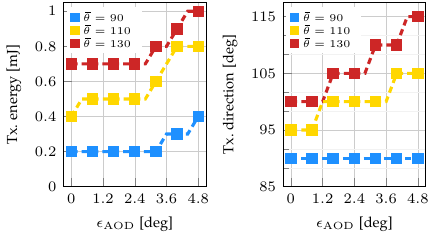}
			\vspace{-5mm}
			\caption{$ \Lambda_{\mathrm{sinr}} = 2 $}
			\label{fig:results-scenario-3a}
		\end{center}
 	\end{subfigure}
    \hfill 
 	\begin{subfigure}[b]{0.33\textwidth}
		\begin{center}
			\includegraphics[]{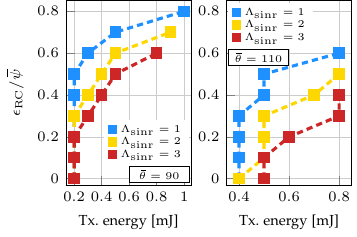}
			\vspace{-5mm}
			\caption{$ \Lambda_{\mathrm{sinr}} = \left\lbrace 1, 2, 3 \right\rbrace \mid \theta \in \left\lbrace 90^\circ, 110^\circ \right\rbrace $}
			\label{fig:results-scenario-3b}
		\end{center}
 	\end{subfigure}
    \hfill 
 	\begin{subfigure}[b]{0.23\textwidth}
		\begin{center}
			\includegraphics[]{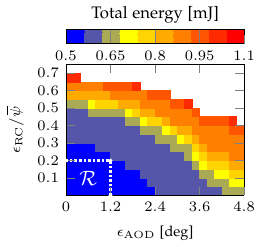}
			\vspace{-5mm}
			\caption{$ \Lambda_{\mathrm{sinr}} = 1 \mid \theta = 110^\circ $}
			\label{fig:results-scenario-3c}
		\end{center}
 	\end{subfigure}
 	\vspace{-1mm}
    \caption{Impact of imperfect AOD and RC on energy consumption (\emph{Scenario III}). \emph{Energy expenditure is impacted by \gls{AOD} and \gls{RC} uncertainty. Their combined effect reveals operational regions characterized by constant total energy consumption despite imperfect information. This stability is attributed to the discrete power levels, which often provides a surplus margin in satisfying \gls{SNR} and \gls{SINR} requirements, thereby providing an inherent degree of robustness.} }
 	\label{fig:results-scenario-3}
\end{figure*}


In this scenario, we analyze how energy consumption varies under different \gls{SNR} and \gls{SINR} requirements. We also investigate how functionality selection influences performance. The setup consists of a single user and a single target, with requirements $ \Upsilon_{\mathrm{snr}} = \left[ 10, 150 \right]  $ and $ \Lambda_{\mathrm{sinr}} = \left[ 0, 15 \right] $.

In \cref{fig:results-scenario-2}, the consumed energy is depicted using a color heatmap across five instances, where the user's LoS angle is fixed at $ \beta = 90^\circ $, and the target's  \gls{AOD} $ \theta $ takes values from $ \left\lbrace 90^\circ, 100^\circ, 110^\circ, 120^\circ \right\rbrace $. In \cref{fig:results-scenario-2a} to \cref{fig:results-scenario-2d}, we set $ S = 1 $, meaning the user and target share the same timeslot. This reflects a common assumption in the existing literature, where timeslot functionality is static, forcing both functionalities to be active simultaneously. In \cref{fig:results-scenario-2e}, we set $ S = 2 $, allowing the user and target to be time-multiplexed when a single timeslot is insufficient to serve both. This configuration enables flexible functionality selection across timeslots, mitigating allocation infeasibility.

In \cref{fig:results-scenario-2a}, due to the high alignment between user and target, all allocations are feasible, with energy consumption not exceeding $ 0.7 $ mJ. In \cref{fig:results-scenario-2b}, a misalignment of $ \theta - \beta = 10^\circ $, leads to the appearance of infeasible regions, represented by white areas, particularly for $ \Lambda_{\mathrm{sinr}} \geq 12 $. \cref{fig:results-scenario-2c} illustrates a misalignment of $ \theta - \beta = 20^\circ $, leading to a more noticeable reduction in the feasible regions. Moreover, the number of allocations requiring high energy consumption rises significantly compared to  previous cases. In \cref{fig:results-scenario-2d}, with a misalignment of $ \theta - \beta = 30^\circ $, total energy consumption increases sharply, even for moderate \gls{SNR} and \gls{SINR} requirements. Additionally, regions requiring maximum energy expand noticeably, while those with minimal energy contract. The number of infeasible regions also increases significantly.

\cref{fig:results-scenario-2e} illustrates how enabling flexible functionality selection resolves allocation infeasibility. The operating conditions are identical to those in \cref{fig:results-scenario-2d}. The black dotted line partitions the energy plane into two regions: $ \mathcal{R}_1 $, where a single timeslot is used, and $ \mathcal{R}_2 $, where two timeslots are active. This boundary marks the point at which the system switches from one to two timeslots. The weighting parameters ($ \eta_1 = 1 $ and $ \eta_2 = 2 $) are chosen such that the second timeslot becomes active only when allocation under a single timeslot is infeasible. Immediately to the left of the boundary (within $ \mathcal{R}_1 $), the allocated energies are high and approach their maximum feasible values. Just to the right of the boundary (within $ \mathcal{R}_2 $), the energies drop sharply because the system time-multiplexes the user and the target. In this mode, each can be served with a narrower beam and therefore with less energy, at the expense of using additional time resources. Overall, enabling flexible functionality selection on a per-timeslot basis expands the feasible solution space by approximately $250\%$.


\subsection{Scenario III: Impact of imperfect AOD and RC}

In this scenario, we investigate how the transmit beam adapts to varying levels of uncertainty in the target's \gls{AOD} and \gls{RC}. We also examine how combined uncertainties in \gls{AOD} and \gls{RC} influence the energy consumption. The setup remains consistent with assumptions in \textbf{Scenario~II}. The communication \gls{SNR} requirement is $ \Upsilon_{\mathrm{snr}} = 50 $, and the user's \gls{LoS} is $ \beta = 90^\circ $. The estimated target's \gls{AOD} $ \widebar{\theta} $ varies within $ \left\lbrace 90^\circ, 110^\circ, 130^\circ  \right\rbrace $, while the estimated \gls{RC} is set to $ \widebar{\psi} = 10^{-3} $. Here, \gls{AOD} and \gls{RC} errors range the intervals $ \epsilon_{\mathrm{AOD}} \in \left[ 0^\circ, 4.8^\circ \right] $ and $ \frac{\epsilon_{\mathrm{RC}}}{\widebar{\psi}} \in \left[ 0,1 \right] $, respectively.


\cref{fig:results-scenario-3a} illustrates the transmit energy and direction as functions of $ \epsilon_{\mathrm{AOD}} $. Both the transmit energy and the steering angle exhibit monotonically non-decreasing trends as $ \epsilon_{\mathrm{AOD}} $ increases.
Specifically, addressing \gls{AOD} uncertainty requires maintaining sufficient radiated energy in the angular neighborhood of $ \widebar{\theta} $ to satisfy the \gls{SINR} requirement. While beam widening could mitigate this issue, it is unsuitable in the considered setting, as doubling the beamwidth would at least halve the directional gain, thereby requiring a proportional increase in transmit power and substantially degrading the objective function. Instead, the system combats \gls{AOD} uncertainty through power and direction adaptation. In particular, increasing transmit power while maintaining the beamwidth incurs a lower cost than widening the beam to achieve the required performance. However, when power increment alone is insufficient, the transmit beam is steered closer to $ \widebar{\theta} $. As this shift reduces the energy delivered to the user, a compensatory increase in transmit power may be needed to maintain the communication \gls{SNR} requirement. 


\cref{fig:results-scenario-3b} shows the impact of \gls{RC} uncertainty on transmit energy for $ \widebar{\theta} = \left\lbrace 90^\circ,  110^\circ \right\rbrace $. As expected, greater \gls{RC} uncertainty degrades sensing performance. For $ \widebar{\theta} = 90^\circ $, the system remains feasible under $ \Lambda_{\mathrm{sinr}} = 1 $, even with \gls{RC} errors up to $ 80\% $. However, stricter $ \Lambda_{\mathrm{sinr}} = \left\lbrace 2, 3 \right\rbrace $ reduce the tolerable error to $ 70\% $ and $ 60\% $, respectively, beyond which the allocation becomes infeasible. For $ \widebar{\theta} = 110^\circ $, a similar trend is observed, though energy consumption increases more sharply, even for small \gls{RC} errors.

\cref{fig:results-scenario-3c} illustrates the combined effect of \gls{AOD} and \gls{RC} uncertainty on total energy consumption for $ \widebar{\theta} = 110^\circ $ and $ \Lambda_{\mathrm{sinr}} = 1 $. The results show that this configuration can tolerate moderate \gls{RC} errors up to $ 20\% $ and \gls{AOD} deviations up to $ 1.2^\circ $, as represented by region $ \mathcal{R} $, without incurring additional energy costs. This stems from the discrete nature of power levels, which often results in achieved \gls{SNR} and \gls{SINR} exceeding the required thresholds, rather than matching them precisely, thereby providing resilience to errors.


\subsection{Scenario IV: Impact of imperfect CSI and RSI}

\begin{figure*}[!t]
 	\begin{subfigure}[b]{0.34\textwidth}
		\begin{center}
			\includegraphics[]{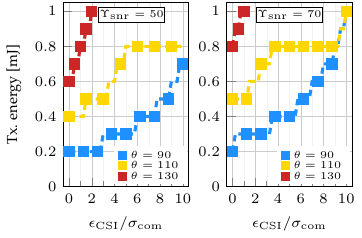}
			\caption{$ \Upsilon_{\mathrm{snr}} = \left\lbrace 50, 70 \right\rbrace  \mid \Lambda_{\mathrm{sinr}} = 1  $}
			\label{fig:results-scenario-4a}
		\end{center}
 	\end{subfigure}
    \hfill 
 	\begin{subfigure}[b]{0.17\textwidth}
		\begin{center}
			\includegraphics[]{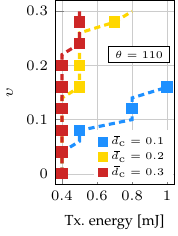}
			\caption{$ \epsilon_{\mathrm{RSI}} = 0  $}
			\label{fig:results-scenario-4b}
		\end{center}
 	\end{subfigure}
    \hfill 
 	\begin{subfigure}[b]{0.24\textwidth}
		\begin{center}
			\includegraphics[]{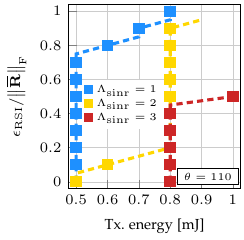}
			\caption{$ \upsilon = 0.15 \mid \widebar{d}_\mathrm{c} = 0.2  $}
			\label{fig:results-scenario-4c}
		\end{center}
 	\end{subfigure}
    \hfill 
 	\begin{subfigure}[b]{0.23\textwidth}
		\begin{center}
			\includegraphics[]{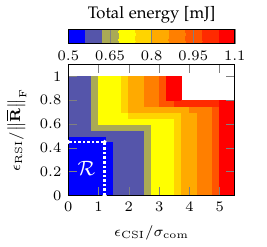}
			\caption{$ \upsilon = 0.1 \mid \widebar{d}_\mathrm{c} = 0.2 $}
			\label{fig:results-scenario-4d}
		\end{center}
 	\end{subfigure}
 	\vspace{-1mm}
    \caption{Impact of imperfect CSI and RSI on energy consumption (\emph{Scenario IV}). \emph{Effective self-interference cancellation (small $ \upsilon $) results in minimal energy increase, whereas ineffective cancellation (large $ \upsilon $) can result in excessive power demands and infeasible allocations. Furthermore, maintaining a reasonable inter-array distance significantly reduces self-interference effects.}}
 	\label{fig:results-scenario-4}
\end{figure*}

In this scenario, we examine how varying levels of uncertainty in \gls{CSI} and \gls{RSI} affect energy consumption. The same parameter settings as in \textbf{Scenario~III} are maintained. The normalized \gls{CSI} error varies within $ \frac{\epsilon_{\mathrm{CSI}}}{\sigma_\mathrm{com}} \in \left[ 0, 10 \right] $, while the normalized RSI error is controlled by $ \frac{\epsilon_\mathrm{RSI}}{\left\| \widebar{\mathbf{R}} \right\|_\mathrm{F} } \in \left[ 0, 1 \right] $.

\cref{fig:results-scenario-4a} shows that transmit energy increases with growing \gls{CSI} error for both $ \Upsilon_\mathrm{snr} = 50 $ and $ \Upsilon_\mathrm{snr} = 70 $. When the user and target are well aligned (i.e., $ \theta = 90^\circ $), the impact is moderate, even for errors up to $ 10 \sigma_\mathrm{com} $. However, \gls{CSI} error becomes more detrimental with stronger misalignment (i.e., $ \theta = 110^\circ $ and $ \theta = 130^\circ $), which significantly amplifies energy consumption. Specifically, for $ \theta = 130^\circ $, the system becomes infeasible for \gls{CSI} errors as small as $ 2 \sigma_\mathrm{com} $ and $ \sigma_\mathrm{com} $ when $ \Upsilon_\mathrm{snr} = 50 $ and $ \Upsilon_\mathrm{snr} = 70 $, respectively.

\cref{fig:results-scenario-4b} illustrates the impact of \gls{RSI} as a function of the distance between the transmit and receive arrays, focusing on the deterministic component whose severity is governed by $ \upsilon $. For closely spaced arrays (i.e., $ \widebar{d}_\mathrm{c} = 0.1$ m), \gls{RSI} causes a pronounced increase in the required transmit energy, even for small values of $ \upsilon $, which correspond to highly effective self-interference cancellation. As the inter-array distance increases, the influence of \gls{RSI} progressively weakens, consistent with trends reported in prior studies (e.g., \cite{abanto2025resilient}). For instance, at $ \widebar{d}_\mathrm{c} = 0.3 $ m, the influence is modest, with transmit energy remaining at $ 0.4 $ mJ for $ \upsilon \leq 0.2 $, and rising slightly to $ 0.5 $ mJ for $ 0.2 < \upsilon < 0.45 $.

\cref{fig:results-scenario-4c} illustrates the impact of the stochastic \gls{RSI} component on transmit energy, assuming $ \upsilon = 0.15 $ and $ \widebar{d}_\mathrm{c} = 0.2 $ m. The results demonstrate that, across all three $ \Lambda{\mathrm{sinr}} $ levels, the system can tolerate relatively large stochastic residuals. This resilience is attributed to both effective self-interference cancellation and sufficiently distanced arrays, which jointly mitigate the impact of direct coupling.

\cref{fig:results-scenario-4d} examine the joint impact of \gls{CSI} and \gls{RSI} uncertainties on total energy consumption when $ \theta = 110^\circ $, $ \Upsilon_{\mathrm{snr}} = 50 $, and $ \Lambda_{\mathrm{sinr}} = 1 $. Notably, the system demonstrates greater robustness to \gls{RSI} variations compared to \gls{CSI} errors, primarily due to the small value of $ \upsilon = 0.1 $, which reflects effective self-interference cancellation. Remarkably, energy consumption remains stable within region $ \mathcal{R} $, despite the presence of both uncertainties. This behavior mirrors the robustness observed in \cref{fig:results-scenario-3c}, and is largely attributed to the use of discrete power levels, which often satisfy system requirements with excess power.


\subsection{Scenario V: Impact of tradeoff weights design}

\begin{figure*}[!t]
 	\begin{subfigure}[b]{0.24\textwidth}
		\begin{center}
			\includegraphics[]{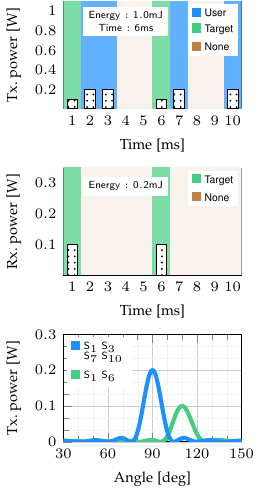}
			\vspace{-4mm}
			\caption{ $ \eta_1 = 10 \mid \eta_2 = 0 $. Here, all $ \omega_s $ are irrelevant since $ \eta_2 = 0 $.}
			\label{fig:results-scenario-5a}
		\end{center}
 	\end{subfigure}
 	\hfill 
 	\begin{subfigure}[b]{0.24\textwidth}
		\begin{center}
			\includegraphics[]{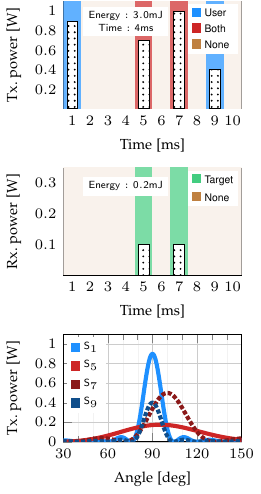}
			\vspace{-4mm}
			\caption{ $ \eta_1 = 0 \mid \eta_2 = 10 $. Here, all $ \omega_s $ are set to $ 1 $.}
			\label{fig:results-scenario-5b}
		\end{center}
 	\end{subfigure}
    \hfill 
 	\begin{subfigure}[b]{0.24\textwidth}
		\begin{center}
			\includegraphics[]{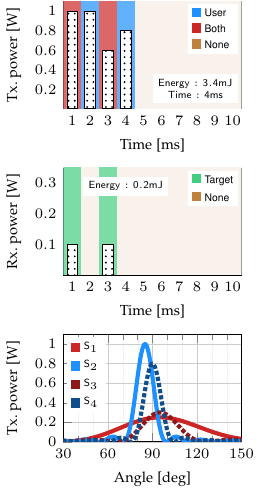}
			\vspace{-4mm}
			\caption{ $ \eta_1 = 0 \mid \eta_2 = 10 $. Here, all $ \omega_s $ follow \textbf{Lemma 1}.}
			\label{fig:results-scenario-5c}
		\end{center}
 	\end{subfigure}
    \hfill 
 	\begin{subfigure}[b]{0.24\textwidth}
		\begin{center}
			\includegraphics[]{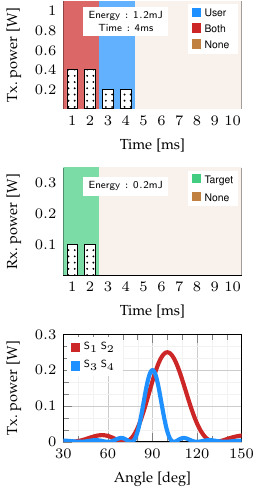}
			\vspace{-4mm}
			\caption{ $ \eta_1 = 1 \mid \eta_2 = 10 $. Here, all $ \omega_s $ follow \textbf{Lemma 1}.}
			\label{fig:results-scenario-5d}
		\end{center}
 	\end{subfigure}
 	\vspace{-1mm}
    \caption{Impact of weight design on resource economy and timeslot cohesion (\emph{Scenario V}). \emph{The weights $ \eta_1 $ and $ \eta_2 $ play a critical role in steering the optimization between energy and time consumption. The choice of weights also influences whether joint user-target servicing is favored or whether individual servicing is more energy- or time-efficient. Importantly, the use of structured weights $ \omega_s $ proves effective in promoting cohesive timeslot usage.}}
 	\label{fig:results-scenario-5}
\end{figure*}

In this scenario, we investigate the impact of weights $ \eta_1 $ and $ \eta_2 $ on time and energy consumption. We also analyze how the weights $ \omega_s $ (within function $ f_2(\boldsymbol{\Omega}) $) influence the cohesiveness of timeslot allocation. The parameters remain consistent with \textbf{Scenario~III}, considering the following: $ \beta = 90^\circ $, $ \theta = 110^\circ $, $ S_\mathrm{com} = 4 $, $ S_\mathrm{sen} = 2 $, and $ S = 10 $.

\cref{fig:results-scenario-5a} corresponds to the case $ \eta_1 = 10 $ and $ \eta_2 = 0 $, where the objective focuses exclusively on minimizing energy consumption. Under this criterion, the system attains the minimum total energy expenditure of $ 1.2 $ mJ by assigning the user and the target to disjoint timeslots. This separation avoids the additional energy cost associated with joint servicing, caused by the poor angular alignment between the user and the target. However, this energy-optimal strategy results in the maximum number of active timeslots, namely $ 6 $, corresponding to a time consumption of $ 6 $ ms. This figure also shows the specific timeslots in which the target is sensed and the transmit beampatterns for each timeslot. Notably, $ \mathsf{S}_{2} $, $ \mathsf{S}_{3} $, $ \mathsf{S}_{7} $, and $ \mathsf{S}_{10} $ employ the same transmit beam, as it achieves the minimum energy required to satisfy the \gls{SNR} constraint. Similarly, $ \mathsf{S}_{1} $ and $ \mathsf{S}_{6} $, employ the same transmit beam to illuminate the target, satisfying its \gls{SINR} requirement at minimum energy.

\cref{fig:results-scenario-5b} illustrates the opposite extreme, with $ \eta_1 = 0 $ and $ \eta_2 = 10 $, where the objective is to minimize time consumption only, assuming uniform weights $ \omega_s = 1 $. Here, the user and target are jointly served in $ \mathsf{S}_{5} $ and $ \mathsf{S}_{7} $, while the user is served independently in $ \mathsf{S}_{1} $ and $ \mathsf{S}_{9} $. In particular, the minimum possible number of timeslots is achieved, totaling $ 4 $ and requiring $ 4 $ ms. Since energy consumption is not included in the objective, multiple solutions can achieve this minimum time, and the resulting total energy consumption ($ 3.2 $ ms) is not optimized. Moreover, the active timeslots are non-contiguous, highlighting the lack of any intrinsic preference for cohesive allocation. Furthermore, four distinct beampatterns are observed, specifically, red beams correspond to joint user-target servicing, while blue beams represent timeslots where only the user is served.

\cref{fig:results-scenario-5c} considers the same time-only objective ($ \eta_1 = 0 $, $ \eta_2 = 10 $) but with structured weights $ \omega_s $ defined according to \textbf{Lemma~\ref{lem:lemma-weights}}. While the minimum time consumption of $ 4 $ ms is again achieved, the introduction of non-uniform weights promotes cohesive allocation, resulting in exactly four contiguous active timeslots. This contrasts with the non-contiguous pattern observed in \cref{fig:results-scenario-5b}. As energy is still excluded from the objective, the total energy consumption remains unoptimized, amounting to $ 3.6 $ mJ. Additionally, four distinct beampatterns are shown, corresponding to each active timeslot.

\cref{fig:results-scenario-5d} corresponds to $ \eta_1 = 1 $ and $ \eta_2 = 10 $, following the weight selection guideline in \textbf{\cref{thm:proposition-9}}, which prioritizes time minimization while retaining energy awareness. The weights $ \omega_s $ again follow \textbf{Lemma~\ref{lem:lemma-weights}}, promoting contiguous timeslot usage. The minimum time consumption of $ 4 $ ms is preserved. Owing to constraint $ \mathrm{K}_{2} $, the per-timeslot energies are sorted in decreasing order, yielding a total energy consumption of $ 1.4 $ mJ. This value is slightly higher than the energy-optimal case in \cref{fig:results-scenario-5a}, where energy minimization was the sole objective, but significantly lower than in the purely time-driven cases. Moreover, only two distinct beampatterns are required: one for joint user-target servicing, and another for the user when served individually.


\subsection{Scenario VI: Impact of flexible user-target pairing}

\begin{figure*}[!t]
		\begin{center}
			\includegraphics[]{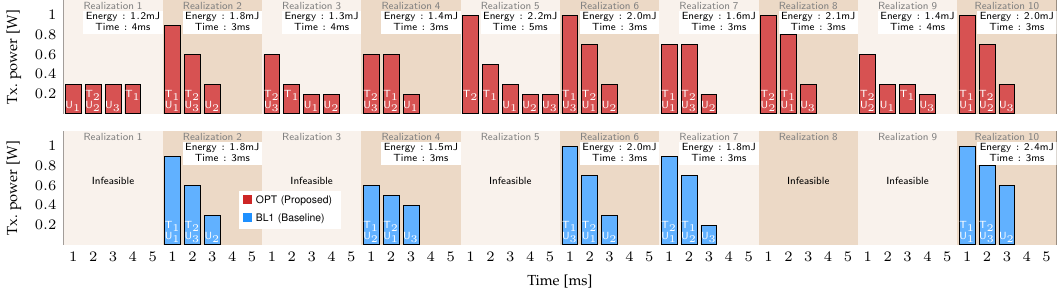}
		\end{center}
		\vspace{-5mm}
    	\caption{Impact of flexible user-target pairing on resource economy (\emph{Scenario VI}). \emph{Dynamic pairing is beneficial for improving both feasibility and resource efficiency. Unlike fixed strategies, the proposed approach adapts pairings based on the system conditions and service requirements, reducing the likelihood of infeasibility and enabling joint or separate servicing as needed. This flexibility is critical for ensuring energy-efficient operation and avoiding costly re-allocations due to infeasibility.} }
    	\vspace{-2mm}
 		\label{fig:results-scenario-8}
\end{figure*}

In this scenario, we evaluate the performance of dynamic user-target pairing in contrast to fixed-criterion pairing. The configuration considers $ U =3 $, $ T = 2 $, $ S = 5 $, $ \Upsilon_{\mathrm{snr}} = 50 $, and $ \Lambda_{\mathrm{sinr}} = 4 $, with random values for $ \beta_u $, $ \theta_t $, $ r_u $, and $ \widebar{\psi}_t $ drawn from the ranges specified in \textbf{Table~\ref{tab:simulation-settings}}.

\Cref{fig:results-scenario-8} shows the allocation outcomes for the first 10 realizations (out of 50) using our proposed approach ({\mytextsf{OPT}}) and a state-of-the-art baseline ({\mytextsf{BL1}}) adapted from \cite{dou2024:channel-sharing-aided-integrated-sensing-communication-energy-efficient-sensing-scheduling-approach}. In particular, {\mytextsf{OPT}} ensures dynamic and optimal user-target pairing, while {\mytextsf{BL1}} constructs a bipartite graph with users and targets as nodes, where edge weights represent channel similarity. A matching game is then solved to determine pairings that maximize overall alignment, reducing resource consumption through beam and timeslot sharing. To illustrate the pairing decisions, the figure displays the specific user and target served in each timeslot, along with the corresponding time and transmit energy consumption per realization.

In {\mytextsf{OPT}}, the number of active timeslots per realization ranges from  $ 3 $ (when user and targets are paired) to $ 5 $ (when users and targets are served separately). In contrast, {\mytextsf{BL1}}, enforces pairings via matching, and consequently always employs exactly $ 3 $ timeslots, i.e., the maximum of $ U $ and $ T $. Here, {\mytextsf{OPT}} yields feasible allocations in all first $ 10 $ realization, whereas {\mytextsf{BL1}} achieves feasibility in only $ 5 $. Of these, only $ 2 $ (realizations $ 2 $ and $ 6 $) achieve the same optimal value as {\mytextsf{OPT}}, while the remaining feasible allocations (realizations $ 4 $, $ 7 $, and $ 10$) incur higher energy consumption due to suboptimal pairings. Across all $ 50 $ realizations, {\mytextsf{OPT}} yields only $ 4 $ infeasible allocations, caused by insufficient power to overcome low \gls{RC} values or adverse \gls{CSI} conditions, resulting in a $ 92\% $ feasibility rate. In contrast, {\mytextsf{BL1}} produces feasible outcomes in only $ 27 $ realizations ($ 54\% $), significantly underperforming compared to {\mytextsf{OPT}}. This performance gap stems from the baseline's limited ability to account for the complexity of real-world conditions. Its reliance on a fixed, single criterion fails to capture contextual factors, which are essential for making optimal pairing decisions.


\subsection{Scenario VII: Impact of imperfect information and comparison with baseline methods}

In this scenario, we evaluate the impact of imperfect information on energy and time consumption. Additionally, we compare our proposed approach against five baselines that differ in their strategies for user-target pairing, beam adaptation, timeslot functionality selection, and in whether they optimize for energy or time minimization alone. The scenario considers $ U = 6 $, $ T = 4 $, $ S_\mathrm{com} = 3 $, $ S_\mathrm{sen} = 3 $, with the remaining parameters detailed in \textbf{Table~\ref{tab:simulation-settings}}.

\Cref{fig:results-scenario-6a} illustrates the energy and time consumption of the proposed approach under both perfect and imperfect information conditions, plotted as a function of $ \Lambda_\mathrm{sinr} $. To capture performance variability, representative individual error levels in \gls{AOD} ($ \epsilon_{\mathrm{AOD}} = 3^\circ $), \gls{RC} ($ \epsilon_{\mathrm{RC}} = 0.3 \widebar{\psi} $), \gls{RSI} ($ \upsilon = 0.2 $, $ \epsilon_{\mathrm{RSI}} = 0.6 \big\| \widebar{\mathbf{R}} \big\|_\mathrm{F} $), and \gls{CSI} ($ \epsilon_{\mathrm{CSI}} = 14 \sigma_\mathrm{com} $) are considered, along with a scenario that incorporates all sources of uncertainty ($ \epsilon_{\mathrm{AOD}} = 2^\circ $, $ \epsilon_{\mathrm{RC}} = 0.1 \widebar{\psi} $, $ \upsilon = 0.2 $, $ \epsilon_{\mathrm{RSI}} = 0.3 \big\| \widebar{\mathbf{R}} \big\|_\mathrm{F} $, $ \epsilon_{\mathrm{CSI}} = 4 \sigma_\mathrm{com} $). Separate plots are provided for energy and time consumption, along with a third plot depicting their combined consumption. Uncertainty is counteracted through increased energy usage, extended transmission time, or a combination of both, depending on the severity and nature of the errors. Given the higher priority assigned to minimizing time consumption, the system initially mitigates uncertainty by increasing transmit power, deferring the activation of extra timeslots as a secondary measure. As a result, energy consumption exhibits pronounced non-monotonic fluctuations as $ \Lambda_\mathrm{sinr} $ increases, whereas time consumption grows gradually and monotonically. This occurs because the activation of additional timeslots (as a last resort to preserve feasibility) enables individual servicing of users or targets, which can often be accomplished with narrower beams and reduced transmit power. Consequently, energy consumption may drop even as time usage increases, reflecting a transition to a different operating point.


\Cref{fig:results-scenario-6b} compares the performance of six schemes across three plots: energy consumption, time consumption, and a joint metric combining both, similarly to \cref{fig:results-scenario-6a}. The benchmarked schemes are as follows:

\noindent {\mytextsf{OPT}}: The proposed approach prioritizes minimizing time consumption over energy. It ensures optimal timeslot allocation, pairing, beam adaptation, and functionality selection.

\noindent {\mytextsf{TLB}}: This baseline serves as a time lower bound, as it considers only time minimization in its objective function, completely disregarding energy consumption.

\noindent {\mytextsf{ELB}}: This baseline serves as an energy lower bound, as it considers only energy minimization in its objective function, entirely disregarding time consumption.

\noindent {\mytextsf{BL1}}: A baseline which, like {\mytextsf{OPT}}, prioritizes time over energy. However, it relies on heuristic user-target pairing, adapted from \cite{dou2024:channel-sharing-aided-integrated-sensing-communication-energy-efficient-sensing-scheduling-approach}, enforcing always concurrent service.

\noindent {\mytextsf{BL2}}: A baseline which, like {\mytextsf{OPT}}, prioritizes time over energy, but has limited beam adaptation, as it uses constant transmit power and beamwidth, similar to \cite{xiao2024:simultaneous-multi-beam-sweeping-mmwave-massive-mimo-integrated-sensing-communication, rahman2019:joint-communication-radar-sensing-5g-mobile-network-compressive-sensing}.

\noindent {\mytextsf{BL3}}: This baseline separates communication and sensing into distinct timeslots, disallowing joint operation, as in \cite{li2022:multi-point-integrated-sensing-communication-fusion-model-functionality-selection}.

\begin{remark} \label{rem:note-1}
	To mitigate infeasibility, particularly in less flexible schemes such as {\mytextsf{BL1}}, which often yield a high number of infeasible allocations, a simple fallback strategy is employed. In particular, users and targets are assigned to disjoint timeslots and served using maximum transmit power with the narrowest beam directed along the \gls{LoS} or \gls{AOD} direction. If this allocation satisfies all the constraints, it is accepted as a feasible solution. Otherwise, the realization is excluded from the analysis.
\end{remark}

In terms of time consumption, {\mytextsf{TLB}} performs best, as it completely ignores energy consumption. {\mytextsf{OPT}} closely follows, achieving comparable time efficiency with a performance gap of less than $ 2\% $, while simultaneously maintaining energy awareness. {\mytextsf{ELB}} performs poorly in terms of time consumption, as time is excluded from its optimization objective. Any favorable outcomes occur coincidentally, typically when users and targets are well aligned, enabling timeslot sharing. {\mytextsf{BL1}} performs well at low values of $ \Lambda_{\mathrm{sinr}} $, but its effectiveness deteriorates rapidly as $ \Lambda_{\mathrm{sinr}} $ increases. This decline stems from its rigid pairing, which becomes increasingly infeasible under stricter sensing requirements and the demand for joint servicing. In such cases, fallback reallocation (see \emph{\cref{rem:note-1}}) leads to additional resource use. {\mytextsf{BL2}} underperforms due to its fixed transmit power and beamwidth configuration, which limits adaptability to varying alignment conditions. In particular, the use of narrow beams, unless  user and target are highly aligned, impedes joint servicing, frequently requiring additional timeslots to meet system requirements. {\mytextsf{BL3}} yields the worst performance in terms of time, as it enforces disjoint communication and sensing by design, resulting in consistently higher time consumption. At $ \Lambda_{\mathrm{sinr}} = 6 $, compared to {\mytextsf{OPT}}, the additional time consumption incurred by {\mytextsf{BL1}}, {\mytextsf{BL2}}, {\mytextsf{ELB}}, and {\mytextsf{BL3}} is $ 13\% $, $ 43\% $, $ 45\% $, and $ 49\% $, respectively.


Regarding energy consumption, {\mytextsf{ELB}} achieves the best performance, as it exclusively optimizes for energy, completely disregarding time consumption. Despite prioritizing time minimization, the proposed approach {\mytextsf{OPT}} still performs competitively, maintaining a low energy consumption. As expected, {\mytextsf{TLB}} exhibits poor energy performance since energy is not accounted for in its optimization objective. {\mytextsf{BL1}} incurs significantly higher energy consumption than {\mytextsf{OPT}}, primarily due to suboptimal user-target pairings that demand more transmit power than the optimal case. {\mytextsf{BL2}} shows degraded energy performance as it operates under a fixed power and beamwidth regime, which limits adaptability to changing conditions. Although it always transmits at a constant power level, its energy consumption is not flat across $ \Lambda_{\mathrm{sinr}} $ values. This is because occasional alignment between users and targets can still enable beam sharing, yielding energy savings, though this effect is primarily confined to low $ \Lambda_{\mathrm{sinr}} $ values. Lastly, {\mytextsf{BL3}} achieves excellent energy performance by decoupling communication and sensing into separate timeslots. This separation enables the use of highly directional beams and minimal power tailored to each function, making it more energy-efficient than joint servicing. However, this advantage comes at the cost of time efficiency, as {\mytextsf{BL3}} consistently employs the maximum number of timeslots, an unfavorable outcome given the problem's emphasis on minimizing time consumption. We also show the joint consumption of time and energy, demonstrating that {\mytextsf{OPT}} achieves the best trade-off. It outperforms {\mytextsf{ELB}}, {\mytextsf{TLB}}, {\mytextsf{BL3}}, {\mytextsf{BL1}}, and {\mytextsf{BL2}} by $ 20\% $, $ 24\% $, $ 25\% $, $ 34\% $, and $ 88\% $, respectively.
\begin{figure}[!t]
 	\begin{subfigure}[b]{0.48\columnwidth}
		\begin{center}
			\includegraphics[]{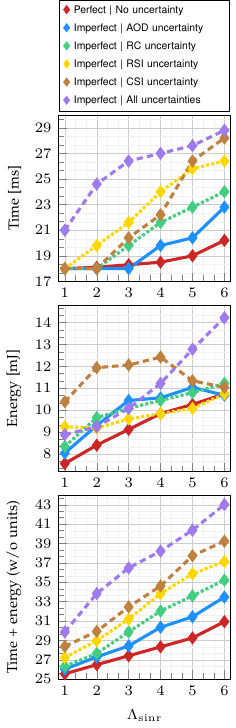}
			\caption{ Impact of imperfect information}
			\label{fig:results-scenario-6a}
		\end{center}
 	\end{subfigure}
 	\hfill 
 	\begin{subfigure}[b]{0.48\columnwidth}
		\begin{center}
			\includegraphics[]{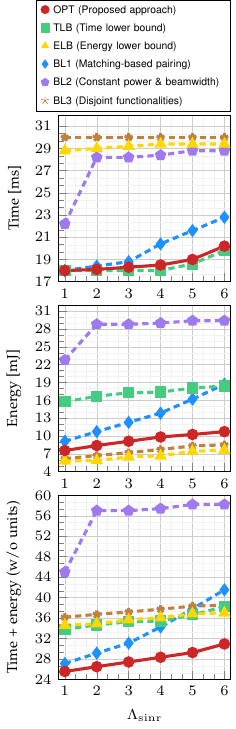}
			\caption{Comparison with baseline methods }
			\label{fig:results-scenario-6b}
		\end{center}
 	\end{subfigure}
 	\vspace{-2mm}
    \caption{Impact of imperfect information on resource economy and comparison with baseline methods (\emph{Scenario VII}). \emph{{\mytextsf{OPT}} is effective under both perfect and imperfect information conditions. Despite prioritizing time minimization, it preserves energy efficiency and is robust against uncertainties. Compared to five baseline schemes, it achieves the most favorable balance between time and energy consumption. These results highlight the critical role of optimized timeslot allocation, dynamic user-target pairing, adaptive beamforming, and flexible functionality selection in enabling efficient joint communication and sensing.}}
 	\label{fig:results-scenario-6}
 	\vspace{-4mm}
\end{figure}




\section{Conclusions} \label{sec:conclusions}

This work investigated a comprehensive \gls{RRM} framework for \gls{ISAC} systems, in which discrete power control, beam direction, and beamwidth at both the transmitter and receiver are jointly optimized. The proposed design further incorporates timeslot allocation, dynamic user-target pairing, and flexible functionality selection, all under imperfect information arising from estimation errors in \gls{AOD}, \gls{RC}, \gls{CSI}, and \gls{RSI}. The objective is to jointly minimize time and energy consumption, with a strict prioritization of time efficiency. The resulting \gls{RRM} problem is formulated as a multi-objective, semi-infinite, nonconvex \gls{MINLP}. By exploiting the underlying problem structure and revealing hidden convexity, we derived an exact reformulation as a \gls{MISDP}, which enables the computation of globally optimal solutions using general-purpose solvers. Through extensive simulations, we demonstrated the individual and joint impact of the key design components, namely timeslot allocation, adaptive beam configuration, dynamic user-target pairing, and timeslot-level functionality selection. The results show that the proposed framework effectively allocates resources to preserve feasibility under increasingly stringent requirements and imperfect information conditions. In particular, feasibility is primarily maintained through adaptive transmit power increases, with timeslot extension serving as a secondary mechanism, in accordance with the imposed prioritization of time minimization. Compared to a broad range of baseline schemes, the proposed approach consistently achieves the most favorable tradeoff between time and energy consumption. It closely approaches the performance of the time-optimal lower bound while substantially outperforming energy-centric and non-adaptive baselines. Overall, this work provides a versatile and practically grounded optimization framework for \gls{ISAC} systems, offering valuable design insights for future high-frequency networks operating under time-critical, energy-constrained, and imperfect-information conditions.  

\section*{Acknowledgment} \label{sec:acknowledgment}
This research was supported by the Federal Ministry of Research, Technology and Space (BMFTR) under Grant 16KIS2411.

\bibliographystyle{IEEEtran}
\bibliography{IEEEabrv,ref}

\vspace*{-2\baselineskip}

\begin{IEEEbiography}[{\includegraphics[width=1in,height=1.25in,clip,keepaspectratio]{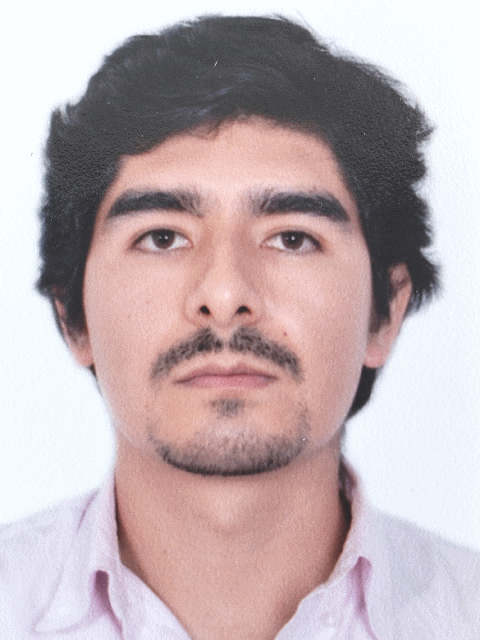}}]{Luis F. Abanto-Leon} received the M.Sc. degree in communications engineering from Tohoku University, Japan, in 2015, and the Ph.D. degree in computer science from Technische Universität Darmstadt, Germany, in 2023. Since March 2024, he has been a Postdoctoral Researcher with the Faculty of Electrical Engineering and Information Technology, Ruhr-Universität Bochum, Germany. His research expertise include optimization theory, signal processing, and machine learning, with a focus on algorithm design for radio resource management in 5G/6G wireless networks.
\end{IEEEbiography}

\vskip -2\baselineskip plus -1fil

\begin{IEEEbiography}[{\includegraphics[width=1in,height=1.25in,clip,keepaspectratio]{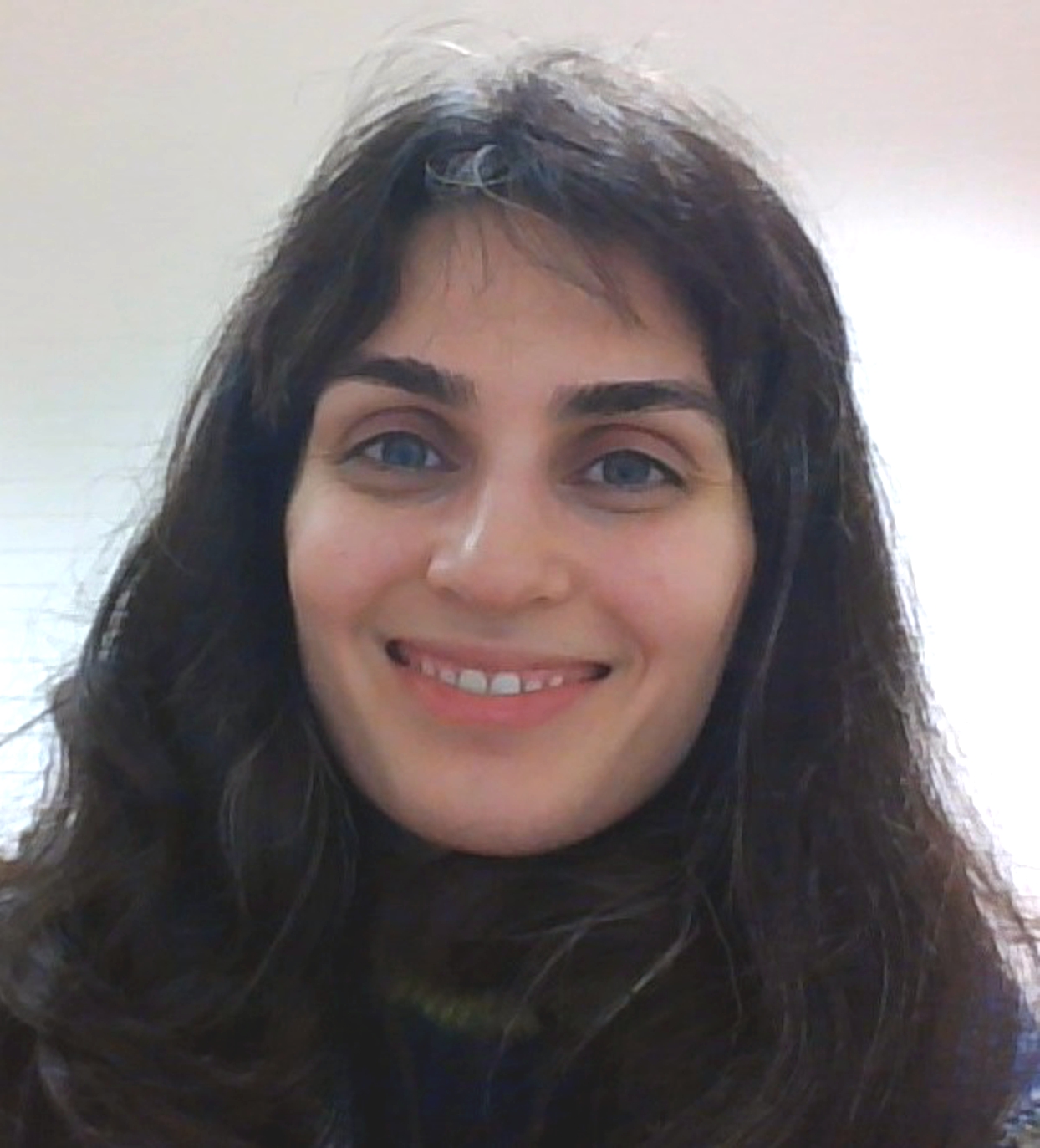}}]{Setareh Maghsudi} received the Ph.D. degree (summa cum laude) from the Technical University of Berlin, Berlin, Germany, in 2015. From 2015 to 2017, she was a Postdoctoral with the University of Manitoba, Winnipeg, MB, Canada, and Yale University, New Haven, CT, USA. From 2017 to 2023, she was an Assistant Professor with the Technical University of Berlin, Berlin, Germany, and Tübingen University, Tübingen, Germany. She is currently a full Professor with the Ruhr-University of Bochum, Bochum, Germany, and a Senior Research Scientist with the Fraunhofer Heinrich-Hertz Institute. Her research interests include the intersection of network analysis and optimization, game theory, machine learning, and data science. She was the recipient several competitive fellowships, awards, and research grants from different institutes, including the German Research Foundation, German Ministry of Education and Research, and Japan Society for the Promotion of Science.
\end{IEEEbiography}

\clearpage


\begin{appendices}

\renewcommand{\thesectiondis}[2]{\Alph{section}:}

\setcounter{equation}{0}
\setcounter{table}{0}
\renewcommand{\theequation}{A.\arabic{equation}}
\renewcommand{\thetable}{A.\arabic{table}}

\begin{table*}[!t]
	\begin{center}
	\renewcommand{\arraystretch}{1.3}
	\fontsize{6}{5.5}\selectfont
	\setlength\tabcolsep{2.8pt}
	\centering
	\caption{Categorization of most relevant related work.}
	\label{tab:related-literature}
	\begin{adjustbox}{width=1\textwidth,center}
	\begin{tabular}{c c c c c c c c c c c c c c c c} 
	\toprule
	\multirow{3}{*}{Works} & 
	\multirow{3}{*}{System} & 
	\multirow{3}{*}{\makecell{Integration \\ level}} & 
	\multirow{3}{*}{\makecell{Time \\ allocation}} & 
	\multirow{3}{*}{\makecell{Beamforming \\ type}} & 
	\multicolumn{4}{c}{Beam adaptation} & 
	\multirow{3}{*}{\makecell{Functionality \\ selection}} & 
	\multirow{3}{*}{\makecell{User-target \\ pairing}} & 
	\multicolumn{4}{c}{Imperfect information} & 
	\multirow{3}{*}{\makecell{Resource \\ economy }}  
	\\
	\cmidrule(lr){6-9}
	\cmidrule(lr){12-15}
	& 
	& 
	& 
	&
	& 
	\makecell{Direction} & 
	\makecell{Power} & 
	\makecell{Beamwidth} & 
	\makecell{Side} & 
	& 
	& 
	\makecell{CSI} & 
	\makecell{AOD} & 
	\makecell{RC} & 
	\makecell{RSI} & 
	\\
	\midrule
	
	\cite{xiao2024:simultaneous-multi-beam-sweeping-mmwave-massive-mimo-integrated-sensing-communication, rahman2019:joint-communication-radar-sensing-5g-mobile-network-compressive-sensing} &
	\multicolumn{1}{>{\columncolor{green!20}}c}{\tabularCenterstack{c}{\textbf{ISAC}}} &
	\multicolumn{1}{>{\columncolor{green!20}}c}{\tabularCenterstack{c}{\textbf{Single} \\ \textbf{waveform}}} &
	N/A &
	\multicolumn{1}{>{\columncolor{green!20}}c}{\tabularCenterstack{c}{\textbf{Analog}}} &
	\multicolumn{1}{>{\columncolor{green!20}}c}{\tabularCenterstack{c}{\textbf{Adaptive} \\ \textbf{\& discrete}}} &
	Fixed &
	Fixed &
	\multicolumn{1}{>{\columncolor{green!20}}c}{\tabularCenterstack{c}{\textbf{Tx \& Rx}}} &
	N/A &
	Fixed &
	\xmark &
	\xmark &
	\xmark &
	\xmark &
	None
	\\
	\midrule

	\cite{wang2024:resource-allocation-isac-networks-application-target-tracking} &
	\multicolumn{1}{>{\columncolor{green!20}}c}{\tabularCenterstack{c}{\textbf{ISAC}}} &
	\multicolumn{1}{>{\columncolor{green!20}}c}{\tabularCenterstack{c}{\textbf{Single} \\ \textbf{waveform}}} &
	\multicolumn{1}{>{\columncolor{green!20}}c}{\tabularCenterstack{c}{\textbf{Discrete} \\ \textbf{multi-slot}}} &
	\multicolumn{1}{>{\columncolor{green!20}}c}{\tabularCenterstack{c}{\textbf{Analog}}} &
	\multicolumn{1}{>{\columncolor{green!20}}c}{\tabularCenterstack{c}{\textbf{Adaptive} \\ \textbf{\& discrete}}} &
	\makecell{Adaptive \\ \& continuous} &
	Fixed &
	\multicolumn{1}{>{\columncolor{green!20}}c}{\tabularCenterstack{c}{\textbf{Tx \& Rx}}} &
	N/A &
	Fixed &
	\xmark &
	\xmark &
	\xmark &
	\xmark &
	None
	\\
	\midrule

	\cite{hersyandika2024:guard-beams-coverage-enhancement-ue-centered-isac-analog-multi-beamforming, ding2018:beam-index-modulation-wireless-communication-analog-beamforming} &
	\multicolumn{1}{>{\columncolor{green!20}}c}{\tabularCenterstack{c}{\textbf{ISAC}}} &
	\multicolumn{1}{>{\columncolor{green!20}}c}{\tabularCenterstack{c}{\textbf{Single} \\ \textbf{waveform}}} &
	N/A &
	\multicolumn{1}{>{\columncolor{green!20}}c}{\tabularCenterstack{c}{\textbf{Analog}}} &
	\multicolumn{1}{>{\columncolor{green!20}}c}{\tabularCenterstack{c}{\textbf{Adaptive} \\ \textbf{\& discrete}}} &
	\makecell{Adaptive \\ \& continuous} &
	Fixed &
	\multicolumn{1}{>{\columncolor{green!20}}c}{\tabularCenterstack{c}{\textbf{Tx \& Rx}}} &
	N/A &
	Fixed &
	\xmark &
	\xmark &
	\xmark &
	\xmark &
	None
	\\
	\midrule

	\cite{karacora2023:event-based-beam-tracking-dynamic-beamwidth-adaptation-terahertz-communications, feng2021:beamwidth-optimization-5g-nr-millimeter-wave-cellular-networks-multi-armed-bandit-approach, chung2021:adaptive-beamwidth-control-mmwave-beam-tracking} &
	Comm. &
	N/A &
	N/A &
	\multicolumn{1}{>{\columncolor{green!20}}c}{\tabularCenterstack{c}{\textbf{Analog}}} &
	Fixed &
	Fixed &
	\makecell{Adaptive \\ \& continuous} &
	\multicolumn{1}{>{\columncolor{green!20}}c}{\tabularCenterstack{c}{\textbf{Tx \& Rx}}} &
	N/A &
	N/A &
	\xmark &
	\xmark &
	\xmark &
	\xmark &
	None 
	\\
	\midrule

	\cite{du2023:integrated-sensing-communications-v2i-networks-dynamic-predictive-beamforming-extended-vehicle-targets, zhang2023:robust-beamforming-design-uav-communications-integrated-sensing-communication} &
	\multicolumn{1}{>{\columncolor{green!20}}c}{\tabularCenterstack{c}{\textbf{ISAC}}} &
	\multicolumn{1}{>{\columncolor{green!20}}c}{\tabularCenterstack{c}{\textbf{Single} \\ \textbf{waveform}}} &
	N/A &
	\multicolumn{1}{>{\columncolor{green!20}}c}{\tabularCenterstack{c}{\textbf{Analog}}} &
	Fixed &
	Fixed &
	\makecell{Adaptive \\ \& continuous} &
	\multicolumn{1}{>{\columncolor{green!20}}c}{\tabularCenterstack{c}{\textbf{Tx \& Rx}}} &
	N/A &
	Fixed &
	\xmark &
	\xmark &
	\xmark &
	\xmark &
	None 
	\\
	\midrule
	
	\cite{liu2020:cognitive-dwell-time-allocation-distributed-radar-sensor-networks-tracking-cone-programming, zhang2023:fast-solver-dwell-time-allocation-phased-array-radar} &
	Sens. &
	N/A &
	Continuous &
	Digital &
	Fixed &
	\makecell{Adaptive \\ \& continuous} &
	Fixed &
	\multicolumn{1}{>{\columncolor{green!20}}c}{\tabularCenterstack{c}{\textbf{Tx \& Rx}}} &
	N/A &
	N/A &
	\xmark &
	\xmark &
	\xmark &
	\xmark &
	None
	\\
	\midrule

	\cite{abanto2020:hydrawave-multi-group-multicast-hybrid-precoding-low-latency-scheduling-ubiquitous-industry-4-0-mmwave-communications, nguyen2017:joint-fractional-time-allocation-beamforming-downlink-multiuser-miso-systems} &
	Comm. &
	N/A &
	Continuous &
	\makecell{Digital} & 
	\multicolumn{3}{c}{\makecell{Beamformer design not subject \\ to practical considerations}} & 
	Tx &
	N/A &
	N/A &
	\xmark &
	\xmark &
	\xmark &
	\xmark &
	\multicolumn{1}{>{\columncolor{green!20}}c}{\tabularCenterstack{c}{\textbf{Time}}}
	\\
	\midrule

	\cite{xu2023:sensing-enhanced-secure-communication-joint-time-allocation-beamforming-design} &
	\multicolumn{1}{>{\columncolor{green!20}}c}{\tabularCenterstack{c}{\textbf{ISAC}}} &
	\makecell{Independent \\ waveforms} &
	Continuous &
	\makecell{Digital} &
	\multicolumn{3}{c}{\makecell{Beamformer design not subject \\ to practical considerations}} & 
	Tx &
	N/A &
	Fixed &
	\xmark &
	\multicolumn{1}{>{\columncolor{green!20}}c}{\tabularCenterstack{c}{\bcmark}} &
	\multicolumn{1}{>{\columncolor{green!20}}c}{\tabularCenterstack{c}{\bcmark}} &
	\xmark &
	None 
	\\
	\midrule

	\cite{li2022:multi-point-integrated-sensing-communication-fusion-model-functionality-selection} &
	\multicolumn{1}{>{\columncolor{green!20}}c}{\tabularCenterstack{c}{\textbf{ISAC}}} &
	\makecell{Single \\ waveform} &
	N/A &
	Digital &
	Fixed &
	\makecell{Adaptive \\ \& continuous} &
	Fixed &
	\multicolumn{1}{>{\columncolor{green!20}}c}{\tabularCenterstack{c}{\textbf{Tx \& Rx}}} &
	\makecell{One at a \\ time} &
	Fixed &
	\xmark &
	\xmark &
	\xmark &
	\xmark &
	None
	\\
	\midrule

	\cite{abanto2024:hierarchical-functionality-prioritization-multicast-isac-optimal-admission-control-discrete-phase-beamforming} &
	\multicolumn{1}{>{\columncolor{green!20}}c}{\tabularCenterstack{c}{\textbf{ISAC}}} &
	\multicolumn{1}{>{\columncolor{green!20}}c}{\tabularCenterstack{c}{\textbf{Single} \\ \textbf{waveform}}} &
	N/A &
	\multicolumn{1}{>{\columncolor{green!20}}c}{\tabularCenterstack{c}{\textbf{Analog}}} &
	\makecell{Adaptive \\ \& discrete} &
	Fixed &
	Fixed &
	\multicolumn{1}{>{\columncolor{green!20}}c}{\tabularCenterstack{c}{\textbf{Tx \& Rx}}} &
	\makecell{One is \\ prioritized} &
	Fixed &
	\xmark &
	\multicolumn{1}{>{\columncolor{green!20}}c}{\tabularCenterstack{c}{\bcmark}} &
	\xmark &
	\xmark &
	None 
	\\
	\midrule

	\cite{dou2024:channel-sharing-aided-integrated-sensing-communication-energy-efficient-sensing-scheduling-approach} &
	\multicolumn{1}{>{\columncolor{green!20}}c}{\tabularCenterstack{c}{\textbf{ISAC}}} &
	\makecell{Independent \\ waveforms} &
	N/A &
	Digital &
	\multicolumn{3}{c}{\makecell{Beamformer design not subject \\ to practical considerations}} & 
	\multicolumn{1}{>{\columncolor{green!20}}c}{\tabularCenterstack{c}{\textbf{Tx \& Rx}}} &
	N/A &
	\multicolumn{1}{>{\columncolor{green!20}}c}{\tabularCenterstack{c}{\textbf{Dynamic}}} &
	\xmark &
	\xmark &
	\xmark &
	\xmark &
	None 
	\\
	\midrule

	\cite{cazzella2024:deep-learning-based-target-to-user-association-integrated-sensing-communication-systems} &
	\multicolumn{1}{>{\columncolor{green!20}}c}{\tabularCenterstack{c}{\textbf{ISAC}}} &
	\multicolumn{1}{>{\columncolor{green!20}}c}{\tabularCenterstack{c}{\textbf{Single} \\ \textbf{waveform}}} &
	N/A &
	\makecell{Hybrid} &
	\multicolumn{3}{c}{\makecell{Beamformer design not subject \\ to practical considerations}} & 
	\multicolumn{1}{>{\columncolor{green!20}}c}{\tabularCenterstack{c}{\textbf{Tx \& Rx}}} &
	N/A &
	\multicolumn{1}{>{\columncolor{green!20}}c}{\tabularCenterstack{c}{\textbf{Dynamic}}} &
	\xmark &
	\xmark &
	\xmark &
	\xmark &
	None
	\\
	\midrule

	\cite{abanto2024:optimal-user-target-scheduling-user-target-pairing-low-resolution-phase-only-beamforming-isac-systems} &
	\multicolumn{1}{>{\columncolor{green!20}}c}{\tabularCenterstack{c}{\textbf{ISAC}}} &
	\multicolumn{1}{>{\columncolor{green!20}}c}{\tabularCenterstack{c}{\textbf{Single} \\ \textbf{waveform}}} &
	N/A &
	\multicolumn{1}{>{\columncolor{green!20}}c}{\tabularCenterstack{c}{\textbf{Analog}}} &
	\makecell{Adaptive \\ \& continuous} &
	Fixed &
	\makecell{Adaptive \\ \& continuous} &
	\multicolumn{1}{>{\columncolor{green!20}}c}{\tabularCenterstack{c}{\textbf{Tx \& Rx}}} &
	N/A &
	\multicolumn{1}{>{\columncolor{green!20}}c}{\tabularCenterstack{c}{\textbf{Dynamic}}} &
	\xmark &
	\xmark &
	\xmark &
	\xmark &
	None 
	\\
	\midrule

	\cite{he2023:full-duplex-communication-isac-joint-beamforming-power-optimization} &
	\multicolumn{1}{>{\columncolor{green!20}}c}{\tabularCenterstack{c}{\textbf{ISAC}}} &
	\makecell{Independent \\ waveforms} &
	N/A &
	Digital &
	\multicolumn{3}{c}{\makecell{Beamformer design not subject \\ to practical considerations}} & 
	\multicolumn{1}{>{\columncolor{green!20}}c}{\tabularCenterstack{c}{\textbf{Tx \& Rx}}} &
	N/A &
	N/A &
	\xmark &
	\xmark &
	\xmark &
	\xmark &
	\multicolumn{1}{>{\columncolor{green!20}}c}{\tabularCenterstack{c}{\textbf{Energy}}} 
	\\
	\midrule

	\cite{huang2022:coordinated-power-control-network-integrated-sensing-communication} &
	\multicolumn{1}{>{\columncolor{green!20}}c}{\tabularCenterstack{c}{\textbf{ISAC}}} &
	\multicolumn{1}{>{\columncolor{green!20}}c}{\tabularCenterstack{c}{\textbf{Single} \\ \textbf{waveform}}} &
	N/A &
	N/A &
	\multicolumn{3}{c}{\makecell{Power control involving multiple \\ single-antenna BSs}} & 
	Tx &
	N/A &
	Fixed &
	\xmark &
	\xmark &
	\xmark &
	\xmark &
	\multicolumn{1}{>{\columncolor{green!20}}c}{\tabularCenterstack{c}{\textbf{Energy}}}  
	\\
	\midrule

	\cite{jia2024:physical-layer-security-optimization-cramer-rao-bound-metric-isac-systems-sensing-specific-imperfect-csi-model} &
	\multicolumn{1}{>{\columncolor{green!20}}c}{\tabularCenterstack{c}{\textbf{ISAC}}} &
	\makecell{Independent \\ waveforms} &
	N/A &
	Digital &
	\multicolumn{3}{c}{\makecell{Beamformer design not subject \\ to practical considerations}} & 
	\multicolumn{1}{>{\columncolor{green!20}}c}{\tabularCenterstack{c}{\textbf{Tx \& Rx}}} &
	N/A &
	N/A &
	\xmark &
	\xmark &
	\xmark &
	\xmark &
	None 
	\\
	\midrule

	\cite{lyu2024:dual-robust-integrated-sensing-communication-beamforming-csi-imperfection-location-uncertainty} &
	\multicolumn{1}{>{\columncolor{green!20}}c}{\tabularCenterstack{c}{\textbf{ISAC}}} &
	\makecell{Independent \\ waveforms} &
	N/A &
	Digital &
	\multicolumn{3}{c}{\makecell{Beamformer design not subject \\ to practical considerations}} & 
	Tx &
	N/A &
	N/A &
	\multicolumn{1}{>{\columncolor{green!20}}c}{\tabularCenterstack{c}{\bcmark}} &
	\multicolumn{1}{>{\columncolor{green!20}}c}{\tabularCenterstack{c}{\bcmark}} &
	\xmark &
	\xmark &
	None
	\\
	\midrule

	\cite{khalili2024:advanced-isac-design-movable-antennas-accounting-dynamic-rcs} &
	\multicolumn{1}{>{\columncolor{green!20}}c}{\tabularCenterstack{c}{\textbf{ISAC}}} &
	\makecell{Independent \\ waveforms} &
	N/A &
	Digital &
	\multicolumn{3}{c}{\makecell{Beamformer design not subject \\ to practical considerations}} & 
	\multicolumn{1}{>{\columncolor{green!20}}c}{\tabularCenterstack{c}{\textbf{Tx \& Rx}}} &
	N/A &
	N/A &
	\xmark &
	\xmark &
	\multicolumn{1}{>{\columncolor{green!20}}c}{\tabularCenterstack{c}{\bcmark}} &
	\xmark &
	\multicolumn{1}{>{\columncolor{green!20}}c}{\tabularCenterstack{c}{\textbf{Energy}}}  
	\\
	\midrule

	\cite{tang2021:self-interference-resistant-ieee-802-11-ad-based-joint-communication-automotive-radar-design} &
	\multicolumn{1}{>{\columncolor{green!20}}c}{\tabularCenterstack{c}{\textbf{ISAC}}} &
	\multicolumn{1}{>{\columncolor{green!20}}c}{\tabularCenterstack{c}{\textbf{Single} \\ \textbf{waveform}}} &
	N/A &
	\multicolumn{1}{>{\columncolor{green!20}}c}{\tabularCenterstack{c}{\textbf{Analog}}} &
	Fixed &
	Fixed &
	Fixed &
	\multicolumn{1}{>{\columncolor{green!20}}c}{\tabularCenterstack{c}{\textbf{Tx \& Rx}}} &
	N/A &
	Fixed &
	\xmark &
	\xmark &
	\xmark &
	\multicolumn{1}{>{\columncolor{green!20}}c}{\tabularCenterstack{c}{\bcmark}} &
	None 
	\\
	\midrule

	\textbf{Proposed} &
	\multicolumn{1}{>{\columncolor{green!20}}c}{\tabularCenterstack{c}{\textbf{ISAC}}} &
	\multicolumn{1}{>{\columncolor{green!20}}c}{\tabularCenterstack{c}{\textbf{Single} \\ \textbf{waveform}}} &
	\multicolumn{1}{>{\columncolor{green!20}}c}{\tabularCenterstack{c}{\textbf{Discrete} \\ \textbf{multi-slot}}} &
	\multicolumn{1}{>{\columncolor{green!20}}c}{\tabularCenterstack{c}{\textbf{Analog}}} &
	\multicolumn{1}{>{\columncolor{green!20}}c}{\tabularCenterstack{c}{\textbf{Adaptive} \\ \textbf{\& discrete}}} & 
	\multicolumn{1}{>{\columncolor{green!20}}c}{\tabularCenterstack{c}{\textbf{Adaptive} \\ \textbf{\& discrete}}} & 
	\multicolumn{1}{>{\columncolor{green!20}}c}{\tabularCenterstack{c}{\textbf{Adaptive} \\ \textbf{\& discrete}}} & 
	\multicolumn{1}{>{\columncolor{green!20}}c}{\tabularCenterstack{c}{\textbf{Tx \& Rx}}} &
	\multicolumn{1}{>{\columncolor{green!20}}c}{\tabularCenterstack{c}{\textbf{None, one,} \\ \textbf{or two}}} &
	\multicolumn{1}{>{\columncolor{green!20}}c}{\tabularCenterstack{c}{\textbf{Dynamic}}} &
	\multicolumn{1}{>{\columncolor{green!20}}c}{\tabularCenterstack{c}{\bcmark}} &
	\multicolumn{1}{>{\columncolor{green!20}}c}{\tabularCenterstack{c}{\bcmark}} &
	\multicolumn{1}{>{\columncolor{green!20}}c}{\tabularCenterstack{c}{\bcmark}} &
	\multicolumn{1}{>{\columncolor{green!20}}c}{\tabularCenterstack{c}{\bcmark}} &
	\multicolumn{1}{>{\columncolor{green!20}}c}{\tabularCenterstack{c}{\textbf{Energy} \\ \textbf{\& time}}} 
	\\
	\bottomrule
	\end{tabular}
	\end{adjustbox}
	\end{center}
	\vspace{-4mm}
\end{table*}

\section{Related Work} \label{app:related-work}

\Cref{tab:related-literature} (see next page) provides a concise visual summary of the key distinctions between the problem addressed in this work and existing studies in the literature, clearly highlighting the substantial differences.

In particular, compared with prior work, our approach uniquely integrates timeslot allocation into the \gls{RRM} design, enables adaptive beam control across multiple dimensions (direction, transmit/receive power, and beamwidth), incorporates functionality selection at the timeslot level, and allows for optimizable user-target pairing.

\setcounter{equation}{0}
\setcounter{table}{0}
\renewcommand{\theequation}{B.\arabic{equation}}
\renewcommand{\thetable}{B.\arabic{table}}
\section{Main Parameters} \label{app:table-parameters}

\cref{tab:parameters-variables} summarizes the most relevant parameters and variable used in the system model.
\begin{table}[H]
 \centering
	\scriptsize
	\caption{Parameters and variables}
	\begin{tabular}{|m{6.7cm} |c|}
		\hline
		\centering {\bf Parameters and Variables} & \bf Notation	\\ 
		\hline
		Number of transmit antennas at the BS  		& $ N_\mathrm{tx} $ \\ 
		Number of receive antennas at the BS  		& $ N_\mathrm{rx} $ \\ 
		Number of transmit codewords at the BS		& $ L_\mathrm{tx} $ \\ 
		Number of receive codewords at the BS		& $ L_\mathrm{rx} $ \\ 
		Number of users										& $ U $ \\
		Number of targets									& $ T $ \\
		Finite horizon									& $ S $ \\
		Set indexing the transmit codewords 			& $ \mathcal{L}_\mathrm{tx} $ \\
		Set indexing the receive codewords 			& $ \mathcal{L}_\mathrm{rx} $ \\
		Set indexing the users  						& $ \mathcal{U} $ \\
		Set indexing the targets 					& $ \mathcal{T} $ \\
		Set indexing the timeslots 						& $ \mathcal{S} $ \\
		Set of channel vectors with errors for $ \mathsf{U}_u $ 			& $ \mathcal{H}_u $ \\
		Set of angles of arrival with errors for $ \mathsf{T}_t $ & $ \Theta_t $ \\
		Set of reflection coefficients with errors for $ \mathsf{T}_t $ & $ \Psi_t $ \\
		Communication channel between the BS and $ \mathsf{U}_u $  			& $ \mathbf{h}_u $ \\
		Sensing channel between the BS and $ \mathsf{T}_t $ 	    	& $ \mathbf{G}_t $ \\
		Residual self-interference channel of the BS's transceiver   	& $ \mathbf{R} $ \\
		Communication SNR threshold for $ \mathsf{U}_u $			& $ \Upsilon_{\mathrm{snr},u} $ \\
		Sensing SINR threshold for $ \mathsf{T}_t $						& $ \Lambda_{\mathrm{sinr},t} $ \\
		Number of timeslots allocated to $ \mathsf{U}_u $			& $ S_{\mathrm{com},u} $ \\
		Number of timeslots allocated to $ \mathsf{T}_t $			& $ S_{\mathrm{sen},t} $ \\
		Communication noise power 							& $ \sigma^2_\mathrm{com} $ \\ 
		Sensing noise power 								& $ \sigma^2_\mathrm{sen} $ \\ 
		Timeslot duration 								& $ S_\mathrm{dur} $ \\ 
		\hline
		Variable indicating whether $ \mathsf{S}_s $ is used for communication & $ \kappa_s $ \\
		Variable indicating whether $ \mathsf{S}_s $ is used for sensing & $ \zeta_s $ \\
		Variable indicating whether $ \mathsf{S}_s $ is used & $ \gamma_s $ \\
		Variable indicating whether $ \mathsf{U}_u $	is served in $ \mathsf{S}_s $ & $ \mu_{u,s} $ \\
		Variable indicating whether $ \mathsf{T}_t $	is served in $ \mathsf{S}_s $ & $ \tau_{t,s} $ \\
		Variable indicating whether $ \textbf{b}_b $ is used in $ \mathsf{S}_s $ & $ \chi_{b,s} $ \\
		Variable indicating whether $ \textbf{c}_c $ is used in $ \mathsf{S}_s $ & $ \rho_{c,s} $ \\
		Transmit beamformer in $ \mathsf{S}_s $			& $ \mathbf{w}_s $ \\
		Receive beamformer in $ \mathsf{S}_s $				& $ \mathbf{v}_s $ \\
		\hline
	\end{tabular}
	\label{tab:parameters-variables}
\end{table}

\setcounter{equation}{0}
\setcounter{table}{0}
\renewcommand{\theequation}{C.\arabic{equation}}
\section{Details on parameters $ S $ and $ \widetilde{S}$} \label{app:finite-horizon}

Without considering pairing, the total number of timeslots required for communication is $ \hat{S} = \sum_{u \in \mathcal{U}} S_{\mathrm{com},u} $ and the total number of timeslots required for sensing is $ \check{S} = \sum_{t \in \mathcal{T}} S_{\mathrm{sen},t} $. 

In the worst case, the \gls{BS} requires up to $ \widetilde{S} = \hat{S} + \check{S} $ timeslots, which corresponds to independent timeslot allocation for users and targets without any pairing. In contrast, the minimum number of required timeslots is $ \bar{S} = \max \big\{ \hat{S}, \check{S} \big\} $ which is achieved when the maximum possible number of user-target pairs share the same timeslots. 

Consequently, the finite scheduling horizon must satisfy $ \bar{S} \leq S \leq \widetilde{S} $. In practice, we intentionally choose $ S \ll \widetilde{S} $, since a smaller finite horizon naturally promotes user-target pairing and promotes timeslot sharing whenever feasible.


\setcounter{equation}{0}
\setcounter{table}{0}
\renewcommand{\theequation}{D.\arabic{equation}}
\section{Details on user-target pairing} \label{app:pairing-details}

Constraints $ \mathrm{C}_{1} $, $ \mathrm{C}_{2} $, $ \mathrm{C}_{5} $, $ \mathrm{C}_{6} $, $ \mathrm{C}_{7} $, and $ \mathrm{C}_{8} $ ensure one-to-one pairing between users and targets whenever feasible, as explained in the following.

For any timeslot $ \mathsf{S}_s $, constraints $ \mathrm{C}_{1} $, $ \mathrm{C}_{2} $, $ \mathrm{C}_{6} $, and $ \mathrm{C}_{8} $ enforce $ \sum_{u \in \mathcal{U}} \mu_{u,s} \in \left\lbrace 0, 1 \right\rbrace $ and $ \sum_{t \in \mathcal{T}} \tau_{t,s} \in \left\lbrace 0, 1 \right\rbrace $. If either or both summations are zero, no pairing occurs in $ \mathsf{S}_s $ since both sensing and communication must be feasible in a timeslot for pairing to be established.  When both summations equal one, constraints $ \mathrm{C}_{5} $ and $ \mathrm{C}_{7} $ ensure that exactly one variable $ \mu_{u,s} $ and one variable $ \tau_{t,s} $ are equal to one. This implies that exactly one user and one target are served in $ \mathsf{S}_s $, establishing a one-to-one pairing between them.

Importantly, pairing may vary across timeslots, as the user and target requirements can differ. Importantly, pairing is not predetermined but rather optimized through \gls{RRM} design, being established only when providing time or energy savings.

\setcounter{equation}{0}
\setcounter{table}{0}
\renewcommand{\theequation}{E.\arabic{equation}}
\section{Extended Model} \label{app:linearized-model}

From (\ref{eqn:response-matrix}) and (\ref{eqn:steering-vectors}), we have the following relation
\begin{align}
	\mathbf{A} \left( \theta_t \right) = \tfrac{1}{\sqrt{N_\mathrm{rx} N_\mathrm{tx}}}  \mathrm{e}^{\mathrm{j} \boldsymbol{\phi}_\mathrm{rx} \cos \left( \theta_t \right)} \left( \mathrm{e}^{\mathrm{j} \boldsymbol{\phi}_\mathrm{tx} \cos \left( \theta_t \right)} \right)^\mathrm{H}
\end{align}

Particularly, note that $ \mathrm{e}^{\mathrm{j} \boldsymbol{\phi}_\mathrm{rx} \cos \left( \theta_t \right)} \left( \mathrm{e}^{\mathrm{j} \boldsymbol{\phi}_\mathrm{tx} \cos \left( \theta_t \right)} \right)^\mathrm{H} = \mathrm{e}^{\mathrm{j} \boldsymbol{\phi}_\mathrm{rx} \cos \left( \theta_t \right)} \left( \mathrm{e}^{-\mathrm{j} \boldsymbol{\phi}_\mathrm{tx} \cos \left( \theta_t \right)} \right)^\mathrm{T} = \mathrm{e}^{\mathrm{j} \boldsymbol{\phi}_\mathrm{rx} \cos \left( \theta_t \right)} \mathrm{e}^{-\mathrm{j} \boldsymbol{\phi}_\mathrm{tx}^\mathrm{T} \cos \left( \theta_t \right)} $. To simplify the expression, we define $ \boldsymbol{\Phi} = \mathbf{1}^\mathrm{T} \otimes \boldsymbol{\phi}_\mathrm{rx} - \mathbf{1} \otimes \boldsymbol{\phi}_\mathrm{tx}^\mathrm{T} $, which allows us to express $ \mathbf{A} \left( \theta_t \right) $ as $ \mathbf{A} \left( \theta_t \right) = \mathrm{e}^{\mathrm{j} \boldsymbol{\Phi}  \cos \left( \theta_t \right)} $. Taking the first-order Taylor expansion of $ \mathbf{A} \left( \theta_t \right) $ around $ \widebar{\theta}_t $, results in
\begin{align}
	\mathbf{A} \left( \theta_t \right) \approx \mathbf{A} \left( \widebar{\theta}_t \right) + \nabla_{\theta_t} \mathbf{A} \left( \theta_t \right) \Bigr\rvert_{\theta_t = \widebar{\theta}_t} \left( \theta_t - \widebar{\theta}_t \right),
\end{align}
where $ \nabla_{\theta_t} \mathbf{A} \left( \theta_t \right) = \mathbf{A} \left( {\theta}_t \right) \circ \mathbf{E} \left( {\theta}_t \right) $ and $ \mathbf{E} \left( {\theta}_t \right) = - \mathrm{j} \boldsymbol{\Phi} \sin \left( {\theta}_t \right) $. Noting that $ \Delta \theta_t = \theta_t - \widebar{\theta}_t $ and defining $ \widetilde{\mathbf{A}} \left( {\theta}_t \right) = \mathbf{A} \left( {\theta}_t \right) \circ \mathbf{E} \left( {\theta}_t \right) $ leads to
\begin{align}
	\mathbf{A} \left( \theta_t \right) \approx  \mathbf{A} \left( \widebar{\theta}_t \right) + \widetilde{\mathbf{A}} \left( \widebar{\theta}_t \right) \Delta \theta_t, 
\end{align}
which is the same expression shown in (\ref{eqn:linearized-response-matrix}).

\setcounter{equation}{0}
\setcounter{table}{0}
\renewcommand{\theequation}{F.\arabic{equation}}
\section{Weight Design} \label{app:proof-lemma-weights}

To ensure cohesive timeslot allocation, it is essential to prioritize consecutive timeslots while giving higher precedence to the earliest available timeslots. Using mathematical induction, we derive a simple weighting policy based on an arithmetic sequence.

Given $ S $ timeslots, we first assume that the \gls{RRM} task is fulfilled using a single timeslot, i.e., $ \sum_{s \in \mathcal{S}} \gamma_s = 1 $. The optimal allocation satisfying the cohesiveness requirement is given by $ \gamma_1 = 1 $ and $ \gamma_s = 0 $, $ \forall s \in \mathcal{S} \setminus \{1\} $, yielding $ \omega_1 < \omega_{i_1} $, $ \forall i_1 \in \{2, 3, \dots, S\} $.  

Next, we consider the case where the RRM task is fulfilled using two timeslots, i.e., $ \sum_{s \in \mathcal{S}} \gamma_s = 2 $. The most cohesive allocation is $ \gamma_1 = \gamma_2 = 1 $ and $ \gamma_s = 0 $, $ \forall s \in \mathcal{S} \setminus \{1,2\} $. This implies that $ \omega_1 + \omega_2 < \omega_{i_1} + \omega_{i_2} $ for all pairs $ (i_1, i_2) \neq (1,2) $. To tighten this inequality, we assume $ \omega_1 = \omega_{i_1} $, leading to $ \omega_2 < \omega_{i_2} $, $ \forall i_2 \in \{3, 4, \dots, S\} $.  

Following the induction process, we generalize to the case where all but one timeslot is allocated, i.e., $ \sum_{s \in \mathcal{S}} \gamma_s = S - 1 $. The optimal allocation is $ \gamma_1 = \gamma_2 = \dots = \gamma_{S-1} = 1 $ and $ \gamma_S = 0 $, resulting in $ \omega_1 + \omega_2 + \dots + \omega_{S-1} < \omega_{i_1} + \omega_{i_2} + \dots + \omega_{i_{S-1}}, \quad \forall (i_1, i_2, \dots, i_{S-1}) \neq (1,2,\dots,S-1). $
To refine this inequality, we assume $ \omega_s = \omega_{i_s} $, $ \forall s \in \{1, 2, \dots, S-2\} $, yielding $ \omega_{S-1} < \omega_S $.  

Collecting all results, we derive the general rule that achieving cohesive timeslot allocation requires satisfying $ \omega_s < \omega_{s+1} $. To generalize this rule in a structured manner, we adopt an arithmetic sequence, defined as
$ \omega_s = \Delta_0 + (s-1) \cdot \Delta_\omega $, where $ \Delta_0 $ is the initial value and $ \Delta_\omega $ represents the increment.

\setcounter{equation}{0}
\setcounter{table}{0}
\renewcommand{\theequation}{G.\arabic{equation}}
\section{Proof to Proposition 1} \label{app:proof-proposition-1}

Leveraging the binary nature of $ \kappa_s $ and $ \zeta_s $, we apply propositional calculus to eliminate the binding logical operation between these variables. Thus, we introduce the substitutions $ a = \gamma_s $, $ b = \kappa_s $, and $ \zeta_s = c $ to simplify notation. The steps involved in the decoupling process are shown below.
\begin{align*} 
	& a \leftrightarrow b \lor c, 
	\\
	& \equiv \left( a \rightarrow \left( b \lor c \right) \right) \wedge \left( \left( b \lor c \right) \rightarrow a \right), 
	\\
	& \equiv \left( \neg a \lor \left( b \lor c \right) \right) \wedge \left( \neg \left( b \lor c \right) \lor a \right), 
	\\
	& \equiv \left( \neg a \lor b \lor c \right) \wedge \left( \left( \neg b \wedge \neg c \right) \lor a \right),
	\\
	& \equiv \left( \neg a \lor b \lor c \right) \wedge \left( \left( \neg b \lor a \right) \wedge \left( \neg c \lor a \right) \right), 
\end{align*} 
\begin{align*} 
	& \equiv \left( \neg a \lor b \lor c \right) \wedge \left( \neg b \lor a \right) \wedge \left( \neg c \lor a \right), 
	\\
	& \equiv \left( \left( 1 - a \right) + b + c \geq 1 \right) \wedge \left( \left( 1 - b \right) + a \geq 1 \right)
	\\
	& ~~~~ \wedge \left( \left( 1 - c \right) + a \geq 1 \right),
	\\
	& \equiv \left( a \leq b + c \right) \wedge \left( a \geq b \right) \wedge \left( a \geq c \right).
\end{align*}

We know that $ a $ must be binary since it results from a logical operation. Yet, this condition is not always enforced by the inequalities above, particularity, when $ b = c = 1 $. While adding constraint $ a \in \{0,1\} $ would prevent this issue, an alternative approach is to define $ a $ within a continuous set, as specified by constraint $ a \leq 1 $. This choice maintains the integrality of $ a $ while reducing the number of binary variables, thereby lowering the complexity associated with binary variable search.

Hence, from the results above, constraint $ \mathrm{C}_{3} $ can be equivalently expressed as the intersection of $ \mathrm{D}_{1}: \gamma_s \leq \kappa_s + \zeta_s, \forall s \in \mathcal{S} $, $ \mathrm{D}_{2}: \gamma_s \geq \kappa_s, \forall s \in \mathcal{S} $, $ \mathrm{D}_{3}: \gamma_s \geq \zeta_s, \forall s \in \mathcal{S} $, and $ \mathrm{D}_{4}: \gamma_s \leq 1, \forall s \in \mathcal{S} $.

\setcounter{equation}{0}
\setcounter{table}{0}
\renewcommand{\theequation}{H.\arabic{equation}}
\section{Proof to Proposition 2} \label{app:proof-proposition-2}

Using the definitions in (\ref{eqn:communication-snr}), and substituting $ \mathbf{w}_s $, as defined in $ \mathrm{C}_{11} $, into  $ \mathrm{C}_{16} $ leads to constraint $ \mathrm{E}_{\mathrm{aux},1}: \min_{ \mathbf{h}_u \in \mathcal{H}_u } \left| \sum_{b \in \mathcal{L}_\mathrm{tx}} \mathbf{h}_u^\mathrm{H} \mathbf{Z}_\mathrm{tx} \mathbf{b}_b \cdot \chi_{b,s} \cdot \mu_{u,s} \right|^2 \geq \Upsilon_{\mathrm{snr},u} \cdot \sigma_\mathrm{com}^2 \cdot \mu_{u,s}, \forall u \in \mathcal{U}, s \in \mathcal{S} $. By analyzing $ \mu_{u,s} $ in $ \mathrm{E}_{\mathrm{aux},1} $, we distinguish two distinct cases, shown in the following.

\noindent \textbf{Case \circled{\footnotesize{1}} $\Rightarrow$ } When $ \mu_{u,s} = 0 $, the inequality is tight as both sides collapse to zero.

\noindent \textbf{Case \circled{\footnotesize{2}} $\Rightarrow$ } When $ \mu_{u,s} = 1 $, the inequality collapses to $ \min_{ \mathbf{h}_u \in \mathcal{H}_u } \left| \sum_{b \in \mathcal{L}_\mathrm{tx}} \mathbf{h}_u^\mathrm{H} \mathbf{Z}_\mathrm{tx} \mathbf{b}_b \cdot \chi_{b,s} \right|^2 \geq \Upsilon_{\mathrm{snr},u} \cdot \sigma_\mathrm{com}^2, \forall u \in \mathcal{U}, s \in \mathcal{S} $.

The two cases above can be combined into a single constraint, leading to  $ \mathrm{E}_{\mathrm{aux},2}: \min_{ \mathbf{h}_u \in \mathcal{H}_u } \left| \sum_{b \in \mathcal{L}_\mathrm{tx}} \mathbf{h}_u^\mathrm{H} \mathbf{Z}_\mathrm{tx} \mathbf{b}_b \cdot \chi_{b,s} \right|^2  \geq \Upsilon_{\mathrm{snr},u} \cdot \sigma_\mathrm{com}^2 \cdot \mu_{u,s}, \forall u \in \mathcal{U}, s \in \mathcal{S} $. Note that the \gls{LHS} of $ \mathrm{E}_{\mathrm{aux},2} $ is not necessarily zero when $ \mu_{u,s} = 0 $. This is because $ \mu_{u,s} = 0 $ only indicates that \( \mathsf{U}_{u} \) is not served in $ \mathsf{S}_s $, but it does not preclude $ \mathsf{S}_s $ from being used to serve another $ \mathsf{U}_{u'} $, where $ u' \neq u $. To improve the tractability of $ \mathrm{E}_{\mathrm{aux},2} $, where  variables $ \chi_{b,s} $ are coupled both additively and multiplicatively in its \gls{LHS}, we apply Jensen's inequality to uncover useful relationships.

Considering arbitrary $ \mathsf{U}_u $ and $ \mathsf{S}_s $, and applying Jensen's inequality to the absolute value term in $ \mathrm{E}_{\mathrm{aux},2} $, we obtain $ \mathrm{E}_{\mathrm{aux},3}: \left| \sum_{b \in \mathcal{L}_\mathrm{tx}} \mathbf{h}_u^\mathrm{H} \mathbf{Z}_\mathrm{tx} \mathbf{b}_b \cdot \chi_{b,s} \right| \leq \sum_{b \in \mathcal{L}_\mathrm{tx}} \left| \mathbf{h}_u^\mathrm{H} \mathbf{Z}_\mathrm{tx} \mathbf{b}_b \cdot \chi_{b,s} \right| $. Let $ \mathbf{b}_{i} $ denote the beamformer serving $ \mathsf{U}_u $ in $ \mathsf{S}_s $. Since only one transmit beamformer can be selected per active timeslot (as enforced by constraint $ \mathrm{C}_{10} $), we have $ \chi_{i,s} = 1 $ and $ \chi_{b,s} = 0 $,  $ \forall b \in \mathcal{L}_\mathrm{tx} \setminus \{i\} $. With this observation, $ \mathrm{E}_{\mathrm{aux},3} $ simplifies to the equivalent form $ \mathrm{E}_{\mathrm{aux},4}: \left| \sum_{b \neq i} \mathbf{h}_u^\mathrm{H} \mathbf{Z}_\mathrm{tx} \mathbf{b}_b \cdot \chi_{b,s} + \mathbf{h}_u^\mathrm{H} \mathbf{Z}_\mathrm{tx} \mathbf{b}_{i} \cdot \chi_{i,s} \right| \leq \sum_{b \neq i} \left| \mathbf{h}_u^\mathrm{H} \mathbf{Z}_\mathrm{tx} \mathbf{b}_b \cdot \chi_{b,s} \right| + \left| \mathbf{h}_u^\mathrm{H} \mathbf{Z}_\mathrm{tx} \mathbf{b}_{i} \cdot \chi_{i,s} \right| $. Exponentiating both sides yields $ \mathrm{E}_{\mathrm{aux},5}: \left| \sum_{b \neq i} \mathbf{h}_u^\mathrm{H} \mathbf{Z}_\mathrm{tx} \mathbf{b}_b \cdot \chi_{b,s} + \mathbf{h}_u^\mathrm{H} \mathbf{Z}_\mathrm{tx} \mathbf{b}_{i} \cdot \chi_{i,s} \right|^2 \leq \left( \sum_{b \neq i} \left| \mathbf{h}_u^\mathrm{H} \mathbf{Z}_\mathrm{tx} \mathbf{b}_b \cdot \chi_{b,s} \right| \right)^2 + \left| \mathbf{h}_u^\mathrm{H} \mathbf{Z}_\mathrm{tx} \mathbf{b}_{i} \cdot \chi_{i,s} \right|^2 + 2 \left( \sum_{b \neq i} \left| \mathbf{h}_u^\mathrm{H} \mathbf{Z}_\mathrm{tx} \mathbf{b}_b \cdot \chi_{b,s} \right| \right) \left| \mathbf{h}_u^\mathrm{H} \mathbf{Z}_\mathrm{tx} \mathbf{b}_{i} \cdot \chi_{i,s} \right| $. Note that the terms under the summations are zero, which makes $ \mathrm{E}_{\mathrm{aux},5} $ hold at the equality, i.e., $ \left| \sum_{b \in \mathcal{L}_\mathrm{tx}} \mathbf{h}_u^\mathrm{H} \mathbf{Z}_\mathrm{tx} \mathbf{b}_b \cdot \chi_{b,s} \right|^2 = \sum_{b \in \mathcal{L}_\mathrm{tx}} \left| \mathbf{h}_u^\mathrm{H} \mathbf{Z}_\mathrm{tx} \mathbf{b}_b \cdot \chi_{b,s} \right|^2 $. This means that Jensen's inequality is tight when applied to $ \mathrm{E}_{\mathrm{aux},2} $, allowing us to equivalently express $ \mathrm{E}_{\mathrm{aux},2} $ as $ \mathrm{E}_{\mathrm{aux},6}: \min_{ \mathbf{h}_u \in \mathcal{H}_u } \sum_{b \in \mathcal{L}_\mathrm{tx}} \left| \mathbf{h}_u^\mathrm{H} \mathbf{Z}_\mathrm{tx} \mathbf{b}_b \cdot \chi_{b,s} \right|^2 \geq \Upsilon_{\mathrm{snr},u} \cdot \sigma_\mathrm{com}^2 \cdot \mu_{u,s}, \forall u \in \mathcal{U}, s \in \mathcal{S} $. Note that extracting the summation from the quadratic term leads to a expression with simplified structure.

Although $ \mathrm{E}_{\mathrm{aux},6} $ is more tractable than $ \mathrm{C}_{16} $, it still contains an infinite number of constraints due to all possible channels for $ \mathsf{U}_u $, as defined by $ \mathrm{C}_{15} $. Thus, we let $ \mathrm{E}_{\mathrm{aux},6} $ absorb $ \mathrm{C}_{15} $, leading to $ \mathrm{E}_{\mathrm{aux},7}:  \min_{ \left\| \boldsymbol{\Delta} \mathbf{h}_u \right\|_2^2 \leq \epsilon_\mathrm{CSI}^2 } \sum_{b \in \mathcal{L}_\mathrm{tx}} \frac{1}{\sigma_\mathrm{com}^2} \left| \left( \widebar{\mathbf{h}}_u^\mathrm{H} + \boldsymbol{\Delta} \mathbf{h}_u^\mathrm{H} \right) \mathbf{Z}_\mathrm{tx} \mathbf{b}_b \right|^2 \chi_{b,s} - \Upsilon_{\mathrm{snr},u} \cdot  \mu_{u,s} \geq 0, \forall u \in \mathcal{U}, s \in \mathcal{S} $, where $ \chi_{b,s} $ has been factored out of the absolute value since $ \chi_{b,s}^2 = \chi_{b,s} $. In particular, $ \mathrm{E}_{\mathrm{aux},7} $ can be split into constraints $ \mathrm{E}_{\mathrm{aux},8} $ and $ \mathrm{E}_{\mathrm{aux},9} $, defined as follows: $ \mathrm{E}_{\mathrm{aux},8}: \boldsymbol{\Delta} \mathbf{h}_u^\mathrm{H}  \boldsymbol{\Delta} \mathbf{h}_u - \epsilon_\mathrm{CSI}^2 \leq 0, \forall u \in \mathcal{U}, $ and $ \mathrm{E}_{\mathrm{aux},9}: \boldsymbol{\Delta} \mathbf{h}_u^\mathrm{H} \mathbf{L}_s \boldsymbol{\Delta} \mathbf{h}_u +  \widebar{\mathbf{h}}_u^\mathrm{H} \mathbf{L}_s \boldsymbol{\Delta} \mathbf{h}_u + \boldsymbol{\Delta} \mathbf{h}_u^\mathrm{H} \mathbf{L}_s \widebar{\mathbf{h}}_u + \widebar{\mathbf{h}}_u^\mathrm{H} \mathbf{L}_s \widebar{\mathbf{h}}_u + \Upsilon_{\mathrm{snr},u} \cdot  \mu_{u,s} \leq 0, \forall u \in \mathcal{U}, s \in \mathcal{S} $, where $ \mathbf{L}_s = - \sum_{b \in \mathcal{L}_\mathrm{tx}} \frac{1}{\sigma_\mathrm{com}^2}  \mathbf{Z}_\mathrm{tx} \mathbf{b}_b \mathbf{b}_b^\mathrm{H} \mathbf{Z}_\mathrm{tx}^\mathrm{H} \chi_{b,s} $. 

We apply \textbf{Lemma \ref{lem:s-procedure}} to $ \mathrm{E}_{\mathrm{aux},8} $ and $ \mathrm{E}_{\mathrm{aux},9} $, leading to constraints $ \mathrm{E}_{1} $ and $ \mathrm{E}_{2} $, shown below,
\begin{align*}
			   	& \mathrm{E}_{1}: \left[ 
 			   					\begin{matrix}
 			   					   	\mathbf{E}_{u,s}^\mathrm{a}
 			   					   	& 
 			   					   	\mathbf{e}_{u,s}^\mathrm{b}
 			   					   	\\
	   					    		\mathbf{e}_{u,s}^\mathrm{c}
	   					    		& 
	   					    		e_{u,s}^\mathrm{d}
 			   				  	\end{matrix} 
			   					\right] 
					   		\succcurlyeq \mathbf{0}, \forall u \in \mathcal{U}, s \in \mathcal{S},
				\\
				& \mathrm{E}_{2}: \alpha_u \geq 0, \forall u \in \mathcal{U},
\end{align*}
where $ \mathbf{E}_{u,s}^\mathrm{a} = \alpha_u \mathbf{I} - \mathbf{L}_s $, $ \mathbf{e}_{u,s}^\mathrm{b} = - \mathbf{L}_s \widebar{\mathbf{h}}_u $, $ \mathbf{e}_{u,s}^\mathrm{c} = - \widebar{\mathbf{h}}_u^\mathrm{H} \mathbf{L}_s $, $ e_{u,s}^\mathrm{d} =  -\alpha_u \cdot \epsilon_\mathrm{CSI}^2 - \widebar{\mathbf{h}}_u^\mathrm{H} \mathbf{L}_s \widebar{\mathbf{h}}_u - \Upsilon_{\mathrm{snr},u} \cdot \mu_{u,s} $, and $ \alpha_u $ are newly introduced auxiliary variables needed by the \emph{S-Procedure}.

\setcounter{equation}{0}
\setcounter{table}{0}
\renewcommand{\theequation}{I.\arabic{equation}}
\section{Proof to Proposition 3} \label{app:proof-proposition-3}

The numerator and denominator of $ \mathsf{SINR}_\mathrm{sen}^{t,s} $, as defined in (\ref{eqn:sensing-sinr}), exhibit couplings between $ \mathbf{w}_s $, $ \mathbf{v}_s $,  and $ \tau_{t,s} $. By substituting the definitions of $ \mathbf{w}_s $ and $ \mathbf{v}_s $, as specified in $ \mathrm{C}_{11} $ and $ \mathrm{C}_{14} $, into $ \mathrm{C}_{21} $, we obtain constraint $ \mathrm{F}_{\mathrm{aux},1}:  \min_{\psi_t \in \Psi_t} \min_{\theta_t \in \Theta_t} \min_{\mathbf{R} \in \mathcal{R}} $ \resizebox{\linewidth}{!}{$ \frac{ \left| \sum_{c \in \mathcal{L}_\mathrm{rx}} \sum_{b \in \mathcal{L}_\mathrm{tx}}  \mathbf{c}_c^\mathrm{H} \mathbf{Z}_\mathrm{rx}^\mathrm{H} \mathbf{G}_t \mathbf{Z}_\mathrm{tx} \mathbf{b}_b \cdot \chi_{b,s} \cdot \rho_{c,s} \cdot \tau_{t,s} \right|^2}{ \left| \sum_{c \in \mathcal{L}_\mathrm{rx}} \sum_{b \in \mathcal{L}_\mathrm{tx}}  \mathbf{c}_c^\mathrm{H} \mathbf{Z}_\mathrm{rx}^\mathrm{H} \mathbf{R} \mathbf{Z}_\mathrm{tx} \mathbf{b}_b \cdot \chi_{b,s} \cdot \rho_{c,s} \cdot \tau_{t,s} \right|^2 + \sigma_\mathrm{sen}^2 \sum_{c \in \mathcal{L}_\mathrm{rx}} \left\| \mathbf{Z}_\mathrm{rx} \mathbf{c}_c \right\|_2^2 \rho_{c,s}} $} $  \geq \Lambda_{\mathrm{sinr},t} \cdot \tau_{t,s}, \forall t \in \mathcal{T}, s \in \mathcal{S} $, revealing a triple coupling between $ \chi_{b,s} $, $ \rho_{c,s} $, and $ \tau_{t,s} $, complicating tractability\footnote{The noise term in the denominator, expressed originally as $ \sigma_\mathrm{sen}^2 \left\| \mathbf{v}_s \right\|_2^2  $, can be equivalently recast as $ \sigma_\mathrm{sen}^2 \sum_{c \in \mathcal{L}_\mathrm{rx}} \left\| \mathbf{Z}_\mathrm{rx} \mathbf{c}_c \right\|_2^2 \rho_{c,s} $. Further elaboration on the transformation of this quadratic term can be found in \textbf{Appendix \ref{app:proof-proposition-8}}.}.

Thus, we introduce new variables $ \pi_{b,c,t,s} $, defined in constraint $ \mathrm{F}_{\mathrm{aux},2}: \pi_{b,c,t,s} = \chi_{b,s} \cdot \rho_{c,s} \cdot \tau_{t,s}, \forall b \in \mathcal{L}_\mathrm{tx}, c \in \mathcal{L}_\mathrm{rx}, t \in \mathcal{T}, s \in \mathcal{S} $, and use propositional calculus to get rid of the coupling between $ \chi_{b,s} $, $ \rho_{c,s} $, and $ \tau_{t,s} $. This process results in constraints $ \mathrm{F}_{1} $, $ \mathrm{F}_{2} $, $ \mathrm{F}_{3} $, and $ \mathrm{F}_{4} $, which are derived as shown at the top of the next page.
\begin{figure*} [!t]
	\fontsize{9}{12}\selectfont
	\begin{align*} 
		& \pi_{b,c,t,s} \leftrightarrow \chi_{b,s} \cdot \rho_{c,s} \cdot \tau_{t,s}, \forall b \in \mathcal{L}_\mathrm{tx}, c \in \mathcal{L}_\mathrm{rx}, t \in \mathcal{T}, s \in \mathcal{S},
		\\
		& \equiv \left( \pi_{b,c,t,s} \rightarrow \left( \chi_{b,s} \wedge \rho_{c,s} \wedge \tau_{t,s} \right) \right) \wedge \left( \left( \chi_{b,s} \wedge \rho_{c,s} \wedge \tau_{t,s} \right) \rightarrow \pi_{b,c,t,s} \right), \forall b \in \mathcal{L}_\mathrm{tx}, c \in \mathcal{L}_\mathrm{rx}, t \in \mathcal{T}, s \in \mathcal{S},  
		\\
		& \equiv \left( \neg \pi_{b,c,t,s} \lor \left( \chi_{b,s} \wedge \rho_{c,s} \wedge \tau_{t,s} \right) \right) \wedge \left( \neg \left( \chi_{b,s} \wedge \rho_{c,s} \wedge \tau_{t,s} \right) \lor \pi_{b,c,t,s} \right), \forall b \in \mathcal{L}_\mathrm{tx}, c \in \mathcal{L}_\mathrm{rx}, t \in \mathcal{T}, s \in \mathcal{S},    
		\\
		& \equiv \left(  \left( \neg \pi_{b,c,t,s} \lor \chi_{b,s} \right) \wedge \left( \neg \pi_{b,c,t,s} \lor \rho_{c,s} \right) \right) \wedge \left( \neg \pi_{b,c,t,s} \lor \tau_{t,s} \right) \wedge \left( \left( \neg \chi_{b,s} \lor \neg \rho_{c,s} \lor \neg \tau_{t,s} \right) \lor \pi_{b,c,t,s} \right), \forall b \in \mathcal{L}_\mathrm{tx}, c \in \mathcal{L}_\mathrm{rx}, t \in \mathcal{T}, s \in \mathcal{S},    
		\\
		& \equiv \left( \left( 1 - \pi_{b,c,t,s} \right) + \chi_{b,s} \geq 1 \right) \wedge \left( \left( 1 - \pi_{b,c,t,s} \right) + \rho_{c,s} \geq 1 \right) \wedge \left( \left( 1 - \pi_{b,c,t,s} \right) + \tau_{t,s} \geq 1 \right) \wedge 
		\\
		& ~~~ \left( \left( 1 - \chi_{b,s} \right) + \left( 1 - \rho_{c,s} \right) + \left( 1 - \tau_{t,s} \right) + \pi_{b,c,t,s} \geq 1 \right), \forall b \in \mathcal{L}_\mathrm{tx}, c \in \mathcal{L}_\mathrm{rx}, t \in \mathcal{T}, s \in \mathcal{S},  
		\\
		& \equiv \underbrace{\left( \pi_{b,c,t,s} \leq \chi_{b,s} \right)}_{\mathrm{F}_1} \wedge \underbrace{\left( \pi_{b,c,t,s} \leq \rho_{c,s} \right)}_{\mathrm{F}_2} \wedge \underbrace{\left( \pi_{b,c,t,s} \leq \tau_{t,s} \right)}_{\mathrm{F}_3} \wedge \underbrace{\left( 2 + \pi_{b,c,t,s} \geq \chi_{b,s} + \rho_{c,s} + \tau_{t,s} \right)}_{\mathrm{F}_4}, \forall b \in \mathcal{L}_\mathrm{tx}, c \in \mathcal{L}_\mathrm{rx}, t \in \mathcal{T}, s \in \mathcal{S}.  
	\end{align*} 
	\hrulefill
\end{figure*}

Additionally, note that variables $ \pi_{b,c,t,s} $ must be binary, which is not always enforced by $ \mathrm{F}_{1} $, $ \mathrm{F}_{2} $, $ \mathrm{F}_{3} $, and $ \mathrm{F}_{4} $. While we could explicitly define $ \pi_{b,c,t,s} $ as binary, we instead define these variables within a continuous closed set, as specified in constraint $ \mathrm{F}_{5} : \pi_{b,c,t,s} \in [0,1] $, $ \forall b \in \mathcal{L}_\mathrm{tx}, c \in \mathcal{L}_\mathrm{rx}, t \in \mathcal{T}, s \in \mathcal{S} $. This formulation is valid because the integrality of $ \pi_{b,c,t,s} $ is inherently preserved, given that the resulting inequalities intersect only at $ 0 $ or $ 1 $, thereby reducing the number of explicitly defined binary variables. Substituting $ \pi_{b,c,t,s} $ into $ \mathrm{F}_{\mathrm{aux},1} $ leads to $ \mathrm{F}_{6}: \min_{\psi_t \in \Psi_t} \min_{\theta_t \in \Theta_t} \min_{\mathbf{R} \in \mathcal{R}} $ $ 
 	\frac
 		{ 
 		\left| \sum_{c \in \mathcal{L}_\mathrm{rx}} \sum_{b \in \mathcal{L}_\mathrm{tx}}  \mathbf{c}_c^\mathrm{H} \mathbf{Z}_\mathrm{rx}^\mathrm{H} \mathbf{G}_t \mathbf{Z}_\mathrm{tx} \mathbf{b}_b \cdot \pi_{b,c,t,s} \right|^2 
 		}
 		{ 
 		\left| \sum_{c \in \mathcal{L}_\mathrm{rx}} \sum_{b \in \mathcal{L}_\mathrm{tx}}  \mathbf{c}_c^\mathrm{H} \mathbf{Z}_\mathrm{rx}^\mathrm{H} \mathbf{R} \mathbf{Z}_\mathrm{tx} \mathbf{b}_b \cdot \pi_{b,c,t,s} \right|^2 + \sigma_\mathrm{sen}^2  \sum_{c \in \mathcal{L}_\mathrm{rx}} \left\| \mathbf{Z}_\mathrm{rx} \mathbf{c}_c \right\|_2^2 \rho_{c,s}
 		} \geq  $ $
 	\Lambda_{\mathrm{sinr},t} \cdot \tau_{t,s}, \forall t \in \mathcal{T}, s \in \mathcal{S} $. As a result, $ \mathrm{C}_{11} $, $ \mathrm{C}_{14} $, and $ \mathrm{C}_{21} $ are equivalently recast as the intersection of $ \mathrm{F}_{1} $, $ \mathrm{F}_{2} $, $ \mathrm{F}_{3} $, $ \mathrm{F}_{4} $, $ \mathrm{F}_{5} $, and $ \mathrm{F}_{6} $.

\setcounter{equation}{0}
\setcounter{table}{0}
\renewcommand{\theequation}{J.\arabic{equation}}
\section{Proof to Proposition 4} \label{app:proof-proposition-4}

By leveraging (\ref{eqn:response-matrix}), we equivalently express $ \mathrm{F}_{6} $ as constraint $ \mathrm{G}_{\mathrm{aux},1}: \min_{\psi_t \in \Psi_t} \left| \psi_t \right|^2 \Big\{ \min_{\theta_t \in \Theta_t} \min_{\mathbf{R} \in \mathcal{R}} $ $  \frac{ 
	\left| \sum_{c \in \mathcal{L}_\mathrm{rx}} \sum_{b \in \mathcal{L}_\mathrm{tx}}  \mathbf{c}_c^\mathrm{H} \mathbf{Z}_\mathrm{rx}^\mathrm{H} \mathbf{A} \left( \theta_t \right) \mathbf{Z}_\mathrm{tx} \mathbf{b}_b \cdot \pi_{b,c,t,s} \right|^2 } 
	{ \left| \sum_{c \in \mathcal{L}_\mathrm{rx}} \sum_{b \in \mathcal{L}_\mathrm{tx}}  \mathbf{c}_c^\mathrm{H} \mathbf{Z}_\mathrm{rx}^\mathrm{H} \mathbf{R} \mathbf{Z}_\mathrm{tx} \mathbf{b}_b \cdot \pi_{b,c,t,s} \right|^2 + \sigma_\mathrm{sen}^2 \sum_{c \in \mathcal{L}_\mathrm{rx}} \left\| \mathbf{Z}_\mathrm{rx} \mathbf{c}_c \right\|_2^2 \rho_{c,s} 
	} \Big\} $ $ \geq \Lambda_{\mathrm{sinr},t} \cdot \tau_{t,s}, \forall t \in \mathcal{T}, s \in \mathcal{S} $, since $ \left| \psi_t \right|^2  $ depends only on $ t $, and can therefore be repositioned within the expression. Further manipulating $ \mathrm{G}_{\mathrm{aux},1} $ leads to $ \mathrm{G}_{\mathrm{aux},2}: \min_{\theta_t \in \Theta_t}  \min_{\mathbf{R} \in \mathcal{R}} $ $ \frac{ 
	\left| \sum_{c \in \mathcal{L}_\mathrm{rx}} \sum_{b \in \mathcal{L}_\mathrm{tx}}  \mathbf{c}_c^\mathrm{H} \mathbf{Z}_\mathrm{rx}^\mathrm{H} \mathbf{A} \left( \theta_t \right) \mathbf{Z}_\mathrm{tx} \mathbf{b}_b \cdot \pi_{b,c,t,s} \right|^2 } 
	{ \left| \sum_{c \in \mathcal{L}_\mathrm{rx}} \sum_{b \in \mathcal{L}_\mathrm{tx}}  \mathbf{c}_c^\mathrm{H} \mathbf{Z}_\mathrm{rx}^\mathrm{H} \mathbf{R} \mathbf{Z}_\mathrm{tx} \mathbf{b}_b \cdot \pi_{b,c,t,s} \right|^2 + \sigma_\mathrm{sen}^2 \sum_{c \in \mathcal{L}_\mathrm{rx}} \left\| \mathbf{Z}_\mathrm{rx} \mathbf{c}_c \right\|_2^2 \rho_{c,s}
	} \geq $ $ \frac{\Lambda_{\mathrm{sinr},t}}{\min_{\psi_t \in \Psi_t} \left| \psi_t \right|^2 } \cdot \tau_{t,s} , \forall t \in \mathcal{T}, s \in \mathcal{S} $.

To ensure that the sensing performance is met even under the worst-case \gls{RC}, we derive the corresponding worst-case \gls{SINR} threshold by analyzing the \gls{RHS} of $ \mathrm{G}_{\mathrm{aux},2} $. Particularly, the worst-case sensing \gls{SINR} threshold occurs when the denominator of the \gls{RHS} of $ \mathrm{G}_{\mathrm{aux},2} $ is minimum because it maximizes the ratio $ \frac{\Lambda_{\mathrm{sinr},t} }{\min_{\psi_t \in \Psi_t} \left| \psi_t \right|^2 } $, thereby imposing a more demanding constraint. Hence, let us consider $ \mathrm{G}_{\mathrm{aux},3}: \min_{\psi_t \in \Psi_t} \left| \psi_t \right|^2 $, which is equivalent to $ \mathrm{G}_{\mathrm{aux},4}: \min_{\left| \Delta \psi_t \right|^2 \leq \epsilon_\mathrm{RC}^2} \left| \widebar{\psi}_t + \Delta  \psi_t \right|^2 $ upon letting it absorb $ \mathrm{C}_{19} $.

Note that the quadratic term in $ \mathrm{G}_{\mathrm{aux},4} $ reaches its minimum when the unknown error $ \Delta \psi_t $ is aligned with the known value $ \widebar{\psi}_t $ but in opposite sense and has maximum power $ \epsilon_{\mathrm{RC}}^2 $, i.e., $ \Delta \psi_t = - \frac{\widebar{\psi}_t}{\left| \widebar{\psi}_t \right|} \epsilon_\mathrm{RC} $. Substituting this expression into $ \mathrm{G}_{\mathrm{aux},4} $ yields the worst-case sensing \gls{SINR} threshold, denoted as $ \widebar{\Lambda}_{\mathrm{sinr},t} = \frac{\Lambda_{\mathrm{sinr},t}}{\left( \left| \widebar{\psi}_t  \right| - \epsilon_{\mathrm{RC}} \right)^2 } $.

As a result, we can equivalently recast $ \mathrm{F}_{6} $ and $ \mathrm{C}_{19} $ as $ \mathrm{G}_{1}: \min_{\theta_t \in \Theta_t}  \min_{\mathbf{R} \in \mathcal{R}} $ $ \frac{ 
	\left| \sum_{c \in \mathcal{L}_\mathrm{rx}} \sum_{b \in \mathcal{L}_\mathrm{tx}}  \mathbf{c}_c^\mathrm{H} \mathbf{Z}_\mathrm{rx}^\mathrm{H} \mathbf{A} \left( \theta_t \right) \mathbf{Z}_\mathrm{tx} \mathbf{b}_b \cdot \pi_{b,c,t,s} \right|^2 } 
	{ \left| \sum_{c \in \mathcal{L}_\mathrm{rx}} \sum_{b \in \mathcal{L}_\mathrm{tx}}  \mathbf{c}_c^\mathrm{H} \mathbf{Z}_\mathrm{rx}^\mathrm{H} \mathbf{R} \mathbf{Z}_\mathrm{tx} \mathbf{b}_b \cdot \pi_{b,c,t,s} \right|^2 + \sigma_\mathrm{sen}^2 \sum_{c \in \mathcal{L}_\mathrm{rx}} \left\| \mathbf{Z}_\mathrm{rx} \mathbf{c}_c \right\|_2^2 \rho_{c,s}
	} \geq  \widebar{\Lambda}_{\mathrm{sinr},t} \cdot $ $ \tau_{t,s}, \forall t \in \mathcal{T}, s \in \mathcal{S} $.

\setcounter{equation}{0}
\setcounter{table}{0}
\renewcommand{\theequation}{K.\arabic{equation}}
\section{Proof to Proposition 5} \label{app:proof-proposition-5}

Since \gls{AOD} errors appear only in the numerator and \gls{RSI} errors only in the denominator, we can equivalently express $ \mathrm{G}_{1} $ as $ \mathrm{H}_{\mathrm{aux},1}: $ \resizebox{\linewidth}{!}{$  \frac{
	 \min_{\theta_t \in \Theta_t} \left| \sum_{c \in \mathcal{L}_\mathrm{rx}} \sum_{b \in \mathcal{L}_\mathrm{tx}}  \mathbf{c}_c^\mathrm{H} \mathbf{Z}_\mathrm{rx}^\mathrm{H} \mathbf{A} \left( \theta_t \right) \mathbf{Z}_\mathrm{tx} \mathbf{b}_b \cdot \pi_{b,c,t,s} \right|^2 
	 } 
	 { \max_{\mathbf{R} \in \mathcal{R}} \left| \sum_{c \in \mathcal{L}_\mathrm{rx}} \sum_{b \in \mathcal{L}_\mathrm{tx}}  \mathbf{c}_c^\mathrm{H} \mathbf{Z}_\mathrm{rx}^\mathrm{H} \mathbf{R} \mathbf{Z}_\mathrm{tx} \mathbf{b}_b \cdot \pi_{b,c,t,s} \right|^2 + \sigma_\mathrm{sen}^2 \sum_{c \in \mathcal{L}_\mathrm{rx}} \left\| \mathbf{Z}_\mathrm{rx} \mathbf{c}_c \right\|_2^2 \rho_{c,s} 
	 } $} 
	 $ \geq \widebar{\Lambda}_{\mathrm{sinr},t} \cdot \tau_{t,s}, \forall t \in \mathcal{T}, s \in \mathcal{S} $. Further manipulation allows  $ \mathrm{H}_{\mathrm{aux},1} $ to be recast as $ \mathrm{H}_{\mathrm{aux},2}: \frac{1}{\sigma_\mathrm{sen}^2} \min_{\theta_t \in \Theta_t} $ $  \left| \sum_{c \in \mathcal{L}_\mathrm{rx}} \sum_{b \in \mathcal{L}_\mathrm{tx}}  \mathbf{c}_c^\mathrm{H} \mathbf{Z}_\mathrm{rx}^\mathrm{H} \mathbf{A} \left( \theta_t \right) \mathbf{Z}_\mathrm{tx} \mathbf{b}_b \cdot \pi_{b,c,t,s} \right|^2 \geq \frac{1}{\sigma_\mathrm{sen}^2} \widebar{\Lambda}_{\mathrm{sinr},t} $ $ \cdot \tau_{t,s}  \max_{\mathbf{R} \in \mathcal{R}} $ $ \left| \sum_{c \in \mathcal{L}_\mathrm{rx}} \sum_{b \in \mathcal{L}_\mathrm{tx}}  \mathbf{c}_c^\mathrm{H} \mathbf{Z}_\mathrm{rx}^\mathrm{H} \mathbf{R} \mathbf{Z}_\mathrm{tx} \mathbf{b}_b \cdot \pi_{b,c,t,s} \right|^2 + \widebar{\Lambda}_{\mathrm{sinr},t} \cdot \tau_{t,s} \sum_{c \in \mathcal{L}_\mathrm{rx}} $ $ \left\| \mathbf{Z}_\mathrm{rx} \mathbf{c}_c \right\|_2^2 \rho_{c,s},  \forall t \in \mathcal{T}, s \in \mathcal{S} $.

Leveraging the transitivity property of inequalities, we introduce new variables $ \mathrm{H}_{1}: z_{t,s} \geq 0, \forall t \in \mathcal{T}, s \in \mathcal{S} $, which allows us to decouple the \gls{LHS} and \gls{RHS} of $ \mathrm{H}_{\mathrm{aux},2} $. This process yields $ \mathrm{H}_{2}: \frac{1}{\sigma_\mathrm{sen}^2} \min_{\theta_t \in \Theta_t}  \left| \sum_{c \in \mathcal{L}_\mathrm{rx}} \sum_{b \in \mathcal{L}_\mathrm{tx}}  \mathbf{c}_c^\mathrm{H} \mathbf{Z}_\mathrm{rx}^\mathrm{H} \mathbf{A} \left( \theta_t \right) \mathbf{Z}_\mathrm{tx} \mathbf{b}_b \cdot \pi_{b,c,t,s} \right|^2 $ $ \geq z_{t,s}, \forall t \in \mathcal{T}, s \in \mathcal{S} $, and $ \mathrm{H}_{3}: \frac{1}{\sigma_\mathrm{sen}^2} \widebar{\Lambda}_{\mathrm{sinr},t} \cdot \tau_{t,s} \cdot \max_{\mathbf{R} \in \mathcal{R}}  \big| \sum_{c \in \mathcal{L}_\mathrm{rx}} \sum_{b \in \mathcal{L}_\mathrm{tx}}  \mathbf{c}_c^\mathrm{H} \mathbf{Z}_\mathrm{rx}^\mathrm{H} \mathbf{R} \mathbf{Z}_\mathrm{tx} \mathbf{b}_b \cdot \pi_{b,c,t,s} \big|^2 + \widebar{\Lambda}_{\mathrm{sinr},t} \cdot \tau_{t,s}  $ $ \sum_{c \in \mathcal{L}_\mathrm{rx}}  \left\| \mathbf{Z}_\mathrm{rx} \mathbf{c}_c \right\|_2^2 \rho_{c,s} \leq z_{t,s}, \forall t \in \mathcal{T}, s \in \mathcal{S} $. 

Additionally, to enforce that variable $ z_{t,s} $ remains zero when $ \mathsf{T}_t $ is not allocated $ \mathsf{S}_s $, we introduce constraint $ \mathrm{H}_{4}: z_{t,s} \leq \tau_{t,s} \cdot \ddot{M}_{t}^\mathrm{UB}, \forall t \in \mathcal{T}, s \in \mathcal{S} $. Particularly, $ \tau_{t,s} = 0 $ forces $ \pi_{b,c,t,s} = 0 $, which in turn makes the \gls{LHS} of $ \mathrm{H}_{2} $ zero. Here, $ \ddot{M}_{t}^\mathrm{UB} $ is an upper bound for $ z_{t,s} $, obtained by maximizing the \gls{LHS} term in $ \mathrm{H}_{2} $, under the ideal assumption that there is no \gls{AOD} error. It is defined as $ \ddot{M}_{t}^\mathrm{UB} = \max_{b \in \mathcal{L}_\mathrm{tx}} \max_{c \in \mathcal{L}_\mathrm{rx}} \frac{\left| \mathbf{c}_c^\mathrm{H} \mathbf{Z}_\mathrm{rx}^\mathrm{H} \mathbf{A} \left( \theta_t \right) \mathbf{Z}_\mathrm{tx}  \mathbf{b}_b \right|^2}{\sigma_\mathrm{sen}^2} $.

\setcounter{equation}{0}
\setcounter{table}{0}
\renewcommand{\theequation}{L.\arabic{equation}}
\section{Proof to Proposition 6} \label{app:proof-proposition-6}

To uncover hidden structural relationships within $ \mathrm{H}_{2} $, we apply Jensen's inequality, which helps simplify the expression, revealing insights that may not be immediately apparent from the original form. 

Considering arbitrary $ \mathsf{T}_t $ and $ \mathsf{S}_s $, and applying Jensen's inequality we obtain $ \mathrm{I}_{\mathrm{aux},1}: \left| \sum_{c \in \mathcal{L}_\mathrm{rx}} \sum_{b \in \mathcal{L}_\mathrm{tx}}  \mathbf{c}_c^\mathrm{H} \mathbf{Z}_\mathrm{rx}^\mathrm{H} \mathbf{A} \left( \theta_t \right) \mathbf{Z}_\mathrm{tx} \mathbf{b}_b \cdot \pi_{b,c,t,s} \right|^2 \leq \sum_{c \in \mathcal{L}_\mathrm{rx}} $ $  \sum_{b \in \mathcal{L}_\mathrm{tx}} \big| \mathbf{c}_c^\mathrm{H} \mathbf{Z}_\mathrm{rx}^\mathrm{H} \mathbf{A} \left( \theta_t \right) \mathbf{Z}_\mathrm{tx} \mathbf{b}_b \cdot \pi_{b,c,t,s} \big|^2 $. This result is derived by applying a procedure similar to that used in \textbf{\cref{thm:proposition-2}}, where the analysis begins with the absolute value terms, followed by exponentiation and elimination of zero-valued terms.

If $ \mathbf{b}_{i} $ and $ \mathbf{c}_{j} $ denote the transmit and receive codewords for $ \mathsf{T}_t $ in $ \mathsf{S}_s $, then $ \chi_{i,s} = 1 $ and $ \rho_{j,s} = 1 $, whereas $ \chi_{b,s} = 0, \forall b \in \mathcal{L}_\mathrm{tx} \setminus \{i\} $, and $ \rho_{c,s} = 0, \forall c \in \mathcal{L}_\mathrm{rx} \setminus \{j\}  $, since only one transmit and one receive codeword can be chosen per timeslot when sensing is active. Selectively expanding the summations of both sides, $ \mathrm{I}_{\mathrm{aux},1} $ is equivalent to $ \mathrm{I}_{\mathrm{aux},2}: \big| \sum_{(c,b) \neq (j,i)} \mathbf{c}_c^\mathrm{H} \mathbf{Z}_\mathrm{rx}^\mathrm{H} \mathbf{A} \left( \theta_t \right) \mathbf{Z}_\mathrm{tx} \mathbf{b}_b \cdot \pi_{b,c,t,s} + \mathbf{c}_j^\mathrm{H} \mathbf{Z}_\mathrm{rx}^\mathrm{H} \mathbf{A} \left( \theta_t \right) \mathbf{Z}_\mathrm{tx} \mathbf{b}_i \cdot \pi_{i,j,t,s} \big|^2 \leq \sum_{(c,b) \neq (j,i)} \big| \mathbf{c}_c^\mathrm{H} \mathbf{Z}_\mathrm{rx}^\mathrm{H} \mathbf{A} \left( \theta_t \right) \mathbf{Z}_\mathrm{tx} \mathbf{b}_b \cdot \pi_{b,c,t,s} \big|^2 + \big| \mathbf{c}_j^\mathrm{H} \mathbf{Z}_\mathrm{rx}^\mathrm{H} \mathbf{A} \left( \theta_t \right) \mathbf{Z}_\mathrm{tx} \mathbf{b}_i \cdot \pi_{i,j,t,s} \big|^2 $.

The terms affected under the summations are all zero, making $ \mathrm{I}_{\mathrm{aux},2} $ hold at the equality and showing that Jensen's inequality is tight when applied to $ \mathrm{H}_{2} $. Based on this outcome, we equivalently redefine $ \mathrm{H}_{2} $ as $ \mathrm{I}_{\mathrm{aux},3}: \frac{1}{\sigma_\mathrm{sen}^2} \min_{\theta_t \in \Theta_t} \sum_{c \in \mathcal{L}_\mathrm{rx}} \sum_{b \in \mathcal{L}_\mathrm{tx}} \big| \mathbf{c}_c^\mathrm{H} \mathbf{Z}_\mathrm{rx}^\mathrm{H} \mathbf{A} \left( \theta_t \right) \mathbf{Z}_\mathrm{tx} \mathbf{b}_b \big|^2  \pi_{b,c,t,s} \geq z_{t,s}, \forall t \in \mathcal{T}, s \in \mathcal{S} $, where $ \pi_{b,c,t,s} $ can be factored out of the absolute value since it is binary.

Considering $ \mathbf{d}_{b,c} = \left( \mathbf{Z}_\mathrm{tx}^* \mathbf{b}_b^* \right) \otimes \left(  \mathbf{Z}_\mathrm{rx} \mathbf{c}_c \right) $ and $ \mathbf{a} \left( \theta_t \right) = \mathrm{vec} \left( \mathbf{A} \left( \theta_t \right) \right) $, the equivalence $ \mathbf{c}_c^\mathrm{H} \mathbf{Z}_\mathrm{rx}^\mathrm{H} \mathbf{A} \left( \theta_t \right) \mathbf{Z}_\mathrm{tx} \mathbf{b}_b = \mathbf{d}_{b,c}^\mathrm{H} \mathbf{a} \left( \theta_t \right) $ holds. Hence, $ \mathrm{I}_{\mathrm{aux},3} $ can be expressed as $ \mathrm{I}_{\mathrm{aux},4}: \frac{1}{\sigma_\mathrm{sen}^2}  \min_{\theta_t \in \Theta_t} \sum_{c \in \mathcal{L}_\mathrm{rx}} \sum_{b \in \mathcal{L}_\mathrm{tx}} \big| \mathbf{d}_{b,c}^\mathrm{H}  \mathbf{a} \left( \theta_t \right)  \big|^2  \pi_{b,c,t,s} \geq z_{t,s}, \forall t \in \mathcal{T}, s \in \mathcal{S} $. Employing (\ref{eqn:linearized-response-matrix}) and $ \mathrm{C}_{18} $, constraint $ \mathrm{I}_{\mathrm{aux},4} $ can be equivalently expressed as $ \mathrm{I}_{\mathrm{aux},5}: \frac{1}{\sigma_\mathrm{sen}^2} \min_{ \left| \Delta \theta_t \right|^2 \leq \epsilon_\mathrm{AOD}^2 } \sum_{c \in \mathcal{L}_\mathrm{rx}} \sum_{b \in \mathcal{L}_\mathrm{tx}} \big| \mathbf{d}_{b,c}^\mathrm{H} \big( \mathbf{a} \left( \widebar{\theta}_t \right) + \widetilde{\mathbf{a}} \left( \widebar{\theta}_t \right) \Delta \theta_t \big) \big|^2 $ $ \pi_{b,c,t,s}  \geq z_{t,s}, \forall t \in \mathcal{T}, s \in \mathcal{S} $, where $ \mathbf{a} \left( \widebar{\theta}_t \right) = \mathrm{vec} \left( \mathbf{A} \left( \widebar{\theta}_t \right) \right) $ and $ \widetilde{\mathbf{a}} \left( \widebar{\theta}_t \right) = \mathrm{vec} \left( \widetilde{\mathbf{A}} \left( \widebar{\theta}_t \right) \right) $. We can separate $ \mathrm{I}_{\mathrm{aux},5} $ into two constraints, specifically, $ \mathrm{I}_{\mathrm{aux},6}: \left| \Delta \theta_t \right|^2 - \epsilon_\mathrm{AOD}^2 \leq 0, \forall t \in \mathcal{T}, $ and $ \mathrm{I}_{\mathrm{aux},7}: z_{t,s} -  \sum_{c \in \mathcal{L}_\mathrm{rx}} \sum_{b \in \mathcal{L}_\mathrm{tx}} \frac{1}{\sigma_\mathrm{sen}^2} \big| \mathbf{d}_{b,c}^\mathrm{H} \big( \mathbf{a} \left( \widebar{\theta}_t \right) + \widetilde{\mathbf{a}} \left( \widebar{\theta}_t \right) \Delta \theta_t \big) \big|^2 $ $ \pi_{b,c,t,s} \leq 0, \forall t \in \mathcal{T}, s \in \mathcal{S} $.

Next, we apply \textbf{Lemma \ref{lem:s-procedure}} to $ \mathrm{I}_{\mathrm{aux},6} $ and $ \mathrm{I}_{\mathrm{aux},7} $, leading to constraints $ \mathrm{I}_{1} $ and $ \mathrm{I}_{2} $, shown below,
\begin{align*}
			   	& 	 			   	\mathrm{I}_{1}: \left[ 
			   			 			   					\begin{matrix}
			   			 			   					   	i_{t,s}^\mathrm{a}
			   			 			   					   	& 
			   			 			   					   	i_{t,s}^\mathrm{b}
			   			 			   					   	\\
			   			 			   					    i_{t,s}^\mathrm{c}
			   			 			   					    & 
			   			 			   					    i_{t,s}^\mathrm{d}
			   			 			   				  	\end{matrix} 
			   		 			   					\right] 
			   		 					   		\succcurlyeq \mathbf{0}, \forall t \in \mathcal{T}, s \in \mathcal{S},
			   	\\	 					   		
				& \mathrm{I}_{2}: \xi_{t} \geq 0, \forall t \in \mathcal{T}, 
\end{align*}
where $ i_{t,s}^\mathrm{a} = \xi_{t} - \widetilde{\mathbf{a}} \left( \widebar{\theta}_t \right)^\mathrm{H} \mathbf{F}_{t,s} \widetilde{\mathbf{a}} \left( \widebar{\theta}_t \right) $, 
$ i_{t,s}^\mathrm{b} = - \widetilde{\mathbf{a}} \left( \widebar{\theta}_t \right)^\mathrm{H} \mathbf{F}_{t,s} \mathbf{a} \left( \widebar{\theta}_t \right) $, 
$ i_{t,s}^\mathrm{c} = - \mathbf{a} \left( \widebar{\theta}_t \right)^\mathrm{H} \mathbf{F}_{t,s} \widetilde{\mathbf{a}} \left( \widebar{\theta}_t \right) $,  
$ i_{t,s}^\mathrm{d} = - \xi_{t} \cdot \epsilon_\mathrm{AOD}^2 - \mathbf{a} \left( \widebar{\theta}_t \right)^\mathrm{H} \mathbf{F}_{t,s} \mathbf{a} \left( \widebar{\theta}_t \right) - z_{t,s} $, 
$ \mathbf{F}_{t,s} = - \sum_{c \in \mathcal{L}_\mathrm{rx}} \sum_{b \in \mathcal{L}_\mathrm{tx}} \frac{1}{\sigma_\mathrm{sen}^2} \mathbf{d}_{b,c} \mathbf{d}_{b,c}^\mathrm{H} \pi_{b,c,t,s} $, and $ \xi_t $ are newly introduced variables needed by the \emph{S-Procedure}.

\setcounter{equation}{0}
\setcounter{table}{0}
\renewcommand{\theequation}{M.\arabic{equation}}
\section{Proof to Proposition 7} \label{app:proof-proposition-7}

Due to the complexity of $ \mathrm{H}_{3} $, we begin by analyzing this constraint for arbitrary $ \mathsf{T}_t $ and $ \mathsf{S}_s $. Assuming $ \mathbf{d}_{b,c} = \left( \mathbf{Z}_\mathrm{tx}^* \mathbf{b}_b^* \right) \otimes \left(  \mathbf{Z}_\mathrm{rx} \mathbf{c}_c \right) $ and $ \mathbf{r} = \mathrm{vec} \left( \mathbf{R}  \right) $, the equivalence $ \mathbf{c}_c^\mathrm{H} \mathbf{Z}_\mathrm{rx}^\mathrm{H} \mathbf{R} \mathbf{Z}_\mathrm{tx} \mathbf{b}_b = \mathbf{d}_{b,c}^\mathrm{H} \mathbf{r} $ holds. Hence, we have $ \mathrm{J}_{\mathrm{aux},1}: \frac{1}{\sigma_\mathrm{sen}^2} \widebar{\Lambda}_{\mathrm{sinr},t} \tau_{t,s} \cdot \max_{\mathbf{r} \in \mathcal{R'}} \big| \sum_{c \in \mathcal{L}_\mathrm{rx}} \sum_{b \in \mathcal{L}_\mathrm{tx}} \mathbf{d}_{b,c}^\mathrm{H} \mathbf{r} \cdot \pi_{b,c,t,s} \big|^2 + \widebar{\Lambda}_{\mathrm{sinr},t} \tau_{t,s}  \sum_{c \in \mathcal{L}_\mathrm{rx}} \left\| \mathbf{Z}_\mathrm{rx} \mathbf{c}_c \right\|_2^2 \rho_{c,s} \leq z_{t,s} $, where $ \mathbf{r} = \upsilon \widebar{\mathbf{r}} + \upsilon \boldsymbol{\Delta} \mathbf{r} $, $ \widebar{\mathbf{r}} = \mathrm{vec} \left( \widebar{\mathbf{R}} \right) $, $ \boldsymbol{\Delta} \mathbf{r} = \mathrm{vec} \left( \boldsymbol{\Delta} \mathbf{R} \right) $, and $ \mathcal{R}' \triangleq \left\lbrace \mathbf{r} \mid \mathbf{r} = \upsilon \widebar{\mathbf{r}} + \upsilon \boldsymbol{\Delta} \mathbf{r}, \left\| \boldsymbol{\Delta} \mathbf{r} \right\|_2^2 \leq \epsilon_\mathrm{RSI}^2 \right\rbrace $. 

Note that $ \pi_{b,c,t,s} $ includes $ \tau_{t,s} $, as detailed in \textbf{\cref{thm:proposition-3}}. Hence, $ \tau_{t,s} $ can be absorbed, transforming $ \mathrm{J}_{\mathrm{aux},1} $ into $ \mathrm{J}_{\mathrm{aux},2}: \frac{1}{\sigma_\mathrm{sen}^2} \widebar{\Lambda}_{\mathrm{sinr},t}  \max_{\mathbf{r} \in \mathcal{R'}} \big| \sum_{c \in \mathcal{L}_\mathrm{rx}} \sum_{b \in \mathcal{L}_\mathrm{tx}} \mathbf{d}_{b,c}^\mathrm{H} \mathbf{r} \cdot \pi_{b,c,t,s} \big|^2 + \widebar{\Lambda}_{\mathrm{sinr},t} \tau_{t,s}  \sum_{c \in \mathcal{L}_\mathrm{rx}} \left\| \mathbf{Z}_\mathrm{rx} \mathbf{c}_c \right\|_2^2 \rho_{c,s} \leq z_{t,s} $

To uncover hidden relationships, we apply Jensen's inequality to the quadratic term on the \gls{LHS} of $ \mathrm{J}_{\mathrm{aux},2} $, which involves additive and multiplicative couplings among multiple variables. This leads us to $ \mathrm{J}_{\mathrm{aux},3}: \big| \sum_{c \in \mathcal{L}_\mathrm{rx}} \sum_{b \in \mathcal{L}_\mathrm{tx}} \mathbf{d}_{b,c}^\mathrm{H} \mathbf{r} \cdot \pi_{b,c,t,s} \big|^2 \leq \sum_{c \in \mathcal{L}_\mathrm{rx}} \sum_{b \in \mathcal{L}_\mathrm{tx}} \big| \mathbf{d}_{b,c}^\mathrm{H} \mathbf{r} \cdot \pi_{b,c,t,s} \big|^2 $. This result is obtained by first applying Jensen's inequality to the absolute value terms, followed by exponentiation and the elimination of zero-valued terms, in a manner similar to that used in \textbf{\cref{thm:proposition-2}}.

  Assuming that $ \mathbf{b}_i $ and $ \mathbf{c}_j $ are the selected transmit and receive codewords, we expand the summations,  yielding $ \mathrm{J}_{\mathrm{aux},4}: \big| \sum_{(c,b) \neq (j,i)} \mathbf{d}_{b,c}^\mathrm{H} \mathbf{r} \cdot \pi_{b,c,t,s} + \mathbf{d}_{i,j}^\mathrm{H} \mathbf{r} \cdot \pi_{i,j,t,s} \big|^2 \leq \sum_{(c,b) \neq (j,i)} \big| \mathbf{d}_{b,c}^\mathrm{H} \mathbf{r} \cdot \pi_{b,c,t,s} \big|^2 + \left| \mathbf{d}_{i,j}^\mathrm{H} \mathbf{r} \cdot \pi_{i,j,t,s} \right|^2 $.

Since each target in an active timeslot can only be assigned one transmit/receive codeword pair, the terms under the summation vanish, ensuring that $ \mathrm{J}_{\mathrm{aux},4} $ holds with equality. Additionally, $ \pi_{b,c,t,s} $ can be factored out of the absolute value because it is binary. Consequently, we can reformulate $ \mathrm{J}_{\mathrm{aux},1} $ as $ \mathrm{J}_{\mathrm{aux},5}: \frac{1}{\sigma_\mathrm{sen}^2} \widebar{\Lambda}_{\mathrm{sinr},t} \cdot \max_{\mathbf{r} \in \mathcal{R'}} \sum_{c \in \mathcal{L}_\mathrm{rx}} \sum_{b \in \mathcal{L}_\mathrm{tx}} \big|  \mathbf{d}_{b,c}^\mathrm{H} \mathbf{r} \big|^2 \pi_{b,c,t,s} + \widebar{\Lambda}_{\mathrm{sinr},t} \tau_{t,s}   \sum_{c \in \mathcal{L}_\mathrm{rx}} \left\| \mathbf{Z}_\mathrm{rx} \mathbf{c}_c \right\|_2^2 \rho_{c,s} \leq z_{t,s} $, which is more amenable for further manipulation. 

Next, we tackle the coupling between $ \rho_{c,s} $ and $ \tau_{t,s} $. Leveraging results from \textbf{\cref{thm:proposition-3}}, we reformulate this coupling as an intersection of multiple constraints. Specifically, the product $ \rho_{c,s} \tau_{t,s} $ can be re-expressed through the intersecting constraints, 
$ \mathrm{J}_{1}: \delta_{c,t,s} \leq \rho_{c,s}, \forall c \in \mathcal{L}_\mathrm{rx}, t \in \mathcal{T}, s \in \mathcal{S} $,
$ \mathrm{J}_{2}: \delta_{c,t,s} \leq \tau_{t,s}, \forall c \in \mathcal{L}_\mathrm{rx}, t \in \mathcal{T}, s \in \mathcal{S}  $, and
$ \mathrm{J}_{3}: \delta_{c,t,s} \geq \rho_{c,s} + \tau_{t,s} - 1, \forall c \in \mathcal{L}_\mathrm{rx}, t \in \mathcal{T}, s \in \mathcal{S} $,
enabling a more tractable optimization approach\footnote{Constraints $ \mathrm{J}_{1} $, $ \mathrm{J}_{2} $, and $ \mathrm{J}_{3} $ can be directly derived from \textbf{\cref{thm:proposition-3}} by setting one of the binary variables to $ 1 $. }. Here, $ \delta_{c,t,s} $ are newly introduced variables that, by definition, must be binary. However, we can relax $ \delta_{c,t,s} $ to a continuous variable to reduce the number of binary variables, thereby lowering computational complexity of the binary search process. To enforce its binary nature in all cases, we impose constraint $ \mathrm{J}_{4}: \delta_{c,t,s} \in \left[ 0, 1 \right], \forall c \in \mathcal{L}_\mathrm{rx}, t \in \mathcal{T}, s \in \mathcal{S} $.

Leveraging the results in $ \mathrm{J}_{\mathrm{aux},5} $, we generalize for all targets and timeslots, allowing us to recast $ \mathrm{H}_{3} $ as $ \mathrm{J}_{\mathrm{aux},6}: \frac{1}{\sigma_\mathrm{sen}^2} \widebar{\Lambda}_{\mathrm{sinr},t} \cdot \max_{\mathbf{r} \in \mathcal{R'}} \sum_{c \in \mathcal{L}_\mathrm{rx}} \sum_{b \in \mathcal{L}_\mathrm{tx}} \big|  \mathbf{d}_{b,c}^\mathrm{H} \mathbf{r} \big|^2 \pi_{b,c,t,s} + \widebar{\Lambda}_{\mathrm{sinr},t} \sum_{c \in \mathcal{L}_\mathrm{rx}} \left\| \mathbf{Z}_\mathrm{rx} \mathbf{c}_c \right\|_2^2 \delta_{c,t,s} \leq z_{t,s}, \forall t \in \mathcal{T}, s \in \mathcal{S} $.

Note that we can express $ \mathrm{J}_{\mathrm{aux},6} $ as 
$ \mathrm{J}_{\mathrm{aux},7}: \frac{1}{\sigma_\mathrm{sen}^2} \widebar{\Lambda}_{\mathrm{sinr},t} \cdot \max_{\left\| \boldsymbol{\Delta} \mathbf{r} \right\|_2^2 \leq  \epsilon_\mathrm{RSI}^2}
\sum_{c \in \mathcal{L}_\mathrm{rx}} \sum_{b \in \mathcal{L}_\mathrm{tx}}
\big| \mathbf{d}_{b,c}^\mathrm{H} \left( \upsilon \widebar{\mathbf{r}} + \upsilon \boldsymbol{\Delta} \mathbf{r} \right) \big|^2 \pi_{b,c,t,s} + \widebar{\Lambda}_{\mathrm{sinr},t} $ $\sum_{c \in \mathcal{L}_\mathrm{rx}} \left\| \mathbf{Z}_\mathrm{rx} \mathbf{c}_c \right\|_2^2 \delta_{c,t,s} \leq z_{t,s}, \forall t \in \mathcal{T}, s \in \mathcal{S} $, after letting $ \mathrm{J}_{\mathrm{aux},6} $ absorb $ \mathcal{R'} $. Here, $ \mathrm{J}_{\mathrm{aux},7} $ can be split into $ \mathrm{J}_{\mathrm{aux},8}: \boldsymbol{\Delta} \mathbf{r}^\mathrm{H} \boldsymbol{\Delta} \mathbf{r} - \epsilon_\mathrm{RSI}^2 \leq 0 $ and $ \mathrm{J}_{\mathrm{aux},9}: \frac{1}{\sigma_\mathrm{sen}^2} \widebar{\Lambda}_{\mathrm{sinr},t} \cdot 
\sum_{c \in \mathcal{L}_\mathrm{rx}} \sum_{b \in \mathcal{L}_\mathrm{tx}}
\big| \mathbf{d}_{b,c}^\mathrm{H} \left( \upsilon \widebar{\mathbf{r}} + \upsilon \boldsymbol{\Delta} \mathbf{r} \right) \big|^2 \pi_{b,c,t,s} + \widebar{\Lambda}_{\mathrm{sinr},t} \sum_{c \in \mathcal{L}_\mathrm{rx}} \left\| \mathbf{Z}_\mathrm{rx} \mathbf{c}_c \right\|_2^2 \delta_{c,t,s} - z_{t,s} \leq 0, \forall t \in \mathcal{T}, s \in \mathcal{S} $.

Applying \textbf{Lemma \ref{lem:s-procedure}} to $ \mathrm{J}_{\mathrm{aux},8} $ and $ \mathrm{J}_{\mathrm{aux},9} $, leads to constraints $ \mathrm{J}_{5} $ and $ \mathrm{J}_{6} $, shown below,
\begin{align*}
			   	& 	 			   	\mathrm{J}_{5}: \left[ 
			   			 			   					\begin{matrix}
			   			 			   					   	\mathbf{J}_{t,s}^\mathrm{a}
			   			 			   					   	& 
			   			 			   					   	\mathbf{j}_{t,s}^\mathrm{b}
			   			 			   					   	\\
			   			 			   					    \mathbf{j}_{t,s}^\mathrm{c}
			   			 			   					    & 
			   			 			   					    j_{t,s}^\mathrm{d}
			   			 			   				  	\end{matrix} 
			   		 			   					\right] 
			   		 					   		\succcurlyeq \mathbf{0}, \forall t \in \mathcal{T}, s \in \mathcal{S},
			  	\\ 		 					   		
				& \mathrm{J}_{6}: \iota \geq 0, 
\end{align*}
where
$ \mathbf{J}_{t,s}^\mathrm{a} = \iota \mathbf{I} + \upsilon^2 \widebar{\Lambda}_{\mathrm{sinr},t} \mathbf{F}_{t,s} $, 
$ \mathbf{j}_{t,s}^\mathrm{b} = \upsilon^2 \widebar{\Lambda}_{\mathrm{sinr},t} \mathbf{F}_{t,s} \widebar{\mathbf{r}} $,
$ \mathbf{j}_{t,s}^\mathrm{c} = \upsilon^2 \widebar{\Lambda}_{\mathrm{sinr},t} \widebar{\mathbf{r}}^\mathrm{H} \mathbf{F}_{t,s}  $,
$ j_{t,s}^\mathrm{d} = - \iota \cdot \epsilon_\mathrm{RSI}^2 + \upsilon^2 \widebar{\Lambda}_{\mathrm{sinr},t} \widebar{\mathbf{r}}^\mathrm{H} \mathbf{F}_{t,s} \widebar{\mathbf{r}} - \widebar{\Lambda}_{\mathrm{sinr},t} \sum_{c \in \mathcal{L}_\mathrm{rx}} \left\| \mathbf{Z}_\mathrm{rx} \mathbf{c}_c \right\|_2^2 \delta_{c,t,s} + z_{t,s} $, 
$ \mathbf{F}_{t,s} = - \sum_{c \in \mathcal{L}_\mathrm{rx}} \sum_{b \in \mathcal{L}_\mathrm{tx}} \frac{1}{\sigma_\mathrm{sen}^2} \mathbf{d}_{b,c} \mathbf{d}_{b,c}^\mathrm{H} \pi_{b,c,t,s} $, and $ \iota $ is a newly introduced variable needed by the \emph{S-Procedure}.

\setcounter{equation}{0}
\setcounter{table}{0}
\renewcommand{\theequation}{N.\arabic{equation}}
\section{Proof to Proposition 8} \label{app:proof-proposition-8}

Since $ f_1 \left( \boldsymbol{\Omega} \right) $ and $ f_2 \left( \boldsymbol{\Omega} \right) $ depend on binary variables constrained to $ \left\lbrace 0, 1 \right\rbrace  $, we can analyze their extremal values and derive optimal weights $ \eta_1 $ and $ \eta_2 $ that enforce a strict prioritization, specifically, favoring the minimization of time consumption over energy expenditure.

Let $ P_{\mathrm{tx}}^{\mathrm{max}} $ and $ P_{\mathrm{rx}}^{\mathrm{max}} $ denote the maximum transmit and receive powers, respectively, and $ P_{\mathrm{tx}}^{\mathrm{min}} $ and $ P_{\mathrm{rx}}^{\mathrm{min}} $ their minimum counterparts. Thus, for a single active timeslot, $ f_1 \left( \boldsymbol{\Omega} \right) $ lies in the interval $ \left[ S_\mathrm{dur} P_{\mathrm{tx}}^{\mathrm{min}} + S_\mathrm{dur} P_{\mathrm{rx}}^{\mathrm{min}}, \; S_\mathrm{dur} P_{\mathrm{tx}}^{\mathrm{max}} + S_\mathrm{dur} P_{\mathrm{rx}}^{\text{max}} \right] $. Extending this to all timeslots, $ f_1 \left( \boldsymbol{\Omega} \right) $ is in the interval $ \left[ S' S_\mathrm{dur} P_{\mathrm{tx}}^{\mathrm{min}} + S' S_\mathrm{dur} P_{\mathrm{rx}}^{\mathrm{min}}, \; S' S_\mathrm{dur} P_{\mathrm{tx}}^{\mathrm{max}} + S' S_\mathrm{dur}  P_{\mathrm{rx}}^{\text{max}} \right] $, where $ S' \leq S $ is the number of active timeslots.

Meanwhile, function $ f_2 \left( \boldsymbol{\Omega} \right) $ takes values in the interval $ \left[ r_1, r_2 \right] $, where $ r_1 = S_\mathrm{dur} \Delta_0 $ and $ r_2 = S_\mathrm{dur} S' \Delta_0 + S_\mathrm{dur} \frac{\Delta_\omega}{2} (S'-2)(S'-1) $. Here, $ r_2 $  corresponds to the cumulative sum of all weights $ \omega_s $, as defined in \textbf{Lemma~\ref{lem:lemma-weights}}, over $ S' $ timeslots.

Notably, the images of $ f_1 \left( \boldsymbol{\Omega} \right) $ and $ f_2 \left( \boldsymbol{\Omega} \right) $ may overlap depending on the values of $ S_\mathrm{dur} $, $ P_{\mathrm{tx}}^{\mathrm{min}} $, $ P_{\mathrm{rx}}^{\mathrm{min}} $, $ P_{\mathrm{tx}}^{\mathrm{max}} $, $ P_{\mathrm{rx}}^{\mathrm{max}} $, $ \Delta_0 $, and $ \Delta_\omega $. To ensure completely non-overlapping intervals for the images of these functions and a strict prioritization of $ f_2 \left( \boldsymbol{\Omega} \right) $, we impose the condition $ \eta_2 f_2 \left( \boldsymbol{\Omega} \right) > \eta_1 f_1 \left( \boldsymbol{\Omega} \right) $, where $ \eta_1 $ and $ \eta_2 $ are weighting coefficients.

For simplicity, we consider the extreme case where the minimum value of $ \eta_2 f_2 $ (i.e., $ \eta_2 S_\mathrm{dur} \Delta_0 $) must exceed the maximum value of $ \eta_1 f_1 $ (i.e., $ \eta_1 S S_\mathrm{dur} P_{\mathrm{tx}}^{\mathrm{max}} + \eta_1 S S_\mathrm{dur}  P_{\mathrm{rx}}^{\mathrm{max}} $), occurring when $ S' = S $. This leads to the sufficient condition  
\begin{align*}
	\eta_2  > \eta_1 \frac{S \left( P_{\mathrm{tx}}^{\mathrm{max}} + P_{\mathrm{rx}}^{\mathrm{max}} \right)}{\Delta_0}.
\end{align*}

We now address the complexity of $ f_1 \left( \boldsymbol{\Omega} \right) $, which arises from its quadratic terms. In contrast, $ f_2 \left( \boldsymbol{\Omega} \right) $ is linear and straightforward to handle.

Given the structural similarity between the terms $ \sum_{s \in \mathcal{S}} \left\| \mathbf{w}_s \right\|_2^2 $ and $ \sum_{s \in \mathcal{S}} \left\| \mathbf{v}_s \right\|_2^2 $, we start by analyzing the first term and later extend the approach to the second. Let us focus on an arbitrary $ \mathsf{S}_s $. Using the definition in $ \mathrm{C}_{11} $, we have $ \left\| \mathbf{w}_s \right\|_2^2 = \mathbf{w}_s^\mathrm{H} \mathbf{w}_s = \left( \sum_{b \in \mathcal{L}_\mathrm{tx}} \mathbf{b}_b \cdot \chi_{b,s} \right)^\mathrm{H} \left( \sum_{b' \in \mathcal{L}_\mathrm{tx}} \mathbf{b}_{b'} \cdot \chi_{b',s} \right) = \sum_{b \in \mathcal{L}_\mathrm{tx}} \sum_{b' \in \mathcal{L}_\mathrm{tx}} \mathbf{b}_b^\mathrm{H} \mathbf{b}_{b'} \cdot \chi_{b,s} \chi_{b',s} $

Employing $ \mathrm{C}_{10} $, we distinguish two cases, which we leverage to simplify the expression above, as shown in the following.

\noindent \textbf{Case \circled{\footnotesize{1}} $\Rightarrow$ } If $ \mathsf{S}_s $ is idle, i.e., $ \gamma_s = 0 $, then $ \sum_{b \in \mathcal{L}_\mathrm{tx}} \chi_{b,s} = 0 $, which implies that $ \chi_{b,s} = 0, \forall b \in \mathcal{L}_\mathrm{tx} $, ensuring that no codeword is selected in $ \mathsf{S}_s $. Hence, $ \left\| \mathbf{w}_s \right\|_2^2 = 0 $.
	
\noindent \textbf{Case \circled{\footnotesize{2}} $\Rightarrow$ } If $ \mathsf{S}_s $ is active, i.e., $ \gamma_s = 1 $, then $ \sum_{b \in \mathcal{L}_\mathrm{tx}} \chi_{b,s} = 1 $, which implies that $ \chi_{i,s} = 1 $ and  $ \chi_{b,s} = 0, \forall b \in \mathcal{L}_\mathrm{tx} \setminus \left\lbrace i \right\rbrace $, assuming that $ \mathbf{b}_i $ denotes the selected codeword. Hence, $ \left\| \mathbf{w}_s \right\|_2^2 =  \sum_{(b,b') \neq (i,i)} \mathbf{b}_b^\mathrm{H} \mathbf{b}_{b'} \cdot \chi_{b,s} \chi_{b',s} +  \mathbf{b}_i^\mathrm{H} \mathbf{b}_{i} \cdot \chi_{i,s} \chi_{i,s} $. The term under the summation yields zero, thereby resulting in $ \left\| \mathbf{w}_s \right\|_2^2 = \mathbf{b}_i^\mathrm{H} \mathbf{b}_{i} $.

We can combine the two cases above, resulting in $ \left\| \mathbf{w}_s \right\|_2^2 =  \sum_{b \in \mathcal{L}_\mathrm{tx}} \left\| \mathbf{b}_b \right\|_2^2 \cdot \chi_{b,s} $. Employing $ \mathrm{C}_{13} $ and $ \mathrm{C}_{14} $, and following a similar procedure as above, we obtain $ \left\| \mathbf{v}_s \right\|_2^2 = \sum_{c \in \mathcal{L}_\mathrm{rx}} \left\| \mathbf{c}_c \right\|_2^2 \cdot \rho_{c,s} $.	As a result, we redefine  $ f_1 \left( \boldsymbol{\Omega} \right) $ as $ f_1' \left( \boldsymbol{\Omega}' \right) = S_\mathrm{dur} \sum_{s \in \mathcal{S}} \sum_{b \in \mathcal{L}_\mathrm{tx}} \left\| \mathbf{b}_b \right\|_2^2 \cdot \chi_{b,s} + S_\mathrm{dur} \sum_{s \in \mathcal{S}} \sum_{c \in \mathcal{L}_\mathrm{rx}} \left\| \mathbf{c}_c \right\|_2^2 \cdot \rho_{c,s} $, which leads to $ f' \left( \boldsymbol{\Omega}' \right) \triangleq \eta_1  f'_1 \left( \boldsymbol{\Omega}' \right) + \eta_2 f_2' \left( \boldsymbol{\Omega}' \right) $, assuming $ f_2' \left( \boldsymbol{\Omega}' \right) = f_2 \left( \boldsymbol{\Omega} \right) $.

\setcounter{equation}{0}
\setcounter{table}{0}
\renewcommand{\theequation}{O.\arabic{equation}}
\section{Derivation of cutting planes} \label{app:derivation-cutting-planes}

To eliminate infeasible solutions without exhaustively searching for the optimal transmit codeword, we derive an upper bound for the numerator of (\ref{eqn:communication-snr}), which helps prune reduce infeasible solutions.

Specifically, when user $ \mathsf{U}_u $ is served in $ \mathsf{S}_s $, the power of the signal received by the user is given by $ \mathrm{K}_{\mathrm{aux},1}: \left| \mathbf{h}_u^\mathrm{H} \mathbf{Z}_\mathrm{tx} \mathbf{w}_s \right|^2 = \left| \left( \widebar{\mathbf{h}}_u^\mathrm{H} + \boldsymbol{\Delta} \mathbf{h}_u^\mathrm{H} \right) \mathbf{Z}_\mathrm{tx} \mathbf{w}_s \right|^2 $, where $ \boldsymbol{\Delta} \mathbf{h}_u $ denotes the channel perturbation.

Applying the triangle inequality and the Cauchy-Schwarz inequality twice, the highest received signal power with respect to the CSI uncertainty can be upper bounded as $ \left| \left( \widebar{\mathbf{h}}_u^\mathrm{H} + \boldsymbol{\Delta} \mathbf{h}_u^\mathrm{H} \right) \mathbf{Z}_\mathrm{tx} \mathbf{w}_s \right| \leq \left| \widebar{\mathbf{h}}_u^\mathrm{H} \mathbf{Z}_\mathrm{tx} \mathbf{w}_s \right| + \epsilon_{\mathrm{CSI}} \left\| \mathbf{Z}_\mathrm{tx} \mathbf{w}_s \right\|_2 \leq \left| \widebar{\mathbf{h}}_u^\mathrm{H} \mathbf{Z}_\mathrm{tx} \mathbf{w}_s \right| + \epsilon_{\mathrm{CSI}} \left\| \mathbf{Z}_\mathrm{tx} \right\|_\mathrm{F} \left\| \mathbf{w}_s \right\|_2 $. Hence, an upper bound for $ \mathrm{K}_{\mathrm{aux},1} $ is expressed by the relation: $ \mathrm{K}_{\mathrm{aux},2}: \left| \mathbf{h}_u^\mathrm{H} \mathbf{Z}_\mathrm{tx} \mathbf{w}_s \right|^2 \leq \left( \left| \widebar{\mathbf{h}}_u^\mathrm{H} \mathbf{Z}_\mathrm{tx} \mathbf{w}_s \right| + \epsilon_{\mathrm{CSI}} \left\| \mathbf{Z}_\mathrm{tx} \right\|_\mathrm{F} \left\| \mathbf{w}_s \right\|_2 \right)^2 $

Next, we maximize the \gls{RHS} term of $ \mathrm{K}_{\mathrm{aux},2} $ with respect to the beamformer $ \mathbf{w}_s $. The term is maximized when $ \mathbf{w}_s $ is chosen to be colinear with $ \mathbf{Z}_\mathrm{tx}^\mathrm{H} \widebar{\mathbf{h}}_u $, i.e., $ \mathbf{w}_s = \frac{\mathbf{Z}_\mathrm{tx}^\mathrm{H} \widebar{\mathbf{h}}_u}{\left\| \mathbf{Z}_\mathrm{tx}^\mathrm{H} \widebar{\mathbf{h}}_u \right\|_2 } \sqrt{P_\mathrm{tx}^\mathrm{max}} $, where $ P_\mathrm{tx}^\mathrm{max} $ denotes the maximum transmit power at the \gls{BS}, i.e, $ P_\mathrm{tx}^\mathrm{max} = \max_{b \in \mathcal{L}_\mathrm{tx}} \left\| \mathbf{b}_b \right\|^2_2 $.

Under this choice, the maximum received signal power admits the upper bound $ \mathrm{K}_{\mathrm{aux},3}: \big( \left\| \mathbf{Z}_\mathrm{tx}^\mathrm{H} \widebar{\mathbf{h}}_u \right\|_2 + \epsilon_{\mathrm{CSI}} \left\| \mathbf{Z}_\mathrm{tx} \right\|_\mathrm{F}  \big)^2 P_\mathrm{tx}^\mathrm{max} $. Consequently, we obtain the following cutting plane, $ \mathrm{K}_{1}: \left| \mathbf{h}_u^\mathrm{H} \mathbf{Z}_\mathrm{tx} \mathbf{w}_s \right|^2 \leq \big( \left\| \mathbf{Z}_\mathrm{tx}^\mathrm{H} \widebar{\mathbf{h}}_u \right\|_2 + \epsilon_{\mathrm{CSI}} \left\| \mathbf{Z}_\mathrm{tx} \right\|_\mathrm{F} \big)^2 P_\mathrm{tx}^\mathrm{max} $.

To further eliminate equally optimal solutions that differ only in the temporal order of user and target assignments, we introduce an additional cutting plane. Specifically, we define constraint $ \mathrm{K}_{\mathrm{aux},4}: \left\| \mathbf{w}_{s+1} \right\|_2^2 + \left\| \mathbf{v}_{s+1} \right\|_2^2 \geq \left\| \mathbf{w}_s \right\|_2^2 + \left\| \mathbf{v}_s \right\|_2^2, \forall s \in \mathcal{S} \setminus \{ S \} $, which enforces that the total transmit power does not increase over successive time slots. This constraint introduces a temporal ordering to the allocation, reducing redundant permutations that yield the same objective value.

To address the quadratic terms in $ \mathrm{K}_{\mathrm{aux},4} $, we leverage the relationships derived in \textbf{Appendix \ref{app:proof-proposition-8}}, which allow for a simplification of these expressions. As a result $ \mathrm{K}_{\mathrm{aux},4} $ can be reformulated into a more tractable form, denoted as $ \mathrm{K}_{2}: \sum_{b \in \mathcal{L}_\mathrm{tx}} \left\| \mathbf{b}_b \right\|_2^2 \cdot \chi_{b,s+1} + \sum_{c \in \mathcal{L}_\mathrm{rx}} \left\| \mathbf{c}_c \right\|_2^2 \cdot \rho_{c,s+1} \geq \sum_{b \in \mathcal{L}_\mathrm{tx}} \left\| \mathbf{b}_b \right\|_2^2 \cdot \chi_{b,s} + \sum_{c \in \mathcal{L}_\mathrm{rx}} \left\| \mathbf{c}_c \right\|_2^2 \cdot \rho_{c,s}, \forall s \in \mathcal{S} \setminus \{ S \} $, yielding a simplified version of the cutting plane.

\setcounter{equation}{0}
\setcounter{table}{0}
\renewcommand{\theequation}{P.\arabic{equation}}
\section{Transformed problem} \label{app:transformed-problem}

For completeness, the problem $ \mathcal{P}' \left( \boldsymbol{\Omega}' \right) $, with all constraints explicitly detailed, is presented in the next page. 

\setcounter{equation}{0}
\setcounter{table}{0}
\renewcommand{\theequation}{Q.\arabic{equation}}
\section{Impact of timeslot requirements and tradeoff weights} \label{app:additional-scenario-viii}

\begin{figure}[!h]
	\centering
	\includegraphics[]{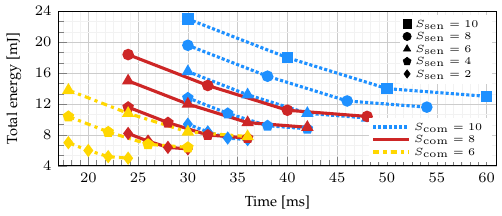}
	\caption{Impact of timeslot requirements and tradeoff weights on resource economy (\emph{Scenario VIII}). \emph{Weight selection serves as an effective mechanism to navigate the tradeoff between energy and time consumption. Also, the chosen weights implicitly influence whether joint or individual servicing of users and targets is favored. Finally, both energy and time consumption scale with the total timeslot requirements.}}
	\label{fig:results-scenario-7}
\end{figure}

In this scenario, we investigate how varying timeslot requirements influence both energy and time consumption. Unlike previous scenarios, which prioritized minimizing time consumption through specific weight selections ($ \eta_1 $ and $ \eta_2 $), we now explore a broader range of weight configurations to illustrate the tradeoff between energy and time across diverse priority settings. We consider $ U = 3 $, $ T = 3 $, $ S_\mathrm{com} = \left\lbrace 6, 8, 10 \right\rbrace $, and $ S_\mathrm{sen} = \left\lbrace 2, 4, 6, 8, 10 \right\rbrace $. The detailed parameter setting is shown in \cref{tab:simulation-settings-s8}.

\begin{table*}[b!]
	\setlength\tabcolsep{3.8pt} 
	\renewcommand{\arraystretch}{1.8}
	\tiny
	\centering
	\caption{Simulation parameters.}
	\begin{tabular}{|c|| c c c c c | c c | c c | c c | c c c c | c c |}
		\hline
		Scenario
		& $ U $
		& $ T $
		& $ S_\mathrm{com} $ 
		& $ S_\mathrm{sen} $
		& $ S $
		& $ \beta_u $ 
		& $ \theta_t \mid \widebar{\theta}_t $ 
		& $ r_u $ 
		& $ \psi_t \mid \widebar{\psi}_t $
		& $ \Upsilon_\mathrm{snr} $
		& $ \Lambda_\mathrm{sinr} $
		& $ \frac{\epsilon_{\mathrm{CSI}}}{\sigma_\mathrm{com}} $ 
		& $ \epsilon_{\mathrm{AOD}} $
		& $ \frac{\epsilon_\mathrm{RC}}{\widebar{\psi}} $
		& $ \frac{\epsilon_\mathrm{RSI}}{\left\| \widebar{\mathbf{R}} \right\|_\mathrm{F} } $
		& $ \upsilon $
		& $ \widebar{d}_c $
		\\
		\hline
		VIII  
		& \multicolumn{1}{>{\columncolor{DodgerBlue1!20}}c}{\tabularCenterstack{c}{$ 3 $}} 
		& \multicolumn{1}{>{\columncolor{DodgerBlue1!20}}c}{\tabularCenterstack{c}{$ 3 $}} 
		& \multicolumn{1}{>{\columncolor{DodgerBlue1!20}}c}{\tabularCenterstack{c}{$ \left[ 6, 10 \right] $}} 
		& \multicolumn{1}{>{\columncolor{DodgerBlue1!20}}c}{\tabularCenterstack{c}{$ \left[ 2, 10 \right] $}} 
		&  \multicolumn{1}{>{\columncolor{DodgerBlue1!20}}c|}{\tabularCenterstack{c}{$ 60 $}} 
		& \multicolumn{1}{>{\columncolor{DodgerBlue1!20}}c}{\tabularCenterstack{c}{$ \left[ 50, 85 \right]  $}} 
		& \multicolumn{1}{>{\columncolor{DodgerBlue1!20}}c|}{\tabularCenterstack{c}{$ \left[ 75, 100 \right] $}} 
		& \multicolumn{1}{>{\columncolor{DodgerBlue1!20}}c}{\tabularCenterstack{c}{$ \left[ 40, 60 \right] $}} 
		& \multicolumn{1}{>{\columncolor{DodgerBlue1!20}}c|}{\tabularCenterstack{c}{$ \left[ 8, 12 \right] 10^{-4} $}} 
		& $ 50 $  
		& $ 2 $  
		& $ 0 $  
		& $ 0 $ 
		& $ 0 $ 
		& $ 0 $ 
		& $ 0 $ 
		& $ - $ 
		\\
		\hline 
	\end{tabular}
	\label{tab:simulation-settings-s8}
	\vspace{-2mm}
\end{table*}

\begin{figure} [!t]
\scalebox{0.96}{\parbox{.25\columnwidth}{%
\begin{align*} 
	\mathcal{P}' \left( \boldsymbol{\Omega}' \right): & \min_{
			\substack{ \boldsymbol{\Omega}' }
	} 
	& & f'\left( \boldsymbol{\Omega}' \right)
	\\
	& ~ \mathrm{s.t.} & \mathrm{C}_{1}: ~ & \kappa_s = \left\lbrace 0, 1 \right\rbrace, \forall s, 
	\\
	& & \mathrm{C}_{2}: ~ & \zeta_s = \left\lbrace 0, 1 \right\rbrace, \forall s,   
	\\
	& & \mathrm{C}_{4}: ~ & \textstyle \sum_{ s \in \mathcal{S} } \gamma_s \leq S, 
	\\
	& & \mathrm{C}_{5}: ~ & \mu_{u,s} \in \left\lbrace 0, 1 \right\rbrace, \forall u, s,   
	\\
	& & \mathrm{C}_{6}: ~ & \textstyle \kappa_s = \sum_{u \in \mathcal{U}} \mu_{u,s}, \forall s,  
	\\
	& & \mathrm{C}_{7}: ~ & \tau_{t,s} \in \left\lbrace 0, 1 \right\rbrace, \forall t, s,   
	\\
	& & \mathrm{C}_{8}: ~ & \textstyle \zeta_s = \sum_{t \in \mathcal{T}} \tau_{t,s}, \forall s,   
	\\
	& & \mathrm{C}_{9}: ~ & \chi_{b,s} \in \left\lbrace 0, 1 \right\rbrace, \forall b, s, 
	\\
	& & \mathrm{C}_{10}: ~ & \textstyle \sum_{b \in \mathcal{L}_\mathrm{tx}} \chi_{b,s} = \gamma_s, \forall s,  
	\\
	& & \mathrm{C}_{12}: ~ & \rho_{c,s} \in \left\lbrace 0, 1 \right\rbrace, \forall c, s,
	\\
	& & \mathrm{C}_{13}: ~ & \textstyle \sum_{c \in \mathcal{L}_\mathrm{rx}} \rho_{c,s} = \zeta_s, \forall s, 
	\\
	& & \mathrm{C}_{17}: ~ & \textstyle \sum_{s \in \mathcal{S}} \mu_{u,s} = S_{\mathrm{com},u}, \forall u,  
	\\
	& & \mathrm{C}_{22}: ~ & \textstyle \sum_{s \in \mathcal{S}} \tau_{t,s} = S_{\mathrm{sent},t}, \forall t,  
	\\ 
	& & \mathrm{D}_{1}: ~ & \gamma_s \leq \kappa_s + \zeta_s, \forall s, 
	\\ 
	& & \mathrm{D}_{2}: ~ & \gamma_s \geq \kappa_s, \forall s,   
	\\ 
	& & \mathrm{D}_{3}: ~ & \gamma_s \geq \zeta_s, \forall s,
	\\ 
	& & \mathrm{D}_{4}: ~ & \gamma_s \leq 1, \forall s, 
	\\ 
	& & \mathrm{E}_{1}: ~ & \left[ 
		 	 			   					\begin{matrix}
		 	 			   					   	\mathbf{E}_{u,s}^\mathrm{a}
		 	 			   					   	& 
		 	 			   					   	\mathbf{e}_{u,s}^\mathrm{b}
		 	 			   					   	\\
		 	 			   					    \mathbf{e}_{u,s}^\mathrm{c}
		 	 			   					    & 
		 	 			   					    e_{u,s}^\mathrm{d}
		 	 			   				  	\end{matrix} 
		  			   					\right] 
		  					   		\succcurlyeq \mathbf{0}, \forall u, s,
	\\ 
	& & \mathrm{E}_{2}: ~ & \alpha_u \geq 0, \forall u,
	\\ 
	& & \mathrm{F}_{1}: ~ & \chi_{b,s} \geq \pi_{b,c,t,s}, \forall b, c, t, s, 
	\\ 
	& & \mathrm{F}_{2}: ~ & \rho_{c,s} \geq \pi_{b,c,t,s}, \forall b, c, t, s,
	\\ 
	& & \mathrm{F}_{3}: ~ & \tau_{t,s} \geq \pi_{b,c,t,s}, \forall b, c, t, s,
	\\ 
	& & \mathrm{F}_{4}: ~ & 2 + \pi_{b,c,t,s} \geq \chi_{b,s} + \rho_{c,s} + \tau_{t,s}, \forall b, c, t, s,
	\\ 
	& & \mathrm{F}_{5}: ~ & \pi_{b,c,t,s} \in \left[ 0, 1 \right], \forall b, c, t, s,
	\\ 
	& & \mathrm{H}_{1}: ~ & z_{t,s} \geq 0, \forall t, s,
	\\ 
	& & \mathrm{H}_{4}: ~ & z_{t,s} \leq \tau_{t,s} \cdot \ddot{M}_{t}^\mathrm{UB}, \forall t, s,
	\\ 
	& & \mathrm{I}_{1}: ~ & \left[ 
					   			 			   					\begin{matrix}
					   			 			   					   	i_{t,s}^\mathrm{a}
					   			 			   					   	& 
					   			 			   					   	i_{t,s}^\mathrm{b}
					   			 			   					   	\\
					   			 			   					    i_{t,s}^\mathrm{c}
					   			 			   					    & 
					   			 			   					    i_{t,s}^\mathrm{d}
					   			 			   				  	\end{matrix} 
					   		 			   					\right] 
					   		 					   		\succcurlyeq \mathbf{0}, \forall t, s,
	\\ 
	& & \mathrm{I}_{2}: ~ & \xi_{t} \geq 0, \forall t, 
	\\ 
	& & \mathrm{J}_{1}: ~ & \delta_{c,t,s} \leq \rho_{c,s}, \forall c, t, s,
	\\ 
	& & \mathrm{J}_{2}: ~ & \delta_{c,t,s} \leq \tau_{t,s}, \forall c, t, s,
	\\ 
	& & \mathrm{J}_{3}: ~ & \delta_{c,t,s} \geq \rho_{c,s} + \tau_{t,s} - 1, \forall c, t, s,
	\\ 
	& & \mathrm{J}_{4}: ~ & \delta_{c,t,s} \geq 0, \forall c, t, s,
	\\ 
	& & \mathrm{J}_{5}: ~ & \left[ 
				   			 			   					\begin{matrix}
				   			 			   					   	\mathbf{J}_{t,s}^\mathrm{a}
				   			 			   					   	& 
				   			 			   					   	\mathbf{j}_{t,s}^\mathrm{b}
				   			 			   					   	\\
				   			 			   					    \mathbf{j}_{t,s}^\mathrm{c}
				   			 			   					    & 
				   			 			   					    j_{t,s}^\mathrm{d}
				   			 			   				  	\end{matrix} 
				   		 			   					\right] 
				   		 					   		\succcurlyeq \mathbf{0}, \forall t, s,
	\\ 
	& & \mathrm{J}_{6}: ~ &  \iota \geq 0. 
	\\ 
	& & \mathrm{K}_{1}: ~ &  \left| \mathbf{h}_u^\mathrm{H} \mathbf{w}_s \right|^2 \leq \dddot{M}_{u}^\mathrm{UB}, \forall u \in \mathcal{U}, s \in \mathcal{S},
	\\ 
	& & \mathrm{K}_{2}: ~ &  \widebar{\ell}_{s+1} (\boldsymbol{\Omega}') \geq \widebar{\ell}_{s} (\boldsymbol{\Omega}'), \forall s \in \mathcal{S} \setminus \{ S \}.
\end{align*}
	}}
\end{figure}

Each curve in \cref{fig:results-scenario-7} is generated by varying the weights $ \eta_1 $ and $ \eta_2 $, thereby adjusting the relative importance assigned to energy and time consumption. For any given curve, the leftmost points represent configurations with a strong emphasis on minimizing time consumption and little regard for energy, while the rightmost points prioritize energy minimization at the expense of time. Notably, the leftmost allocations typically favor joint servicing of users and targets, resulting in fewer active timeslots. In contrast, the rightmost allocations tend to assign separate timeslots to users and targets to minimize energy consumption, particularly in cases of poor alignment.


\end{appendices}

\end{document}